\documentclass[11pt]{article}
\usepackage[margin=1in]{geometry}
\usepackage{amsmath,amssymb,amsthm,mathtools}
\usepackage{dsfont}
\usepackage[vlined,linesnumbered,ruled]{algorithm2e}
\SetKw{KwReturn}{return}
\usepackage{booktabs}
\usepackage{tikz}
\usepackage{array}

\usepackage{soul}
\usepackage{xcolor}
\usepackage[most]{tcolorbox}
\definecolor{chatlightgreen}{rgb}{0.86,1.00,0.86}
\definecolor{chatpurple}{rgb}{0.45,0.00,0.65}

\newtheorem{theorem}{Theorem}[section]
\newtheorem{lemma}[theorem]{Lemma}
\newtheorem{corollary}[theorem]{Corollary}
\newtheorem{proposition}[theorem]{Proposition}
\newtheorem{fact}[theorem]{Fact}
\newtheorem{observation}[theorem]{Observation}

\theoremstyle{definition}
\newtheorem{definition}[theorem]{Definition}
\theoremstyle{remark}
\newtheorem{remark}[theorem]{Remark}

\def\inline#1:{\par\vskip 7pt\noindent{\bf #1:}\hskip 10pt}
\def\midinline#1:{\par\noindent{\bf #1:}\hskip 10pt}
\def\dnsinline#1:{\par\vskip -7pt\noindent{\bf #1:}\hskip 10pt}
\def\ddnsinline#1:{\newline{\bf #1:}\hskip 10pt}
\def\largeinline#1:{\par\vskip 7pt\noindent{\large\bf #1:}\hskip 10pt}

\def\dnsparagraph{\vspace{-5pt}\paragraph}

\def\BalSeg{\mbox{\sf Bal\_Segments}}
\def\BalBlocks{\mbox{\sf Bal\_Blocks}}
\def\SymColor{\mbox{\sf Sym\_Color}}
\def\SplitSym{\mbox{\sf Split\_Sym}}
\def\LogCompact{\mbox{\sf LogCap\_Compact}}
\def\ShiftIn{\mbox{\sf Shift\_In}}
\def\Ruling{\mbox{\sf Ruling}}

\def\cA{\mathcal{A}}
\def\cC{\mathcal{C}}
\def\cE{\mathcal{E}}
\def\cG{\mathcal{G}}
\def\cI{\mathcal{I}}
\def\cL{\mathcal{L}}
\def\cP{\mathcal{P}}
\def\cS{\mathcal{S}}
\def\cT{\mathcal{T}}
\def\dBCC{d\mathcal{BCC}}
\def\const{C}

\def\CONGEST{\mbox{\tt CONGEST}}
\def\LOCAL{\mbox{\tt LOCAL}}

\def\clr{\phi}
\def\numclr{\chi}
\def\freq{{\tt f}}
\def\imb{{\tt g}}
\newcommand{\E}{\mathbb{E}}
\renewcommand{\Pr}{\mathbb{P}}
\newcommand{\dist}{\mathrm{dist}}
\def\Zthree{\mathbb{Z}_3}
\def\Rand{R}
\def\ID{\mbox{\tt ID}}
\def\cliqsize{\kappa}
\def\ksplit{\kappa_{\mathrm{split}}}
\def\Vmid{V_{\mathrm{mid}}}
\def\lmax{\ell_{\max}}
\def\mumax{\mu_{\max}}

\title{Distributed Near-Equitable Coloring \\
in the $\LOCAL$ Model
}

\author{Amit Nir\thanks{Weizmann Institute of Science. E-mail: {amit.nir,david.peleg@weizmann.ac.il}}
\and
David Peleg$^*$}
\date{\today}

\begin{document}
\maketitle

\begin{abstract}
For an $n$-vertex graph of maximum degree $\Delta$, an \emph{equitable} $(\Delta+1)$-coloring is a proper coloring all of whose color classes have size $\sigma = n/(\Delta+1)$ up to rounding. Known distributed algorithms for relaxations of this target rely on aggregation along a spanning structure, at a cost that scales with the diameter $D$. We study \emph{near-equitable} coloring in the $\LOCAL$ model and prove that, for deterministic algorithms, exact balance is inherently \emph{global}: already on the cycle $\cC_n$, exact balanced free $3$-coloring requires $\Omega(n)$ rounds. More generally, additive imbalance $\imb$ requires $\Omega(n/\imb)$ rounds, matching a deterministic $O((n/\imb)\log^* n)$-round algorithm up to the $\log^* n$ factor, and the
resulting balance problems realize a dense family of intermediate deterministic
$\LOCAL$ complexities $\tilde\Theta(n^\alpha)$, for every
$\alpha \in (0,1)$, on cycles. The exact lower bound is sharp in two further
senses: the identifier-universe threshold is exactly $N=n+1$, and the
$\log^* n$ factor separating the deterministic bounds cannot be removed by
computing the ruling-set anchors faster. 
The bound is then extended
to \emph{twisted fiber products}, an explicit family realizing, for every maximum
degree $\Delta \ge 3$, every diameter scale from the Moore-bound minimum
$\Theta(\log n/\log\Delta)$ to the connectivity-imposed maximum
$\Theta(n/\Delta)$. Exact equitable $(\Delta+1)$-coloring requires
$\Theta(D)$ rounds deterministically on this family, so exactness costs diameter time already on bounded-degree graphs of logarithmic diameter.

In contrast, coarse balance admits \emph{local} algorithms. On cycles, 
constant multiplicative equity
costs $\Theta(\log^* n)$ rounds. On general graphs with large color classes,
all class sizes can be kept within $(1\pm\eta)$ times the average class size in time independent of the diameter. This can be achieved with palette exactly $\Delta+1$ for moderate
degrees ($\Delta \le 2^{\sqrt{\log n}/C}$), and with palette
$(1+\eta)(\Delta+1)$ for every $\Delta \le n^{1-o(1)}$, the latter in
$O(\log n\,(\log\log n)^2)$ rounds.
\end{abstract}


\section{Introduction}
\label{sec:intro}

\subsection{Background}

Distributed graph coloring is a canonical symmetry-breaking task of distributed computing: a proper coloring partitions the network into independent sets that may act concurrently, and the number of colors governs the length of the induced schedule. For scheduling and load-balancing applications one often needs more: the color classes should have (nearly) equal sizes, so that no time slot is overloaded and none is wasted. For a graph $G$ of maximum degree $\Delta$, a proper coloring with palette $\Delta+1$ whose \emph{color frequencies} (or class sizes) all equal $\sigma = n/(\Delta+1)$ up to rounding is \emph{equitable}. The Hajnal--Szemer\'edi theorem~\cite{HS70} guarantees its existence for every graph, and it can be computed using the $O(\Delta n^2)$-time sequential construction by Kierstead, Kostochka, Mydlarz and Szemer\'edi~\cite{KKMS10}.

A recent line of work~\cite{NP25} developed \emph{near-equitable} coloring algorithms, relaxing the palette to $(1+\rho)(\Delta+1)$ colors and/or the class sizes to a range $[\freq_{min}, \freq_{max}]$ around $\sigma$, in the sequential, $\CONGEST$, and Congested Clique models. All of these algorithms are driven by global coordination primitives (\emph{color accounting} and \emph{quota assignment}), which aggregate per-color statistics across the entire graph; in $\CONGEST$, each such aggregation costs $\Theta(D + \Delta)$ rounds, where $D$ is the diameter, and this term appears in every $\CONGEST$ bound of~\cite{NP25}. Table~\ref{tab:compare} summarizes some of the relevant prior bounds.

This paper studies the problem in the $\LOCAL$ model~\cite{Linial92,Peleg00}, where messages are unbounded and the only resource is locality. Proper $(\Delta+1)$-coloring is well known to be local, requiring only $\Theta(\log^* n)$ rounds on bounded-degree graphs~\cite{CV86,Linial92,Naor91} and $\mathrm{poly}(\log\log n)$ randomized rounds in general~\cite{CLP18,RG20}. The color frequency (or class size) constraint, by contrast, couples all $n$ output decisions through a single cardinality requirement, hence the required coloring is no longer a \emph{locally checkable labeling (LCL)}~\cite{NS95}. The natural question, raised in~\cite{NP25}, 
is whether this global constraint forces diameter-type lower bounds even on rings and trees, or whether locality, randomness and slack can substitute for global accounting. Our results show that both answers are correct, in two cleanly separated parameter regimes.

Hereafter we use the terms \emph{local} and \emph{global} in this operational sense. A guarantee is local when its round complexity is bounded independently of the diameter (or, specifically on cycles, is $o(n)$ and typically polylogarithmic or $\log^* n$). It is global when every deterministic implementation needs $\Omega(D)$ rounds (on cycles -- $\Omega(n)$ rounds). Thus the terminology refers to the amount of network-scale coordination forced by the balance requirement, not to local checkability of the output.

\subsection{Related work}
\label{sec:related}

Equitable coloring originates in~\cite{HS70}; see~\cite{KKMS10} for the algorithmic state of the art. The distributed near-equitable suite of~\cite{NP25} is the closest prior work. 
Table~\ref{tab:compare} restates 
some of its results 
for comparison; 
none of our current results depends on~\cite{NP25}. We are unaware of other treatments of globally frequency-constrained coloring in the $\LOCAL$ model. Distributed proper coloring is a vast subject~\cite{CV86,Linial92,Naor91,BEPS16,HSS18,CLP18,RG20}; we use the modern toolbox (ruling sets, shattering, network decomposition, $(\mathrm{deg}{+}1)$-list coloring) as black boxes where possible. The LCL classification program~\cite{NS95,BHKLOS18,CKP19,BBOS18} provides the complexity-gap backdrop for Corollary~\ref{cor:continuum}; the intermediate complexities constructed in~\cite{BBOS18} are for locally \emph{checkable} problems on trees, whereas our candidate continuum concerns a natural non-checkable constraint on cycles. Distributed degree splitting~\cite{GS17,GHKMSU17} studies per-vertex (locally checkable up to slack) balance constraints, a fundamentally different regime from a single global cardinality constraint. Lower bounds for global aggregation-type tasks in bandwidth-limited models~\cite{SHKKNPPW12} rely on communication bottlenecks that do not exist in $\LOCAL$, which is precisely why a different, locality-based obstruction is needed here.

\subsection{Our results}
\label{sec:results}

A coloring $\clr$ with palette size $\numclr$ is \emph{$\imb$-additively balanced} if every color frequency (class size) is in $[n/\numclr - \imb,\, n/\numclr + \imb]$. A coloring is \emph{proper} if adjacent vertices receive distinct colors; unless otherwise specified, any coloring problem mentioned below requires properness. Throughout, $\sigma = n/(\Delta+1)$, and a proper coloring is \emph{$\eta$-multiplicatively equitable} if $\numclr = \Delta + 1$ and every class size is in $[(1-\eta)\sigma, (1+\eta)\sigma]$. Some lower bounds below apply even when this properness condition is dropped; we call such an output a \emph{free coloring}. Precise model definitions, including the treatment of identifiers and of the knowledge of $n$ and $\Delta$, appear in Section~\ref{sec:prelim}.

\clearpage
The results are organized around one message:
\begin{center}
\emph{Exact and near-exact balance are global, deterministically; \\
coarse multiplicative balance is local.}
\end{center}
Concretely, we prove four theorem-level results. On rings, Section~\ref{sec:det-lb} gives deterministic lower bounds for exact and near-exact balance (Theorem~\ref{thm:intro-ring-lb}), while Section~\ref{sec:ub} gives the cycle upper bounds, including the matching deterministic tradeoff up to a $\log^* n$ factor and the faster randomized algorithm (Theorem~\ref{thm:intro-cycle-ub}). A black-box reduction turns the exact endpoint into a deterministic $\Theta(D)$ bound for exact equitable $(\Delta+1)$-coloring on an explicit graph family realizing every maximum degree $\Delta \ge 3$ and every diameter scale permitted by the Moore bound (Theorem~\ref{thm:intro-diam}). 
(Note that in the $\LOCAL$ model, exact equitable coloring is solvable in $D$ rounds\footnote{Every vertex learns the entire graph, applies the algorithm of~\cite{KKMS10} locally, and outputs its color. All vertices compute the same coloring, so the output is consistent, proper, and exactly balanced.}.)
Finally, on general graphs with large color classes, $(1\pm\eta)$-equity is achievable in rounds independent of $D$ (Theorem~\ref{thm:intro-general}). All results are unconditional.

\subsubsection{Rings: 
lower and upper bounds}

The deterministic lower bound separates exact from coarsely approximate balance on rings by an exponential ($n$ vs.\ $\log^* n$) gap, and its stability version pins down the lower-bound curve between the two endpoints. 

Two key notions of the analysis are \emph{rigidity} and \emph{stability}.
Balancing algorithms employ counting on various graph segments (which in the context of rings are referred to as \emph{windows}). The first ingredient of the proof is a rigidity statement addressing \emph{exact} balances: 
if a deterministic local window rule has the same global count under every identifier assignment, then that count is forced to be an integer multiple of the cycle length. Stability concerns \emph{relaxed} balances. The stability version of the above statement allows the count to vary in an interval of radius $\imb$ and still forces this interval to meet such an integer multiple. These are global cardinality constraints on a local rule; they are unrelated to topological rigidity or probabilistic stability. 

We get the following results (proved by purely combinatorial arguments).

\begin{theorem}[Section~\ref{sec:det-lb}]

\label{thm:intro-ring-lb}
Let $N$ denote the size of the identifier universe.
\begin{itemize}
\item[(a)] 
Let $3 \mid n$, $1 \le T \le n/40 - 1$ and $N \ge 2n+2$. No deterministic $T$-round $\LOCAL$ algorithm on $\cC_n$ outputs, for \emph{every} assignment of distinct identifiers, a free $3$-coloring whose three color classes each have size exactly $n/3$. In particular, exact equitable $3$-coloring of rings has deterministic $\LOCAL$ complexity $\Theta(n) = \Theta(D)$.
\item[(b)] 
Let $\imb \ge 0$ be a real number with $\imb < n/3$, let $1 \le T \le \frac{n}{40(\imb+1)} - 1$ and $N \ge 3n+2$. No deterministic $T$-round $\LOCAL$ algorithm on $\cC_n$ outputs, for every assignment, a free $3$-coloring whose three color classes all have size $n/3 \pm \imb$.
\end{itemize}
\end{theorem}


The cycle upper bounds are proved in Section~\ref{sec:ub}. They give the matching deterministic tradeoff up to a $\log^* n$ factor, and show how much randomization can help. Throughout, ``with high probability'' (w.h.p.) means with probability at least $1 - n^{-\gamma}$ for an arbitrarily large constant $\gamma$ fixed in advance\footnote{All randomized algorithms in this paper are parameterized by the failure exponent $\gamma$, and the constants in their bounds may depend on it.}.
We get the following result.

\begin{theorem}[Section~\ref{sec:ub}]
\label{thm:intro-cycle-ub}
\mbox{}\par\nobreak
\begin{itemize}
\item[(a)] 
Let $L = \max\{1, \log^* n\}$. On $\cC_n$, a proper $3$-coloring with imbalance $O\big(\sqrt{(n/T)\,L\log n}\big)$ is computable w.h.p.\ in $O(T)$ rounds, for any $T \ge CL$; equivalently, for \emph{every} $\imb \ge 1$, imbalance $\imb$ is achievable w.h.p.\ in $O\big(\min\{n,\ (n/\imb^2)\,L\log n + L\}\big)$ rounds. 
\item[(b)] 
For $3 \le \imb < n/3$, there is a deterministic $O\big((n/\imb)\log^* n\big)$-round $\LOCAL$ algorithm for $\imb$-additively balanced proper $3$-coloring of $\cC_n$. Hence, by Theorem~\ref{thm:intro-ring-lb}(b), the deterministic $\LOCAL$ complexity of the problem is $\Theta(n/\imb)$, up to a $\log^* n$ factor, throughout this range.
\item[(c)] 
For every constant $\eta > 0$, an $\eta$-multiplicatively equitable $3$-coloring of $\cC_n$ is computable deterministically in $O(\log^* n)$ rounds. This is tight: $\Omega(\log^* n)$ rounds are required already for plain proper $3$-coloring, deterministically~\cite{Linial92} and randomized~\cite{Naor91}.
\end{itemize}
\end{theorem}

As a corollary of the tight tradeoff, for $\imb = n^\beta$ with \emph{any} fixed $\beta \in (0, 1)$ the balance problems realize a dense family of deterministic $\LOCAL$ complexities $\tilde\Theta(n^{1-\beta})$ on cycles (Corollary~\ref{cor:continuum}). This covers the full exponent range $(0,1)$, not only $(1/2, 1)$; in contrast, LCLs on cycles admit only the three classes $O(1)$, $\Theta(\log^* n)$ and $\Theta(n)$. Moreover, at every such $\beta$ there is a polynomial deterministic-vs-randomized separation, of factor $n^{\min\{\beta, 1-\beta\}}/\mathrm{poly}\log n$.

\subsubsection{Exactness costs diameter time, deterministically, at every graph scale}

On the cycle, $D = \Theta(n)$, so Theorem~\ref{thm:intro-ring-lb} alone cannot distinguish a \emph{diameter} barrier from a \emph{size} barrier, and it concerns only $\Delta = 2$. Our next result removes both limitations by a black-box reduction; for exact target frequencies and deterministic algorithms, it answers the question of whether the $\Theta(D+\Delta)$-per-aggregation cost of the global $\CONGEST$ primitives of~\cite{NP25} reflects a genuine barrier. We get the following result.

\begin{theorem}
[Sections~\ref{sec:mcc} and~\ref{sec:twisted}]

\label{thm:intro-diam}
For every fixed $\Delta \ge 3$ there is an explicit family $\dBCC$ of maximum-degree-$\Delta$ graphs (the \emph{de Bruijn clique cycles} of Section~\ref{sec:twisted}) realizing, up to constant factors and divisibility adjustments, every pair $(n, D)$ with
\[ C\,\frac{\log n}{\log\Delta} \;\le\; D \;\le\; \frac{n}{C\Delta}, \]
namely, every diameter scale from the Moore-bound minimum to the connectivity-imposed maximum, on every member of which exact equitable $(\Delta+1)$-coloring has deterministic $\LOCAL$ complexity $\Theta(D)$. The lower bound holds for free colorings, and its stability version gives $\Omega(D/(\imb+1))$ rounds for additive imbalance $\imb$ up to $\Theta(D)$. In particular, exactness costs $\Theta(D)$ time already on bounded-degree graphs of logarithmic diameter, where 
color frequencies are $\sigma = \Theta(n)$. 
\end{theorem}

The simplest member of the graph family $\dBCC$ ($r = 1$ component per layer) is the \emph{matched clique cycle} $B_{m,K}$, on which the reduction is presented first and where $\sigma = \Theta(D)$; a complete-tripartite variant (see Remark \ref{rem:fiber-compare}) covers $D = o(\sqrt n)$ with $\sigma = \Theta(D)$ as well, at degree $\Theta(n/D)$.

By contrast, on the same family, $(1\pm\eta)$-equity with the optimal palette is achievable in time independent of $D$ (Theorem~\ref{thm:intro-general}(a), for $\Delta \le 2^{\sqrt{\log n}/C}$): with unbounded bandwidth, the diameter-scale cost of global accounting is necessary for exact balance targets and avoidable under two-sided slack. This is the general-graph analogue of the $\Theta(\log^* n)$-vs-$\Theta(n)$ separation on rings. The stability theorem lifts as well: additive imbalance $\imb < 2\sigma/K$ still costs $\Omega(D/\imb)$ rounds on $B_{m,K}$, and the divisibility assumption $(K+2) \mid m$ can be dropped (Corollary~\ref{cor:diam-stab}, Remark~\ref{rem:diam-div}). Note that $\sigma = \Theta(D)$, so the bound holds although every class is as large as the lower bound itself.

\subsubsection{Coarse balance is local}

The general-graph algorithmic theorem shows that above the appropriate slack scale, balance requires no global coordination. We get the following result.

\begin{theorem}[Section~\ref{sec:concentration}]
\label{thm:intro-general}\mbox{}\par\nobreak
\begin{itemize}
\item[(a)] \emph{Exact palette.} There is a constant $C$ such that for every $n$-vertex graph $G$ with $\Delta \le 2^{\sqrt{\log n}/C}$ and every $\eta \ge 2^{-\sqrt{\log n}/C}$, it is possible to compute a proper coloring of $G$ with palette exactly $\Delta+1$ and every class size in $[(1-\eta)\sigma, (1+\eta)\sigma]$ w.h.p.\ in $O\big(\log(\Delta/\eta)\big) + O(\log^3\log n)$ rounds, independent of $D$. For $\Delta = O(1)$ the bound improves to $O(\log(1/\eta) + \log^* n)$.
\item[(b)] \emph{Palette slack.} For every fixed $\eta \in (0,1]$ and every graph with $\Delta \le n^{1-o(1)}$, a proper coloring with palette size $\chi \le (1+\eta)(\Delta+1)$ and all
color frequencies 
in $(1\pm\eta)\,n/\chi$ is computable w.h.p.\ in $O\big(\log n\,(\log\log n)^2\big)$ rounds, again independent of $D$, and with no polynomial dependence on $1/\eta$ in the general form $O\big((\log\log n + \log(2/\eta))^2\log n\big)$. (The precise parameter range appears in Corollary~\ref{cor:split-all}.)
\end{itemize}
\end{theorem}

The proof of (a) has two ingredients: a \emph{transfer lemma} showing that any \emph{palette-symmetric} partial-coloring procedure with influence radius $\tau$ colors, in expectation, exactly $(n - \E[\freq(\bot)])/(\Delta+1)$ vertices with each color (where $\freq(\bot)$ is the number of vertices it leaves uncolored), with all class sizes concentrated within $\pm\,b\sqrt{n\log n}$ for $b$ the maximal ball size at radius $\tau$; and a two-phase algorithm whose symmetric phase is kept short ($O(\log(\Delta/\eta))$ rounds suffice to color all but an $\eta\sigma/4$-sized residue, by an averaging argument we prove from scratch in Lemma~\ref{lemma:decay}), so that the asymmetric completion phase, which finishes the residue by $(\mathrm{deg}+1)$-list coloring~\cite{HKNT22}, is too small to disturb any class. The degree restriction is a limitation of this proof, which controls influence via worst-case ball growth; a growth-sensitive restatement (Theorem~\ref{thm:symcolor-growth}) keeps the palette exactly $\Delta+1$ for all $\Delta \le n^{1/4-o(1)}$ on graphs of mild ball growth, including the matched clique cycles of Theorem~\ref{thm:intro-diam}. For (b), the mechanism is a one-level random split of the vertices and of the palette into private sub-instances of polylogarithmic degree, inside which the same machinery runs with polylogarithmic influence balls; the palette slack is inherent to this route, since splitting conserves the palette-to-degree ratio (Remark~\ref{rem:split-barrier}).

Table~\ref{tab:compare} compares Theorem~\ref{thm:intro-general}(a) with the prior distributed bounds. In its regime, it achieves the optimal palette $\Delta+1$ (previously reached only with upper frequency threshold $O(\sigma\log\Delta)$), two-sided multiplicative control $(1\pm\eta)\sigma$ centered at the target frequency, and a running time with no $D$ or $\Delta$ term. The range $(1\pm\eta)\sigma$ lies strictly inside the prior two-sided ranges $[\sigma/2, 2k\sigma]$ and $[\sigma/2, O(\sigma\log\Delta)]$, and is incomparable with the $[\sigma/2, \sigma]$ row, whose upper cap is 
stronger but which uses up to $2(\Delta+1)$ colors. Conversely, our lower bounds show that some 
restriction 
on the target frequencies is necessary: in the small-slack (near-exact) regime, already on cycles, balance is genuinely global. We conjecture that the true boundary of the concentration regime is $\sigma = \Theta(\eta^{-2}\log n)$, the threshold below which even an idealized multinomial allocation is not $(1\pm\eta)$-balanced w.h.p.; this matches the condition $\sigma = \Omega(\log n)$ that algorithm $\LogCompact$ of~\cite{NP25} had to assume.

\begin{table}[t]
\centering
\scriptsize
\setlength{\tabcolsep}{2pt}
\resizebox{\textwidth}{!}{%
\begin{tabular}{@{}llll@{}}
\toprule
Guarantee (palette; frequencies) & Model; rounds & Type & Source \\
\midrule
$\Delta+1$; $[\lfloor\sigma\rfloor,\lceil\sigma\rceil]$ & seq.\ $O(\Delta n^2)$ & det. & \cite{KKMS10} \\
$\le 2(\Delta+1)$; $[\sigma/2, \sigma]$ & $\CONGEST$; $O(\log\Delta(D{+}\Delta{+}\log^5\log n))$ & rand. & \cite{NP25} \\
$\le (1{+}\frac1k)(\Delta{+}1)$; $[\sigma/2, 2\sigma]$ & $\CONGEST$; $O((\log\Delta{+}k)(D{+}\Delta{+}\log^5\log n))$ & rand. & \cite{NP25} \\
$\Delta+1$; $[\sigma/2, O(\sigma\log\Delta)]$, $\sigma=\Omega(\log n)$ & $\CONGEST$; $O((D{+}\Delta)\log n\log\Delta)$ & rand. & \cite{NP25} \\
\midrule
$\Delta+1$; $(1{\pm}\eta)\sigma$, $\Delta\le 2^{\sqrt{\log n}/C}$ & $\LOCAL$; $O(\log\frac{\Delta}{\eta}) + O(\log^3\log n)$ & rand. & Thm.~\ref{thm:intro-general}(a) \\
$\le(1{+}\eta)(\Delta{+}1)$; $(1{\pm}\eta)\frac{n}{\numclr}$, $\Delta \le n^{1-o(1)}$ & $\LOCAL$; $O(\log n(\log\log n)^2)$ & rand. & Thm.~\ref{thm:intro-general}(b) \\
$3$ (cycles); $n/3 \pm \imb$ & $\LOCAL$; $O((n/\imb^2)\log n\log^* n + \log^* n)$ & rand. & Thm.~\ref{thm:intro-cycle-ub}(a) \\
$3$ (cycles); $n/3 \pm \imb$ & $\LOCAL$; $O((n/\imb)\log^* n)$ & det. & Thm.~\ref{thm:intro-cycle-ub}(b) \\
$3$ (cycles); $n/3 \pm \imb$ needs $T = \Omega(\frac{n}{\imb+1})$ & $\LOCAL$ & det.\ LB & Thm.~\ref{thm:intro-ring-lb}(b) \\
$3$ (cycles); exact: impossible, $T \le \frac{n}{40}{-}1$ & $\LOCAL$ & det.\ LB & Thm.~\ref{thm:intro-ring-lb}(a) \\
$\Delta{+}1$ (twisted fiber prod.); exact: $\Theta(D)$, all $\Delta \ge 3$, all $D$ scales & $\LOCAL$ & det. & Thm.~\ref{thm:intro-diam} \\
\bottomrule
\end{tabular}
}
\caption{Near-equitable coloring: prior distributed bounds and this paper. All randomized guarantees hold w.h.p. In the 
last row (twisted-fiber-products), both $\Delta \ge 3$ and the diameter scale $D$ are arbitrary, and for fixed $\Delta$ the frequencies are $\Theta(n)$. 
}
\label{tab:compare}
\end{table}

We remark that we do not prove a randomized lower bound on cycles: we conjecture that the randomized upper bound of Theorem~\ref{thm:intro-cycle-ub}(a) is optimal, and Section~\ref{sec:open} states the conjecture together with the proved partial evidence.

\subsection{Discussion}
\label{sec:overview}

%
%
%

\paragraph{Relation to standard lower-bound techniques.}
Classical $\LOCAL$ lower bounds proceed by indistinguishability, namely, by exhibiting instances whose radius-$T$ views coincide at some vertex while correctness forces different outputs; they are often supported by Ramsey-type arguments~\cite{Linial92,Naor91}, by round elimination, or, for aggregation tasks in bandwidth-limited models, by communication bottlenecks~\cite{SHKKNPPW12}. None of these applies here as is: the balance constraint relates all $n$ outputs at once and admits exponentially many valid outputs on every input, so no two-instance comparison witnesses a violation locally; messages are unbounded, so there is no communication bottleneck; and the problem is not an LCL, so the automatic-speedup machinery~\cite{CKP19} does not apply. The rigidity argument works in the opposite, global-to-local direction: it characterizes the values of a global linear statistic that a local rule can fix surely on all inputs (the integer multiples of $n$), and the stability version characterizes the short intervals to which such a statistic can be confined.

\paragraph{Where the boundary lies}
Four points situate the results. 
\\
(a) 
The exact-vs-approximate distinction is not cosmetic: on rings, the complexity jumps from $\Theta(\log^* n)$ for constant multiplicative equity (Theorem~\ref{thm:intro-cycle-ub}(c)) to $\Theta(n)$ for exact balance (Theorem~\ref{thm:intro-ring-lb}(a)), and the stability theorem maps the entire curve between the two endpoints, $\tilde\Theta(n/\imb)$ at imbalance $\imb$. \\
(b) 
The obstruction to exactness is coordination, not information: messages are unbounded, every vertex knows $n$, $\Delta$ and the target frequencies, and the lower bounds hold for free colorings; what is impossible is coordinating, from local views, membership in classes of exactly prescribed sizes. 
\\
(c)
Two-sided slack removes the diameter from the picture: in the concentration regime the round complexity depends on neither $D$ nor, beyond the palette size, on global aggregation, so the diameter-scale cost of the accounting primitives of~\cite{NP25} is necessary precisely for exact target frequencies (Theorem~\ref{thm:intro-diam}) and avoidable under $(1\pm\eta)$ slack (Theorem~\ref{thm:intro-general}(a)). 
\\
(d)
A single global cardinality constraint changes the complexity landscape of distributed coloring: it produces a dense family of intermediate deterministic complexities on cycles (Corollary~\ref{cor:continuum}), escaping the three-class LCL gap structure, together with a polynomial gap between deterministic and randomized complexity; both phenomena are impossible for locally checkable labelings on cycles, and both are unconditional here.

\subsection{Organization}

The rest of the paper is organized as follows. Section~\ref{sec:prelim} defines the model, the balance notions and the tools. 
Section~\ref{sec:det-lb} proves the rigidity theorem (Theorem~\ref{thm:intro-ring-lb}(a), proved as Theorem~\ref{thm:det-exact}), the sharp identifier-universe threshold (Theorem~\ref{thm:universe}), its stability version (Theorem~\ref{thm:intro-ring-lb}(b), proved as Theorem~\ref{thm:det-stab}) and the complexity continuum (Corollary~\ref{cor:continuum}); Section~\ref{sec:diameter-lb} applies both, as black boxes, first to matched clique cycles and then, via the twisted fiber lift (Section~\ref{sec:twisted}), at every degree and every diameter scale (Theorem~\ref{thm:intro-diam}). 
Section~\ref{sec:ub} proves the cycle upper bounds (Theorem~\ref{thm:intro-cycle-ub}) together with the anchor barrier (Proposition~\ref{prop:ruling-lb}), and Section~\ref{sec:concentration} proves the general-graph algorithms (Theorem~\ref{thm:intro-general}), including the growth-sensitive exact-palette strengthening and the palette-splitting proof. Section~\ref{sec:open} discusses limitations and open problems. 

\section{Preliminaries}
\label{sec:prelim}

\subsection{Model}
\label{sec:model}

We use the standard $\LOCAL$ model~\cite{Linial92,Peleg00}.
The network is modeled as an undirected graph $G = (V,E)$, and the input graph is the network itself.
Computation proceeds in synchronous rounds, driven by a global clock. In each round, every vertex sends an arbitrary message to each neighbor, and local computation is unbounded.
Initially, every vertex knows only its own identity and incident edges, as well as the values of $n$ and $\Delta$. In each round, it may send all the information it has accumulated so far to its neighbors. We denote by $\Gamma(v)$ the set of neighbors of $v$, and by $\Gamma_T(v)$ the radius-$T$ ball around $v$, including $v$ itself and all inputs in the ball. Therefore, after $T$ rounds, the output of $v$ is a function of $\Gamma_T(v)$.
All algorithms are \emph{uniform} (every vertex runs the same procedure). Each vertex of a cycle distinguishes its two ports arbitrarily; no globally consistent orientation is assumed, and none of our algorithms requires one (orientations appearing in the proofs are fixed for convenience of analysis, not computed). Algorithms are required to be correct for every legal port numbering; our lower bound proofs are therefore free to restrict attention to a convenient one. The parameters $n$ and $\Delta$ are known to all vertices; this strengthens all our lower bounds, and our upper bounds use $n$ only through the target frequencies $\sigma$ and $n/3$.

\emph{Identifiers and randomness.} In the identifier model, vertices carry distinct identifiers from a universe $[N]$; our lower-bound statements make the required $N$ explicit. All algorithmic upper bounds assume the standard polynomial universe $N = n^{O(1)}$, i.e.\ $O(\log n)$-bit identifiers; with a larger universe, the deterministic $O(\log^* n)$ terms (Linial/Cole--Vishkin color reduction, the ruling set of Fact~\ref{fact:ruling}, the completion of Fact~\ref{fact:completion}) read $O(\log^* N)$. In the randomized setting every vertex additionally holds a private, independent, unbounded random string $\Rand_v$. All logarithms are base $2$ unless stated otherwise ($\ln$ denotes the natural logarithm); the base affects only the constants, which we make no attempt to optimize. 
All randomized algorithms are parameterized by the desired failure exponent $\gamma$, and the constants in their time bounds and parameter thresholds may depend on $\gamma$. All lower bounds concern deterministic algorithms, which must be correct on every assignment of identifiers.

\subsection{Balance notions}


A coloring $\clr: V \to \cP$, $|\cP| = \numclr$, maps each vertex to a color. It is \emph{proper} if adjacent vertices receive distinct colors. In problem statements, an unqualified coloring requirement is proper; a \emph{free coloring} may assign the same color to adjacent vertices. Denote the \emph{frequency} of color $c$ by $\freq(c) = |\clr^{-1}(c)|$. The coloring is \emph{$\imb$-additively balanced} if $|\freq(c) - n/\numclr| \le \imb$ for every $c \in \cP$. A proper coloring is \emph{$\eta$-multiplicatively equitable} if $\numclr = \Delta+1$ and $\freq(c) \in [(1-\eta)\sigma, (1+\eta)\sigma]$ for all $c$. \emph{Exact balance} means $\freq(c) \in \{\lfloor n/\numclr\rfloor, \lceil n/\numclr\rceil\}$ for all $c$; when $\numclr \mid n$ this forces $\freq(c) = n/\numclr$ exactly.

\paragraph{Divisibility convention.}
To avoid carrying floors, lower-bound statements assume $\numclr \mid n$ where they say so (e.g., $3 \mid n$ in Theorem~\ref{thm:intro-ring-lb}(a)); Remark~\ref{rem:diam-div} shows how the stability theorem removes such assumptions. Upper-bound statements hold for all $n$, with rounding absorbed into the $O(\cdot)$.

\subsection{Probabilistic tools}

We use the following standard concentration bounds.

\begin{theorem}[Hoeffding; McDiarmid]
\label{thm:hoeffding}
(a) If $Y_1, \dots, Y_M$ are independent with $\E[Y_j] = 0$ and $|Y_j| \le 1$, then $\Pr[|\sum_j Y_j| \ge t] \le 2e^{-t^2/(2M)}$.
(b) If $h$ is a function of independent variables $Z_1, \dots, Z_M$ and changing any one $Z_j$ changes $h$ by at most $b_j$, then $\Pr[|h - \E[h]| \ge t] \le 2\exp(-2t^2/\sum_j b_j^2)$.
\end{theorem}

A family of random variables indexed along a cycle is \emph{$1$-dependent} if every subfamily whose indices are pairwise at cyclic distance at least $2$ is mutually independent. Splitting the index set into the color classes of a proper $3$-coloring of the cycle (which exists for every cycle length, whereas a $2$-coloring does not exist on odd cycles) partitions a sum of $1$-dependent variables into at most three sums, each of mutually independent variables.

\subsection{Graphs and cycles 
}
\label{ssec:walks}

Throughout the paper, let $\cC_n$ denote the $n$-vertex cycle, written as $\cC_n=(V_n,E_n)$ with $V_n=\{v_0,\ldots,v_{n-1}\}$ and $E_n=\{(v_i,v_{i+1 \bmod n}) \mid i\in[0,n-1]\}$ in the fixed cyclic order. For every $i\in[0,n-1]$, define the difference $\delta_i=\clr(v_{i+1 \bmod n})-\clr(v_i)\bmod 3$. The following observation is easily verified.

\begin{observation}
\label{obs:walk}
A $3$-coloring $\clr: V_n \to \{0,1,2\}$ is proper if and only if every difference satisfies $\delta_i \in \{+1, -1\}$. In a segment of length
$\ell \equiv 0 \pmod 3$ that is colored periodically (i.e., all
differences are $+1$), the three color frequencies in the segment equal exactly $(\ell/3, \ell/3, \ell/3)$, regardless of its starting color.
\end{observation}


Throughout the cycle arguments, color arithmetic is in $\Zthree=\{0,1,2\}$. We say that a segment is colored \emph{periodically from $s$} if, along the locally specified direction, it receives the pattern $s,s+1,s+2,s,\dots$ modulo $3$.

We also use the standard degree--diameter terminology. The \emph{Moore bound}
says that an $n$-vertex graph of maximum degree $\Delta \ge 3$ and diameter
$D$ satisfies
$n \le 1+\Delta\sum_{h=0}^{D-1}(\Delta-1)^h$,
and hence $D=\Omega(\log n/\log\Delta)$; see the survey of Miller and
Sir\'an~\cite{MS05}. We call this lower limit on $D$ the
\emph{Moore-bound minimum}. At the other end, connected maximum-degree-$\Delta$
graphs with $n$ vertices can have diameter $\Theta(n/\Delta)$, and we refer to
this as the \emph{connectivity-imposed maximum}. In the product constructions
of Section~\ref{sec:diameter-lb}, a \emph{fiber} is the graph copied over each
vertex of a base cycle, and a \emph{twist} is the fixed permutation matching one
fiber copy to the next. The phrase \emph{de Bruijn} refers to the standard
base-$\cliqsize$ shift map on residues, following de Bruijn's classical
construction~\cite{deBruijn46}.


For a graph $G$ and an integer $k\ge 1$, the \emph{power graph} $G^k$ has vertex set $V(G)$ and an edge between two distinct vertices whose distance in $G$ is at most $k$. An \emph{independent set} is a set of pairwise non-adjacent vertices. A \emph{maximal independent set (MIS)} is an independent set to which no further vertex can be added while preserving independence. A \emph{distance-$w$ ruling set} in a graph $G$ is a maximal independent set of $G^w$; equivalently, its vertices are pairwise at distance at least $w+1$ in $G$, and every vertex of $G$ is within distance $w$ of the set. For a set $S\subseteq V(\cC_n)$, the \emph{consecutive gaps} are the cyclic distances, in the fixed cyclic order of $\cC_n$, between consecutive vertices of $S$. A graph is \emph{growth-bounded} if, for every radius $r$, every radius-$r$ ball contains only a bounded number (as a function of $r$) of pairwise non-adjacent vertices; below we use this only for $\cC_n^w$, where the bound is at most $2r+1$. We use the following fact.

\begin{fact}
\label{fact:ruling}
For every $w \ge 1$ and every $n \ge w+1$, there is a deterministic $O(w\log^* n)$-round $\LOCAL$ algorithm $\Ruling(w)$ computing a distance-$w$ ruling set $S \subseteq V(\cC_n)$ whose consecutive gaps all lie in $[w+1,\, 2w+1]$.
\end{fact}

\begin{proof}
Consider the power graph $\cC_n^w$.
A single communication round on $\cC_n^w$ can be simulated in $w$ rounds on $\cC_n$. Note that $\cC_n^w$ contains cliques of size $w+1$, so no coloring with a number of colors independent of $w$ exists; instead, observe that $\cC_n^w$ is growth-bounded.
By the deterministic maximal-independent-set algorithm of Schneider and Wattenhofer for growth-bounded graphs~\cite{SW08}, an MIS $S$ of $\cC_n^w$ can be computed in $O(\log^* n)$ communication rounds of $\cC_n^w$, hence in $O(w\log^* n)$ rounds of $\cC_n$. (The algorithm of~\cite{SW08} is stated for abstract bounded-independence graphs: it uses only the combinatorial bound on the number of pairwise non-adjacent vertices in each ball, with no geometric representation or degree bound, so it applies to $\cC_n^w$ without change.)

The set $S$ is a distance-$w$ ruling set of $\cC_n$, its vertices being pairwise at distance at least $w+1$ in $\cC_n$, and its consecutive gaps obey the stated bounds, as follows. If $n \le 2w+1$, then $\cC_n^w$ is complete, so $S$ is a single vertex, and its unique cyclic gap is $n$, which lies in $[w+1, 2w+1]$ by the hypothesis $n \ge w+1$. If $n \ge 2w+2$, then $S$ is not a single vertex (a vertex at cyclic distance at least $w+1$ from it exists and could be added), every gap is at least $w+1$ by independence, and no gap exceeds $2w+1$: in a gap of length at least $2w+2$ between consecutive elements $u, u'$ of $S$, the interior vertex at arc distance $w+1$ from $u$ is at distance at least $w+1$ from $u'$ as well, hence from all of $S$, and could be added, contradicting maximality. (Any standard construction of distance-$w$ ruling sets on cycles could replace this one; we invoke~\cite{SW08} only for brevity.)

The hypothesis $n \ge w+1$ cannot be dropped: for $n \le w$, every nonempty subset of $\cC_n$ has a cyclic gap of length at most $n \le w$, so no set with the stated gap guarantee exists. All invocations of the fact in this paper are well inside the admissible regime: the algorithms of Section~\ref{sec:ub} call $\Ruling(w)$ with $w = O(T/\log^* n)$ and $T \le n$, so $w + 1 \le n$ for all sufficiently large $n$.
\end{proof}

%
%
%

\section{Deterministic lower bound for exact or $g$-additive balance}
\label{sec:det-lb}

Section~\ref{sec:det-exact-LB} proves Theorem~\ref{thm:intro-ring-lb}(a) in the following stronger form, stated explicitly as Theorem~\ref{thm:det-exact} below. The theorem says, essentially, that the \emph{only} global counts that a local deterministic rule can guarantee exactly, on every input, are the two extreme ones, namely, $0$ and the maximum count $n$.
Section~\ref{sec:det-stab} then proves a stability version (Theorem~\ref{thm:det-stab}), claiming that the count cannot even be \emph{confined to a short interval} away from the multiples of $n$. This makes the deterministic upper bound of Section~\ref{sec:ub} tight up to a $\log^* n$ factor (Theorems~\ref{thm:intro-ring-lb}(b) and~\ref{thm:intro-cycle-ub}(b)). 

\subsection{The lower bound for exact balance}
\label{sec:det-exact-LB}

We start with an overview of the proof structure.
We say that $K$ is a \emph{sure count} of the local rule $F$ if the induced global count satisfies $f(\ID)=K$ for every identifier assignment $\ID$; equivalently, $f$ is \emph{surely constant} with value $K$.
First, Lemma~\ref{lemma:exchange} shows that, under the sure-count hypothesis, replacing the identifier at the center of any fixed context cannot change the sum of the $2T+1$ window values that see the center. 
Consequently, we get Lemma~\ref{lemma:boundary}, stating that the inner window sum of a linear arrangement of identifiers depends only on its first and last $2T$ identifiers. 
This enables showing (in Lemma~\ref{lemma:loops}) that all closed chains of windows of the same length from the same basepoint carry the same weight, and that this weight is additive in the length. 
Moreover, additivity makes the weight linear in the length, with a single \emph{integer} slope $\lambda$ independent of the basepoint, as shown in Lemma~\ref{lemma:linear}. 
This, in turn, leads to contradiction, since the $n$ windows of an actual cyclic assignment form a loop of length $n$, so the guaranteed count equals $\lambda\, n$ (Figure~\ref{fig:rigidity}); for the count $n/3$, no integer $\lambda$ exists.

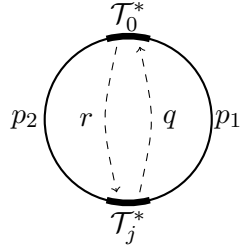
\begin{figure}[t]
\centering
\begin{tikzpicture}[scale=0.65]
\draw[thick] (0,0) circle (1.7);
\draw[line width=2.5pt] (75:1.7) arc (75:105:1.7);
\node at (90:2.1) {$\cT^*_0$};
\draw[line width=2.5pt] (255:1.7) arc (255:285:1.7);
\node at (270:2.2) {$\cT^*_j$};
\node at (0:2.05) {$p_1$};
\node at (180:2.1) {$p_2$};
\draw[dashed,->] (279:1.5) .. controls (0.55,0) .. (81:1.5);
\draw[dashed,->] (99:1.5) .. controls (-0.55,0) .. (261:1.5);
\node at (0.85,0) {$q$};
\node at (-0.85,0) {$r$};
\end{tikzpicture}
\caption{The splitting step in the proof of Theorem~\ref{thm:det-exact}. The $n$ windows of a cyclic assignment form a loop of length $n$ and weight $K$ through the boundary tuples $\cT^*_0$ and $\cT^*_j$; the length-$2T$ shift-in paths $q, r$ close the two halves $p_1, p_2$ into loops. Integer-slope linearity gives weights $\lambda(j+2T)$, $\lambda(n-j+2T)$ and $\lambda\cdot 4T$ for the loops $p_1q$, $p_2r$ and $qr$, and adding up yields $\lambda\, n = K$.}
\label{fig:rigidity}
\end{figure}

We require the following notation.
For an assignment $\ID$ of distinct identifiers to $\cC_n$ and a vertex $v_i$, define the \emph{radius-$T$ identifier window} $W(v_i,\ID,T)$, also written $W(i,\ID,T)$, as the ordered tuple $(\ID(i-T), \ID(i-T+1), \dots, \ID(i+T))$, with indices taken modulo $n$ in the fixed cyclic order. The main result of this section is the following.

\begin{theorem}
\label{thm:det-exact}
Let $1 \le T \le n/40 - 1$ and $N \ge 2n+2$. Let $F$ be any function mapping $(2T+1)$-tuples of distinct identifiers from $[N]$ to $\{0,1\}$, and for an assignment $\ID$ of distinct identifiers to $\cC_n$ let $f(\ID) = \sum_{v} F\big(W(v,\ID,T)\big)$. If $K$ is a sure count of $F$, namely, if $f(\ID) = K$ for every assignment $\ID$, then $K \in \{0, n\}$.
\end{theorem}

Theorem~\ref{thm:det-exact} yields the ring lower bound.

\begin{corollary}[Theorem~\ref{thm:intro-ring-lb}(a), restated]
\label{cor:det-exact}
Let $3 \mid n$ and $1 \le T \le n/40 - 1$. There is no deterministic $T$-round $\LOCAL$ algorithm on $\cC_n$ with identifiers from $[N]$, for $N \ge 2n+2$, that outputs on every assignment a free $3$-coloring with all frequencies exactly $n/3$. Consequently, the deterministic complexity on $\cC_n$ of exact equitable $3$-coloring, and already of exact balanced free $3$-coloring, is $\Theta(n)$.
\end{corollary}

\begin{proof}
For the lower bound, restrict attention to the instances whose port numbering is consistently oriented, with every vertex's port $1$ pointing to its clockwise neighbor; the algorithm must in particular be correct on these instances. On such instances, since the algorithm is uniform, the output at $v$ is a fixed function of the \emph{oriented} radius-$T$ identifier window $W(v,\ID,T)$. The indicator rule of any single color therefore satisfies the hypothesis of Theorem~\ref{thm:det-exact} with $K = n/3 \notin \{0,n\}$, a contradiction.
This lower bound applies to free colorings and hence also to proper colorings. 

\noindent
The matching $n$-round upper bound is obtained by collecting global information 
to all vertices.
\end{proof}

The proof of Theorem~\ref{thm:det-exact} is purely combinatorial. Let us first set up the necessary machinery.
Throughout, a \emph{$k$-tuple} $\cT=(I_1,\ldots,I_k)$ is a sequence of $k$ distinct identifiers $I_i$ from $[N]$; for a $(2T+1)$-tuple $\hat{\cT}$ we write $\mathrm{head}(\hat{\cT})$ and $\mathrm{tail}(\hat{\cT})$ for its first and last $2T$ entries. A \emph{step} $\cT \to \cT'$ between $2T$-tuples $\cT$ and $\cT'$ is a $(2T+1)$-tuple $\hat{\cT}$ with $\mathrm{head}(\hat{\cT}) = \cT$ and $\mathrm{tail}(\hat{\cT}) = \cT'$ (so if $\cT=(I_1,\ldots,I_{2T})$ then $\cT' = (I_2, \dots, I_{2T}, z)$ for some identifier $z$ not in $\cT$).
A \emph{chain} is a sequence of consecutive steps; a \emph{loop $\cL$ of length $\mu$} from $\cT_0$ is a chain $\cT_0 \to \cT_1 \to \dots \to \cT_\mu = \cT_0$, and its \emph{weight} is $\omega(\cL) = \sum_{k=0}^{\mu-1} F(\hat{\cT}_k)$, the sum of $F$ over its steps.
Note that the $n$ windows of any identifier assignment $\ID$ for $\cC_n$ form a loop of length $n$ and weight $f(\ID)$, and that for every $2T$-tuple $\cT$ and every $\mu$ with $2T + 1 \le \mu \le n$, loops of length $\mu$ from $\cT$ exist (arrange $\cT$ followed in a cycle by $\mu - 2T$ fresh identifiers; those are available since $\mu \le n < N$, and only lengths $\le n$ are used below).

An \emph{arrangement} is a tuple $a = (a_1, \dots, a_\ell)$ with $4T + 2 \le \ell \le n$; its \emph{inner sum} is
\[P(a) \;=\; \sum_{i = T+1}^{\ell - T} F\big(a_{i-T}, \dots, a_{i+T}\big).\]
For arrangements, $\mathrm{head}(a)$ and $\mathrm{tail}(a)$ denote the first and last $2T$ entries.

Consider some $K \notin \{0,n\}$, and assume, towards contradiction, that $f(\ID)=K$ on all assignments $\ID$. The following four lemmas are proved under this standing assumption.

A \emph{context} is a sequence $\cI = (I_{-2T}, \dots, I_{-1};\, I_1, \dots, I_{2T})$ of $4T$ distinct identifiers from $[N]$. For an identifier $x$ avoiding $\cI$, we denote
\[
\cI[x] \;=\; (I_{-2T}, \dots, I_{-1},\, x,\, I_1, \dots, I_{2T}) ,
\]
and we say that an assignment $\ID$ \emph{conforms to} $\cI[x]$ if it places the sequence $\cI[x]$ on the consecutive vertices $v_{-2T}, \dots, v_{2T}$ (its values outside this block are unconstrained). Define
\[
\Psi(\cI; x) \;=\; \sum_{i=-T}^{T} F\big(W(v_i,\ID,T)\big),
\]
where $\ID$ is any assignment conforming to $\cI[x]$.
Note that the value is well defined, namely, independent of the choice of the conforming assignment $\ID$: every window in the sum is contained in the block $v_{-2T}, \dots, v_{2T}$, on which $\ID$ reads $\cI[x]$. Our first lemma is the following.

\begin{lemma}
\label{lemma:exchange}
$\Psi(\cI; x) = \Psi(\cI; y)$ for every context $\cI$ and all identifiers $x, y$ avoiding $\cI$.
\end{lemma}

\begin{proof}
Build an assignment $\ID$ of $\cC_n$ conforming to $\cI[x]$, filling the other $n - 4T - 1$ positions with distinct identifiers avoiding $\cI \cup \{x, y\}$ (which is possible since $N \ge 2n \ge n + 4T + 3$). Let $\ID'$ be $\ID$ with $x$ replaced by $y$; then $\ID'$ conforms to $\cI[y]$.
For an identifier assignment $\ID^*$ and a vertex $z$, set
\begin{align*}
f^+(\ID^*,z) &\;=\; \sum_{i:\, z\in \Gamma_T(v_i)} F\big(W(v_i,\ID^*,T)\big),\\
f^-(\ID^*,z) &\;=\; f(\ID^*)-f^+(\ID^*,z).
\end{align*}
By the contradiction hypothesis,
\begin{equation}
\label{eq: f+ f- equalities}
f^+(\ID,v_0)+f^-(\ID,v_0)
\;=\; f(\ID)
\;=\; f(\ID')
\;=\; f^+(\ID',v_0)+f^-(\ID',v_0).
\end{equation}
Every window that does not contain the center vertex $v_0$ has the same ordered identifier window under $\ID$ and $\ID'$, hence $f^-(\ID,v_0)=f^-(\ID',v_0)$. The windows containing $v_0$ are exactly $W(v_i,\ID,T)$ and $W(v_i,\ID',T)$ for $-T \le i \le T$, and therefore $f^+(\ID,v_0)=\Psi(\cI;x)$ and $f^+(\ID',v_0)=\Psi(\cI;y)$. 
Combining with Eq. \eqref{eq: f+ f- equalities}
gives $\Psi(\cI;x)=\Psi(\cI;y)$.
\end{proof}

The 
lemma implies that inner sums are determined by their boundaries, as shown next.

\begin{lemma}
\label{lemma:boundary}
For every $\ell \in [4T+2, n]$ there is a function $H_\ell$ such that $P(a) = H_\ell(\mathrm{head}(a), \mathrm{tail}(a))$ for every arrangement $a$ of length $\ell$.

\end{lemma}

\begin{proof}
First, changing a single \emph{interior} entry $a_i$ ($2T+1 \le i \le \ell-2T$) to any fresh value $z \notin a$ leaves $P(a)$ unchanged: the affected terms of $P$ are exactly the windows centered at $v \in [i-T, i+T] \subseteq [T+1, \ell-T]$, and their sum is $\Psi(\cI;\, a_i) = \Psi(\cI;\, z)$ by Lemma~\ref{lemma:exchange}, where $\cI$ is the context formed by the $2T$ entries of $a$ on each side of position $i$.

Now let $a, a'$ be two arrangements of length $\ell$ with equal heads and tails; we transform $a$ into $a'$ by single interior replacements, each preserving $P$ by the previous paragraph. Maintain a current arrangement $b$, initially $a$, always agreeing with $a'$ on head and tail. While $b \ne a'$, pick an interior position $i$ with $b_i \ne a'_i$ and do the following. If $a'_i$ does not occur in $b$, replace $b_i \leftarrow a'_i$. Otherwise $a'_i$ occurs in $b$ at some position $j \ne i$, and $j$ is interior (the head and tail of $b$ equal those of $a'$, which are disjoint from the interior value $a'_i$); first replace $b_j \leftarrow z$ for some identifier $z \notin b \cup a'$ (such $z$ exists because $|b \cup a'| \le 2\ell - 4T \le 2n < N$, using $N \ge 2n+2$), and then replace $b_i \leftarrow a'_i$. Every replacement keeps all entries distinct and is an interior exchange, hence preserves $P$; and each iteration strictly increases the number of positions at which $b$ agrees with $a'$ (position $i$ becomes correct, and position $j$, which disagreed before, still disagrees). The process therefore terminates with $b = a'$, proving $P(a) = P(a')$. The lemma follows.
\end{proof}

Telescoping the boundary functions yields the following.

\begin{lemma}
\label{lemma:loops}
Fix a $2T$-tuple $A$. 
\\ (a)
For every step $\cT \to \cT'$ via window $\hat{\cT}$ with entries disjoint from $A$, and every $\ell$ with $4T+2 \le \ell \le n-1$,
\[ H_{\ell+1}(A, \cT') \;=\; H_\ell(A, \cT) \;+\; F(\hat{\cT}). \]
(b)
For every loop $\cL$ of length $\mu \le n - 4T - 2$ from basepoint $\cT_0$, with all window entries disjoint from $A$,
\[ \omega(\cL) \;=\; H_{4T+2+\mu}(A, \cT_0) - H_{4T+2}(A, \cT_0). \]
(c)
In particular all loops of the same length $\mu$ from the same basepoint $\cT_0$ have the same weight $\omega$, denoted $\Lambda_{\cT_0}(\mu)$, 
which is an integer in $[0, \mu]$.
\par\smallskip\noindent (d)
$\Lambda_{\cT_0}$ is additive, namely, $\Lambda_{\cT_0}(\mu_1 + \mu_2) = \Lambda_{\cT_0}(\mu_1) + \Lambda_{\cT_0}(\mu_2)$, whenever $\mu_1 + \mu_2 \le n - 4T - 2$.
\end{lemma}

\begin{proof}
Let us first prove (a). Take any arrangement $a'$ of length $\ell+1$ with head $A$ and last $2T+1$ entries equal to $\hat{\cT}$, filling the interior with fresh identifiers avoiding $A \cup \hat{\cT}$. Here we use $N \ge 2n+2$ and $\ell + 1 \le n$. The length-$\ell$ prefix of $a'$ is an arrangement with head $A$ and tail $\cT$, and by definition of the inner sum,
\begin{equation}
\label{eq:P+F}
P(a') = P(\mbox{prefix}) + F(\hat{\cT}).
\end{equation}
By Lemma~\ref{lemma:boundary},
\[
P(\mbox{prefix}) = H_\ell(A,\cT)
\qquad\mbox{and}\qquad
P(a') = H_{\ell+1}(A,\cT').
\]
Substituting these two identities in Eq. \eqref{eq:P+F} gives the recursion in (a). The heads and tails are disjoint as sets since $\ell \ge 4T+2$, and $A$ is disjoint from $\hat{\cT} \supseteq \cT \cup \cT'$ by assumption.

Turning to (b), let $\cL = (\cT_0 \to \dots \to \cT_\mu = \cT_0)$. Apply the recursion of (a) along the loop starting from level $\ell^* = 4T+2$:
\[
H_{\ell^*+k+1}(A,\cT_{k+1})
  = H_{\ell^*+k}(A,\cT_k) + F(\hat{\cT}_k),
\qquad 0 \le k \le \mu-1.
\]
Each application uses its own freshly built arrangement, and the levels stay at most $\ell^*+\mu \le n$. Summing these equalities over 
all $k$
telescopes the $H$-terms, giving the loop-weight formula in (b).

For (c), note that the weight $\omega(\cL)$ is a sum of $\mu$ values of $F \in \{0,1\}$, hence an integer in $[0,\mu]$. It depends only on $(\mu,\cT_0)$: given two loops $\cL,\cL'$ of the same length $\mu$ from the same basepoint $\cT_0$, they together use at most $2(\mu+2T) \le 2n-4T-4 \le N-2T$ distinct identifiers, using $\mu \le n-4T-2$ and $N \ge 2n+2$. Hence some $2T$-tuple $A$ is disjoint from the identifiers of both loops. Applying the formula of (b) with this common $A$ gives (c), as
\[
\omega(\cL) = H_{4T+2+\mu}(A,\cT_0)-H_{4T+2}(A,\cT_0) = \omega(\cL').
\]
To prove (d), let $\cL_1$ and $\cL_2$ be loops from $\cT_0$ of lengths $\mu_1$ and $\mu_2$, respectively, where $\mu_1+\mu_2 \le n-4T-2$. Their concatenation $\cL_1\cL_2$ is a loop from $\cT_0$ of length $\mu_1+\mu_2$, and its weight is the sum of the two weights:
\[
\omega(\cL_1\cL_2)=\omega(\cL_1)+\omega(\cL_2).
\]
By (c), these weights depend only on the corresponding lengths. This yields (d), as
\[
\Lambda_{\cT_0}(\mu_1+\mu_2)
  = \omega(\cL_1\cL_2)
  = \omega(\cL_1)+\omega(\cL_2)
  = \Lambda_{\cT_0}(\mu_1)+\Lambda_{\cT_0}(\mu_2).
\qedhere\]
\end{proof}

We shall use the following path operation several times. Let $\cT=(I_1,\ldots,I_{2T})$ and $\cT'=(J_1,\ldots,J_{2T})$ be disjoint $2T$-tuples. The path $\ShiftIn(\cT,\cT')$ is the length-$2T$ chain
\[
\cS_0,\cS_1,\ldots,\cS_{2T},
\qquad
\cS_h=(I_{h+1},\ldots,I_{2T},J_1,\ldots,J_h),
\]
where empty segments are suppressed. Thus $\cS_0=\cT$ and $\cS_{2T}=\cT'$, and the $h$-th step appends $J_{h+1}$ to $\cS_h$. Since the two endpoint tuples are disjoint, all intermediate windows have distinct identifiers.

\paragraph{A range ledger.}
Before proving that the loop weights are linear, we collect the numerical facts that the next two proofs consume; all follow from $T \ge 1$ and $n \ge 40(T+1)$, and this is also where the hypothesis $N \ge 2n+2$ is spent. Denote $M = 4T+2$ for the reference length used below.
\begin{itemize}
\item[(R1)] Loops of a given length from a given basepoint exist for every length in $[2T+1,\, n]$; the loop weight $\Lambda_{\cT}(\mu)$ is defined for $\mu \in [2T+1,\, n-4T-2]$ (Lemma~\ref{lemma:loops}(c)); and additivity (Lemma~\ref{lemma:loops}(d)) is legal whenever both summands are at least $2T+1$ and their sum is at most $n - 4T - 2$.
\item[(R2)] The particular lengths evaluated below all lie in the relevant ranges: $4T$, $M$ and $M+1$ are at least $2T+1$ (as $T \ge 1$); $3M = 12T+6 \le n - 8T - 6$ (as $n \ge 20T + 12$); and $4T + \mu \le n - 8T - 6$ whenever $\mu \le n - 12T - 6$. The conservative range $[2T+1,\, n-16T-10]$ in Lemma~\ref{lemma:linear} is contained in all of these.
\item[(R3)] The universe hypothesis $N \ge 2n+2$ enters only through fresh-identifier counts: the interior replacement of Lemma~\ref{lemma:boundary} needs one identifier outside a set of size at most $2n$; the common disjoint tuple in Lemma~\ref{lemma:loops}(c) needs $2(\mu+2T) + 2T \le N$; and the linking of basepoints in Lemma~\ref{lemma:linear} needs $6T \le N$.
\end{itemize}

We are now ready to show that the loop weights are linear.

\begin{lemma}
\label{lemma:linear}
Suppose $n \ge 40(T+1)$.
There is an integer $\lambda$ such that 
$\Lambda_{\cT}(\mu) = \lambda\,\mu$
for every $2T$-tuple $\cT$ and every $\mu \in [2T+1,\, n - 16T - 10]$.
\end{lemma}

\begin{proof}
Fix a basepoint $\cT$. Let $M = 4T+2$, and set
\[
\lambda(\cT)=\Lambda_{\cT}(M+1)-\Lambda_{\cT}(M).
\]
Note that $\lambda(\cT) \in \mathbb{Z}$, being a difference of integers. Additivity (Lemma~\ref{lemma:loops}(d)) gives, for all $\mu$ with $\mu + 1 + M \le n - 4T - 2$, i.e.\ $\mu \le n - 8T - 5$,
\[
\Lambda_{\cT}(\mu+1) + \Lambda_{\cT}(M)
  = \Lambda_{\cT}(\mu + 1 + M)
  = \Lambda_{\cT}(\mu) + \Lambda_{\cT}(M+1).
\]
Hence the difference
\[
\Lambda_{\cT}(\mu+1)-\Lambda_{\cT}(\mu)=\lambda(\cT)
\]
is constant for $\mu \in [2T+1, n-8T-6]$.
Hence $\Lambda_{\cT}$ is \emph{affine} on this range, namely, 
\begin{equation}
\label{eq:affinity}
\Lambda_{\cT}(\mu) = \Lambda_{\cT}(M) + (\mu - M)\,\lambda(\cT).
\end{equation}
We next show that the affine $\mu$-intercept is zero. By additivity, $\Lambda_{\cT}(3M) = 3\Lambda_{\cT}(M)$. On the other hand, Eq.~\eqref{eq:affinity} may be applied at $\mu=3M$, since $3M = 12T+6$ lies in the range $[2T+1,n-8T-6]$ by $n \ge 40(T+1)$. Hence Eq.~\eqref{eq:affinity} gives
\[
\Lambda_{\cT}(3M)=\Lambda_{\cT}(M)+2M\,\lambda(\cT).
\]
Comparing the two expressions for $\Lambda_{\cT}(3M)$ yields $\Lambda_{\cT}(M) = M \lambda(\cT)$, and thus
\begin{equation}
\label{eq:linear-no-intercept}
\Lambda_{\cT}(\mu) = \lambda(\cT)\,\mu
\end{equation}
for all $\mu \in [2T+1, n - 8T - 6]$.

We next argue that the slope of Eq.~\eqref{eq:linear-no-intercept} is independent of the basepoint. Let $\cT$ and $\cT'$ be disjoint $2T$-tuples, and set
\[
\alpha=\ShiftIn(\cT,\cT')
\qquad\mbox{and}\qquad
\beta=\ShiftIn(\cT',\cT).
\]
Let $L$ be a loop of length $\mu$ from $\cT'$ whose identifiers, outside $\cT'$, are disjoint from $\cT$. Such a loop exists: choose its $\mu-2T$ fresh identifiers from $[N]\setminus(\cT\cup\cT')$, whose size is at least $2n+2-4T \ge \mu-2T$ in the range used below. The concatenations $\alpha L\beta$ and $\alpha\beta$ are loops from $\cT$ of lengths $4T+\mu$ and $4T$, respectively. Hence
\begin{align*}
\Lambda_{\cT'}(\mu)
  &=~ \omega(L)
  ~=~ \omega(\alpha L\beta)-\omega(\alpha\beta)
  ~=~ \Lambda_{\cT}(4T+\mu)-\Lambda_{\cT}(4T)
  \\
  &=~ \lambda(\cT)(4T+\mu)-\lambda(\cT)\cdot 4T
  ~=~ \lambda(\cT)\,\mu .
\end{align*}
The evaluations $\Lambda_{\cT}(4T+\mu)$ and $\Lambda_{\cT}(4T)$ in the third equality are covered by Eq.~\eqref{eq:linear-no-intercept} whenever $4T+\mu \le n-8T-6$, i.e.\ $\mu \le n-12T-6$. It follows that $\lambda(\cT')=\lambda(\cT)$ for disjoint pairs. Finally, any two $2T$-tuples are linked through a third tuple disjoint from both, since they occupy together at most $4T$ identifiers and $N \ge 2n \ge 6T$. Therefore a single integer $\lambda$ satisfies $\Lambda_{\cT}(\mu)=\lambda\mu$ for every basepoint $\cT$ throughout $[2T+1,n-16T-10]$, as required.
\end{proof}

\begin{proof}[Proof of Theorem~\ref{thm:det-exact}]
Assume $n \ge 40(T+1)$, i.e.\ $T \le n/40 - 1$. The proof has three states.

\midinline Constructing the loops:
Take any assignment $\ID$ of $\cC_n$ and let $\cT^*_0, \cT^*_1, \dots, \cT^*_{n-1}$ be the cyclic sequence of its boundary $2T$-tuples, so that the windows of $\ID$ form a loop of length $n$ from $\cT^*_0$ of weight $f(\ID) = K$. Set $j = \lfloor n/2 \rfloor$; the loop splits into the path $p_1: \cT^*_0 \to \cT^*_j$ ($j$ steps) and the path $p_2 : \cT^*_j \to \cT^*_0$ ($n - j$ steps). The tuples $\cT^*_0$ and $\cT^*_j$ occupy disjoint position sets (as $2T \le j$ and $2T \le n - j$, both by $n \ge 40(T+1)$), hence are disjoint as identifier sets. Define
\[
q=\ShiftIn(\cT^*_j,\cT^*_0)
\qquad\mbox{and}\qquad
r=\ShiftIn(\cT^*_0,\cT^*_j).
\]
Then $p_1 q$ is a loop from $\cT^*_0$ of length $j + 2T$, $p_2 r$ is a loop from $\cT^*_j$ of length $(n-j) + 2T$, and $q r$ is a loop from $\cT^*_j$ of length $4T$.

\midinline Checking the lengths:
All three loop lengths are at least $2T+1$ (indeed, $j \ge 1$, $n - j \ge 1$, and $4T \ge 2T+1$ for $T \ge 1$) and at most $\lceil n/2\rceil + 2T \le n/2 + 2T + 1 \le n - 16T - 10$, the last inequality because $n \ge 40(T+1) \ge 36T + 22$. Hence every $\Lambda$-value below is covered by Lemma~\ref{lemma:linear}, and --- this is the point of the basepoint-independence part of that lemma --- all three are evaluated with one and the same integer slope $\lambda$, although the loops use the two different basepoints $\cT^*_0$ and $\cT^*_j$.

\midinline Evaluating and cancelling:
Summing weights along the two decompositions,
\begin{align*}
\Lambda_{\cT^*_0}(j + 2T)
  + \Lambda_{\cT^*_j}\big((n-j) + 2T\big)
= \omega(p_1) + \omega(p_2) + \omega(q) + \omega(r)
= K + \Lambda_{\cT^*_j}(4T).
\end{align*}
Here $\omega(p_1) + \omega(p_2) = f(\ID) = K$ is the guaranteed count, and $\Lambda_{\cT^*_j}(4T) = \omega(q) + \omega(r) = \omega(qr)$ is the closing correction contributed by the two shift-in paths. Evaluating all three $\Lambda$-terms by Lemma~\ref{lemma:linear}, with the common slope $\lambda$, and rearranging,
\begin{align*}
K
~=~ \lambda\,(j + 2T) + \lambda\,(n - j + 2T) - \lambda\cdot 4T
  ~=~ \lambda\, n .
\end{align*}
Since $\lambda$ is an integer and $0 \le K \le n$, we conclude $K \in \{0, n\}$. The theorem follows.
\end{proof}

The proof also yields the following robustness observation.

\begin{remark}
\label{rem:det-robust}
Theorem~\ref{thm:det-exact} uses no properness, no palette structure, and no probabilistic argument; it applies as is to free colorings with any constant number of colors (via any single color's indicator), to any exactly-guaranteed count $K \notin \{0,n\}$ (not only $n/3$), and it is insensitive to all vertices knowing $n$, $\Delta$, the target frequencies, and the identifier assignment's distribution. The restriction $T \ge 1$ is immaterial: a $0$-round rule is a special case of a $1$-round rule, so the case $T = 0$ follows by applying the theorem with $T = 1$ whenever $n \ge 80$.
\end{remark}

It is natural to ask whether the restrictive identifier-universe hypothesis $N \ge 2n+2$ is essential or is an artifact of the proof. Some hypothesis is essential, but its true form is sharper than one might expect: the threshold sits at $N = n + 1$, and $N = n$ is the \emph{unique} universe size at which rigidity fails. We establish the following.

\begin{theorem}
\label{thm:universe}
(a) For $N = n$ with $3 \mid n$, rigidity fails outright, already for $0$-round rules: the rule ``output $(\mathrm{ID} \bmod 3)$'' is a free $3$-coloring all of whose classes surely have size exactly $n/3$ on every assignment; more generally, every $K \in \{0, 1, \dots, n\}$ is the sure count of some $0$-round rule. 
\\
(b) For every $N \ge n+1$ and every $1 \le T \le n/50 - 1$, the only sure counts are $0$ and $n$.
\end{theorem}

\begin{proof}[Proof sketch]
Part (a): with universe exactly $[n]$, the identifier multiset is \emph{known} to be a permutation of $[n]$, so identifiers carry global cardinality information; the class of $(\mathrm{ID}\bmod 3)$-preimages has size exactly $n/3$ regardless of the assignment. Part (b) is the proof of Theorem~\ref{thm:det-exact} with two changes of accounting: all loops evaluated are kept of length below $\approx n/2$ (the identifier expenditure of Lemmas~\ref{lemma:boundary} and~\ref{lemma:loops} is proportional to the loop length, and this is the only place the hypothesis $N \ge 2n+2$ was used), and the final splitting of the $n$-window loop of an actual assignment is done into \emph{three} arcs of length $\approx n/3$ rather than two of length $\approx n/2$, so that every loop that must be evaluated fits inside the restricted budget; the exchange lemma itself needs only $N \ge n+1$, where the count is exactly tight. The price is only in the constant ($n/50$ for $n/40$). The complete accounting appears in Appendix~\ref{app:universe}.
\end{proof}

Thus, in the local regime, the ``$\mathrm{ID}\bmod 3$'' loophole is not merely closed by large universes: it is the only loophole there is, and one extra identifier already closes it. The hypothesis $N \ge 2n+2$ of Theorem~\ref{thm:det-exact} is retained above only for the sake of its cleaner constant.

\subsection{Stability: imbalance $\imb$ costs $\Omega(n/\imb)$ rounds}
\label{sec:det-stab}

Theorem~\ref{thm:det-exact} characterizes the global counts that a local rule can guarantee \emph{exactly}. We now prove a stability version, essentially stating that a local rule cannot even ensure that its count is confined to a short interval of possible values, unless that interval surrounds an integer multiple of $n$. This settles the deterministic side of the tradeoff on cycles (establishing that imbalance $\imb$ costs $\Omega(n/(\imb+1))$ rounds, which matches the upper bound of Theorem~\ref{thm:balblocks} in Section~\ref{sec:ub} up to the $\log^* n$ factor) and shows that the improvement achieved by the randomized algorithm (Theorem~\ref{thm:balseg}) is unavailable deterministically.

Recall from Section~\ref{sec:det-lb} the local window sum $\Psi(\cI; x)$, defined for a context $\cI$ of $4T$ distinct identifiers and a center $x \notin \cI$. It is the sum of the $2T+1$ values of $F$ on the windows of $\cI[x]$ that contain the center, and it is determined by $F$ alone, not by any identifier assignment. A \emph{flexible gadget} is a pair $\cG=(\cI, \{x, y\})$ with $x, y \notin \cI$ and $\Psi(\cI; x) \ne \Psi(\cI; y)$; it \emph{occupies} the $4T + 2$ identifiers of $\cI^+=\cI \cup \{x, y\}$, and two gadgets $\cG_i,\cG_j$ are \emph{disjoint} if their occupied sets $\cI^+_i,\cI^+_j$ are. Lemma~\ref{lemma:exchange} says that under the surely constant count hypothesis no flexible gadget exists. With slack $\imb$ gadgets may exist, but (this is the dichotomy below) either there are few enough to delete, or many enough to splice (Figure~\ref{fig:splice}). The deletion route rests on the following lemma, which isolates the rigidity content that Case B of the proof below extracts from Section~\ref{sec:det-exact-LB}: exchange-neutrality over a large sub-universe already pins every count to an integer multiple of $n$.

\begin{figure}[t]
\centering
\begin{tikzpicture}[scale=0.65]
\draw[thick] (0,0) circle (1.7);
\foreach \a in {90, 210, 330} {
  \draw[line width=2.5pt] (\a-20:1.7) arc (\a-20:\a+20:1.7);
  \fill (\a:1.7) circle (1.7pt);
}
\node[anchor=south] at (90:1.85) {$\cI_1;\ x_1{\to}y_1$};
\node[anchor=east] at (210:1.85) {$\cI_2;\ x_2{\to}y_2$};
\node[anchor=west] at (330:1.85) {$\cI_3;\ x_3{\to}y_3$};
\node at (30:2.40) {fresh};
\node at (150:2.40) {fresh};
\node at (270:2.05) {fresh};
\end{tikzpicture}
\caption{Case A of the proof of Theorem~\ref{thm:det-stab} (splice), drawn for $\lfloor 2\imb\rfloor+1 = 3$: pairwise disjoint flexible gadgets (contexts $\cI_j$ with centers $x_j$, drawn as dots) are placed, well separated, in a single assignment whose remaining positions carry fresh identifiers. Flipping every center $x_j \to y_j$ moves the count by at least $\lfloor 2\imb\rfloor + 1 > 2\imb$, so the count cannot be confined to any interval of length $2\imb$.}
\label{fig:splice}
\end{figure}
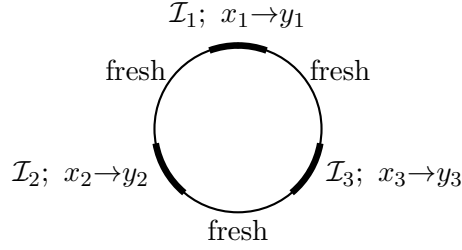

\begin{lemma}
\label{lemma:restricted-rigidity}
Let $n \ge 40(T+1)$, and let $U \subseteq [N]$ with $|U| \ge 2n+2$ be such that $\Psi(\cI; x) = \Psi(\cI; y)$ for every context $\cI$ with entries in $U$ and all identifiers $x, y \in U$ avoiding $\cI$. Then there is an integer $\lambda$ such that every loop with identifiers from $U$ of length $\mu \in [2T+1,\, n - 16T - 10]$ has weight $\lambda\,\mu$; consequently, $f(\ID) = \lambda\, n$ for every assignment $\ID$ of identifiers from $U$ to $\cC_n$.
\end{lemma}

\begin{proof}
Lemmas~\ref{lemma:boundary}, \ref{lemma:loops} and~\ref{lemma:linear} use the hypothesis ``$f$ is surely constant'' \emph{only} through the conclusion of Lemma~\ref{lemma:exchange}, which the present hypothesis supplies over $U$; beyond it, their proofs use only that $F$ is integer-valued and that fresh identifiers can be drawn, and $|U| \ge 2n+2$ supplies the latter exactly as $N \ge 2n+2$ did (cf.\ item (R3) of the range ledger in Section~\ref{sec:det-exact-LB}): the interior-replacement step of Lemma~\ref{lemma:boundary} needs one identifier of $U$ outside a set of size at most $2n$; the common-disjoint-tuple step of Lemma~\ref{lemma:loops} needs $|U| \ge (2n - 4T - 4) + 2T$; and the linking step of Lemma~\ref{lemma:linear} needs $|U| \ge 6T$. Hence the three lemmas hold with $[N]$ replaced by $U$, providing the integer $\lambda$. For the consequence, take an assignment $\ID$ of identifiers from $U$ and split its $n$-window loop at $j = \lfloor n/2\rfloor$ exactly as in the proof of Theorem~\ref{thm:det-exact}; the three auxiliary loops carry identifiers from $U$ and have lengths in the admissible range, so
\[
f(\ID) \;=\; \lambda\,(j + 2T) + \lambda\,(n - j + 2T) - \lambda \cdot 4T \;=\; \lambda\, n .
\]
(The proof of Theorem~\ref{thm:det-exact} used the constancy of $f$ a second time only to call this value $K$; here we simply keep $f(\ID)$.)
\end{proof}

We can now establish the stability theorem.

\begin{theorem}
\label{thm:det-stab}
Let $\imb \ge 0$ be a real number (an integer in all our applications), let $1 \le T \le \frac{n}{40(\imb+1)} - 1$, and let $N \ge 3n + 2$. Let $F$ map $(2T+1)$-tuples of distinct identifiers from $[N]$ to $\{0,1\}$, and suppose that for some real number $K$, every assignment $\ID$ of distinct identifiers to $\cC_n$ satisfies
\[
|f(\ID) - K| \;\le\; \imb,
\qquad
f(\ID) = \sum_{v} F\big(W(v,\ID,T)\big),
\]
where $f$ is as in Theorem~\ref{thm:det-exact}.
Then $|K - \lambda\, n| \le \imb$ for some integer $\lambda$. In particular, no such $F$ exists for $K = n/3$ and $\imb < n/3$.
\end{theorem}

\begin{proof}
The hypothesis on $T$ can be rewritten as $(\imb+1)(T+1) \le n/40$; since $T \ge 1$ it gives $\imb + 1 \le n/80$, and since $\imb \ge 0$ it gives $n \ge 40(T+1)$. We distinguish two cases according to the maximum number of pairwise disjoint flexible gadgets.

\inline Case A: \emph{there exist 
pairwise disjoint flexible gadgets} $\cG_j=(\cI_j, \{x_j, y_j\})$, $j = 1, \dots, \lfloor 2\imb\rfloor+1$. 
\\
Orient each pair so that $\Psi(\cI_j; x_j) < \Psi(\cI_j; y_j)$. Build an assignment $\ID_0$: for each $j$, place on the $4T+1$ consecutive positions $(j-1)(4T+1), \dots, j(4T+1) - 1$ the sequence $\cI_j[x_j]$, so that the center $x_j$ sits at position $i_j = (j-1)(4T+1) + 2T$; fill the remaining positions with fresh distinct identifiers avoiding all occupied ones. This is possible, as (a) the blocks fit, since $(\lfloor 2\imb\rfloor+1)(4T+1) \le (2\imb+1)(4T+1) \le 8(\imb+1)(T+1) \le n/5$, and (b) the identifiers suffice, since the cycle carries $n$ identifiers and the $\lfloor 2\imb\rfloor+1$ centers $y_j$ are reserved, with $n + 2\imb + 1 \le n + n/40 \le N$. Let $\ID_1$ agree with $\ID_0$ except that $\ID_1(i_j) = y_j$ for every $j$.
We claim that
\begin{equation}
\label{eq:ID0-ID1}
f(\ID_1) - f(\ID_0) \;=\; \sum_{j=1}^{\lfloor 2\imb\rfloor+1} \big( \Psi(\cI_j; y_j) - \Psi(\cI_j; x_j) \big).
\end{equation}
Indeed, denote $W_s(\ID)=W(v_s,\ID,T)$, with indices taken modulo $n$. For each gadget center $i_j$, set
\[
A_j=\{i_j-T,\ldots,i_j+T\},
\qquad
B_j=\{i_j-2T,\ldots,i_j+2T\},
\]
again as cyclic intervals. The assignments $\ID_0$ and $\ID_1$ differ exactly at the positions $i_j$, and therefore $W_s(\ID_0) \ne W_s(\ID_1)$ can occur only when $s \in \bigcup_j A_j$. The sets $A_j$ are pairwise disjoint: consecutive centers are $4T+1>2T$ apart, and the cyclic gap between the last center and the first is at least
\[
n-(2\imb+1)(4T+1) \ge 4n/5 > 2T.
\]
Moreover, for every $s \in A_j$, the window $W_s(\ID_b)$ is contained in the block $B_j$, for $b \in \{0,1\}$. On this block, $\ID_0$ reads $\cI_j[x_j]$ and $\ID_1$ reads $\cI_j[y_j]$. Hence
\[
\sum_{s\in A_j}
\big(F(W_s(\ID_1))-F(W_s(\ID_0))\big)
  = \Psi(\cI_j;y_j)-\Psi(\cI_j;x_j).
\]
Summing this identity over $j$ proves Eq.~\eqref{eq:ID0-ID1}.

Each summand in the right-hand side of Eq.~\eqref{eq:ID0-ID1} is a positive integer, hence at least $1$, so $f(\ID_1) - f(\ID_0) \ge \lfloor 2\imb\rfloor + 1 > 2\imb$. But the hypothesis gives $f(\ID_0) \ge K - \imb$ and $f(\ID_1) \le K + \imb$, so $f(\ID_1) - f(\ID_0) \le 2\imb$, contradiction. Case A is therefore impossible under the hypothesis.

\inline Case B: \emph{every family of pairwise disjoint flexible gadgets has size at most $\lfloor 2\imb\rfloor$.} 
\\
Fix a maximal such family $\{\cG_1,\ldots,\cG_\varphi\}$, $\varphi\le \lfloor 2\imb\rfloor$, and denote by $X=\bigcup_{j=1}^\varphi \cI^+_j \subseteq [N]$ the union of its occupied sets, so $|X| \le \lfloor 2\imb\rfloor(4T+2) \le 2\imb\,(4T+2) \le 8(\imb+1)(T+1) \le n/5$; set $U = [N] \setminus X$, so $|U| \ge 3n + 2 - n/5 \ge 2n+2$. By maximality, every flexible gadget $\cG$ intersects $X$.
Consequently, for every context $\cI$ with entries in $U$ and all centers $x, y \in U$ avoiding $\cI$,
$\Psi(\cI; x) \;=\; \Psi(\cI; y)$,
for otherwise, $(\cI, \{x, y\})$ would be a flexible gadget disjoint from $X$, extending the family.

Thus $U$ satisfies the hypothesis of Lemma~\ref{lemma:restricted-rigidity} (and $n \ge 40(T+1)$, as noted at the start of the proof), which yields an integer $\lambda$ with $f(\ID) = \lambda\, n$ for every assignment $\ID$ of identifiers from $U$ to $\cC_n$; such assignments exist, as $|U| \ge n$. The confinement hypothesis applied to one of them gives $|K - \lambda n| = |K - f(\ID)| \le \imb$, as claimed. Finally, for $K = n/3$ we have $|n/3 - \lambda n| \ge n/3$ for every integer $\lambda$, so the conclusion forces $\imb \ge n/3$. The theorem follows.
\end{proof}

Combining the stability (Theorem \ref{thm:det-stab}) with the upper bound yields the following.

\begin{corollary}[Theorems~\ref{thm:intro-ring-lb}(b) and~\ref{thm:intro-cycle-ub}(b), restated: the deterministic tradeoff is tight]
\label{cor:det-stab}
Let $\imb$ be an integer with $0 \le \imb < n/3$ and let $N \ge 3n+2$. Every deterministic $\LOCAL$ algorithm on $\cC_n$ with identifiers from $[N]$ that outputs, on every assignment, a free $3$-coloring with all three color classes of size $n/3 \pm \imb$ requires more than $\frac{n}{40(\imb+1)} - 1$ rounds. For $3 \le \imb < n/3$, the deterministic $\LOCAL$ complexity of $\imb$-additively balanced proper $3$-coloring of $\cC_n$ is $\Theta(n/\imb)$ up to a $\log^* n$ factor.
\end{corollary}

\begin{proof}
As in Corollary~\ref{cor:det-exact}, restrict attention to consistently oriented instances, on which the output at $v_i$ is a fixed function of the oriented radius-$T$ identifier window $W(v_i,\ID,T)$; let $F$ be the indicator of color $1$, so that the guarantee gives $|f(\ID) - n/3| \le \imb$ on every assignment. A $0$-round algorithm is impossible outright: by pigeonhole, some color $b$ is output by at least $N/3 \ge n$ of the radius-$0$ identifier values, and an assignment drawn from those gives class $b$ size $n > n/3 + \imb$. And for $1 \le T \le \frac{n}{40(\imb+1)} - 1$, Theorem~\ref{thm:det-stab} with $K = n/3$ yields an integer $\lambda$ with $|n/3 - \lambda n| \le \imb < n/3$, which no integer satisfies. Hence every correct algorithm has $T > \frac{n}{40(\imb+1)} - 1$.

For $3 \le \imb < n/3$, Theorem~\ref{thm:balblocks} gives a deterministic $O\big((n/\imb)\log^* n\big)$-round algorithm for $\imb$-additively balanced proper $3$-coloring of $\cC_n$. Together with the lower bound above, this gives the stated deterministic tradeoff up to the $\log^* n$ factor.
\end{proof}


\begin{remark}
\label{rem:det-stab}
(a) 
Theorem~\ref{thm:det-stab} holds as is for non-binary rules $F$, with values in $\{0, 1, \dots, \vartheta\}$. The flip effects in Case A are still nonzero integers, and Case B uses only integrality of $F$. Likewise, the robustness of Theorem~\ref{thm:det-exact} recorded in Remark~\ref{rem:det-robust} carries over: free colorings, any constant number of colors, any real target $K$ (note the theorem does not require $K$ to be an integer; this is used for two-valued exact targets in Remark~\ref{rem:diam-div}), and full knowledge of $n$ and the targets.
\\
(b) 
The constants are not optimized; the universe requirement $N \ge 3n+2$, versus $N \ge 2n+2$ for exactness, pays for the deleted hitting set of Case B. The identifier-universe constant can be improved (to $N \ge 2n+2+4\imb(2T+1)$, by re-accounting the deleted gadgets); details are omitted.
\\
(c)
Imbalance $\imb$ costs $\tilde\Theta(n/\imb)$ time deterministically, but only $\tilde O(n/\imb^2)$ randomized (Theorem~\ref{thm:balseg}): an unconditional polynomial deterministic-vs-randomized separation (e.g., of factor $\tilde\Theta(n^{1/3})$ at $\imb = n^{1/3}$) for a natural problem on the cycle, whereas for LCLs on cycles the deterministic and randomized complexity classes coincide~\cite{NS95,Naor91,CKP19}. In this regime, cancellation is available to randomized algorithms but not to deterministic ones; whether the randomized upper bound is optimal is an open problem (see Section~\ref{sec:open}).
\end{remark}

The tight tradeoff 
yields a complexity continuum.

\begin{corollary}
\label{cor:continuum}
For every fixed $\beta \in (0, 1)$ let $\Pi_\beta$ denote the problem of computing a proper $3$-coloring of $\cC_n$ with $|\freq(c) - n/3| \le n^\beta$ for all $c$.
\\
(a)
The deterministic $\LOCAL$ complexity of $\Pi_\beta$ is $\tilde\Theta(n^{1-\beta})$, tight up to a $\log^* n$ factor: $O(n^{1-\beta}\log^* n)$ rounds suffice by Theorem~\ref{thm:balblocks}, and $\Omega(n^{1-\beta})$ rounds are required by Corollary~\ref{cor:det-stab} (both apply throughout $3 \le \imb < n/3$, hence for every fixed $\beta \in (0,1)$ and all large $n$).
\\
(b)
$\Pi_\beta$ is solvable w.h.p.\ in $O(n^{1-2\beta}\log n\log^* n + \log^* n)$ rounds (Theorem~\ref{thm:balseg}).
\end{corollary}

By part (a) of the corollary, the problems $\{\Pi_\beta\}$ realize a dense family of deterministic $\LOCAL$ complexities $\tilde\Theta (n^{\alpha})$ for \emph{every} $\alpha \in (0, 1)$ on cycles. 
In contrast, locally checkable labelings (LCLs) admit only the three complexity classes $O(1)$, $\Theta(\log^* n)$ and $\Theta(n)$~\cite{NS95,CKP19,BHKLOS18}. 

Moreover, contrasting parts (a) and (b) of the corollary, for every fixed $\beta \in (0,1)$, the problems $\{\Pi_\beta\}$ exhibit an unconditional polynomial deterministic-vs-randomized separation on cycles, of factor $n^{\min\{\beta,\,1-\beta\}}/\mathrm{poly}\log n$, 
whereas for LCLs on cycles the deterministic and randomized complexities coincide~\cite{NS95,Naor91,CKP19}.

Near-equitable coloring is a locally checkable labeling \emph{plus one global cardinality constraint}; correctness cannot be verified from constant-radius views, and Corollary~\ref{cor:continuum} shows that such problems escape the LCL complexity-gap structure on cycles. This suggests a broader classification program for globally cardinality-constrained LCLs, which we leave for future work.

\section{Exact balance on general graphs: a deterministic $\Omega(D)$ lower bound}
\label{sec:diameter-lb}

Theorem~\ref{thm:det-exact} shows that exact balance costs $\Omega(n)$ rounds on the cycle. On the cycle, however, $D = \Theta(n)$, so the theorem cannot distinguish a \emph{diameter} barrier from a \emph{size} barrier, and it says nothing about graphs of larger degree. The distributed near-equitable suite of~\cite{NP25} pays $\Theta(D + \Delta)$ rounds per invocation of its global accounting primitives in $\CONGEST$, and its concluding section asks whether a cost at this scale is inherent. In this section we apply Theorem~\ref{thm:det-exact} \emph{as a black box} to show that, for \emph{exact} equity, a diameter-scale lower bound is necessary even in $\LOCAL$, hence with unbounded bandwidth, for every maximum degree $\Delta \ge 3$, and at every diameter scale from the Moore-bound minimum $\Theta(\log n/\log\Delta)$~\cite{MS05} up to the connectivity-imposed maximum $\Theta(n/\Delta)$. For $(1\pm\eta)$-equity, by contrast, such a diameter-scale cost provably is not necessary (Remark~\ref{rem:diam-sep}). Section~\ref{sec:mcc} first presents the reduction on a simple graph family with $D=\Omega(\sqrt{n})$, the matched clique cycle; Section~\ref{sec:twisted} then generalizes this family to \emph{twisted fiber products}, which is what removes the $D = \Omega(\sqrt n)$ restriction.

\subsection{Lower bound on matched clique cycles}
\label{sec:mcc}


We start with a definition of the graph family we work with.
For integers $m \ge 3$ and $\cliqsize \ge 2$, the \emph{matched clique cycle} $B_{m,\cliqsize}$ has vertex set $\mathbb{Z}_m \times [\cliqsize]$ and two kinds of edges: \emph{clique edges} $\{(i,a),(i,b)\}$ for all $i \in \mathbb{Z}_m$ and $a \ne b$, and \emph{matching edges} $\{(i,a),(i+1,a)\}$ for all $i \in \mathbb{Z}_m$ and $a \in [\cliqsize]$ (ring arithmetic modulo $m$). We denote the $i$-th \emph{clique} by $Q_i = \{i\}\times[\cliqsize]$. Equivalently, $B_{m,\cliqsize}$ is the Cartesian product $\cC_m \square K_{\cliqsize}$; contracting every clique $Q_i$ gives the quotient cycle $\cC_m$.
See Figure~\ref{fig:cliquecycle}.

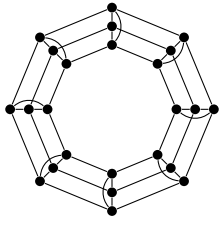
\begin{figure}[t]
\centering
\begin{tikzpicture}[scale=0.55]
\foreach \i in {0,...,7} {
  \foreach \r in {1,2,3} {
    \node[circle,fill,inner sep=1.3pt] (v\i\r) at ({45*\i}:{1.1+0.45*\r}) {};
  }
}
\foreach \i in {0,...,7} {
  \draw (v\i1) -- (v\i2) -- (v\i3);
  \draw (v\i1) to[bend right=40] (v\i3);
}
\foreach \i [evaluate=\i as \j using {int(mod(\i+1,8))}] in {0,...,7} {
  \foreach \r in {1,2,3} { \draw (v\i\r) -- (v\j\r); }
}
\end{tikzpicture}
\caption{The matched clique cycle $B_{8,3}$, composed of
8 triangles (drawn radially) arranged in a ring, with consecutive triangles joined by perfect matchings. Every edge changes the ring coordinate by at most one, so radius-$T$ balls project onto radius-$T$ windows of the quotient cycle $\cC_8$.}
\label{fig:cliquecycle}
\end{figure}

The basic parameters of the family are summarized in the following straightforward observation.

\begin{observation}
\label{obs:clique-cycle}
$B_{m,\cliqsize}$ is $(\cliqsize+1)$-regular on $n = m\cliqsize$ vertices, so $\Delta + 1 = \cliqsize+2$ and $\sigma = m\cliqsize/(\cliqsize+2)$, which is an integer whenever $(\cliqsize+2) \mid m$. Its distances are $\dist\big((i,a),(j,b)\big) = d_{\cC_m}(i,j) + \mathds{1}[a \ne b]$ for $(i,a) \ne (j,b)$, where $d_{\cC_m}$ is the cyclic distance on $\mathbb{Z}_m$; hence $\mathrm{diam}(B_{m,\cliqsize}) = \lfloor m/2 \rfloor + 1 =: D$.
\end{observation}


We get the following result.

\begin{theorem}
\label{thm:diam-lb}
Let $\cliqsize \ge 2$, let $m \ge 80$ be a multiple of $\cliqsize+2$, and set $\sigma = m\cliqsize/(\cliqsize+2)$ and $n = m\cliqsize$. Let $1 \le T \le m/40 - 1$, and let the identifier universe be $[N_B]$ for any $N_B \ge 2n + 2\cliqsize$. No deterministic $T$-round $\LOCAL$ algorithm on the family $B_{m,\cliqsize}$ outputs, for \emph{every} assignment of distinct identifiers from $[N_B]$, a proper coloring with palette $[\cliqsize+2]$ all of whose color classes have size exactly $\sigma$. 
\end{theorem}

\begin{proof}
Suppose such an algorithm $\cA$ exists for $B_{m,\cliqsize}$ and $N_B$; we extract from it a window rule on the cycle family $\cC_m$ violating Theorem~\ref{thm:det-exact}. Since the algorithm must be correct for every port numbering, fix the following \emph{structural} one, independent of identifiers\footnote{Identifier-dependence here would be unsafe: a radius-$T$ view contains, at each vertex on the boundary of the ball, the port numbers of the edges entering the ball, and under, say, identifier-sorted ports the rank at $(i+T,a)$ of its edge to $(i+T-1,a)$ depends on the identifiers in $Q_{i+T+1}$, which lie outside the window used below.}. 
At $(i,a)$, ports $1, \dots, \cliqsize-1$ lead to the clique neighbors $(i,b)$ in increasing order of $b$, port $\cliqsize$ to the matching neighbor $(i+1,a)$, and port $\cliqsize+1$ to $(i-1,a)$. 

The reduction works as follows. Given a cycle instance $\langle \cC_m, \ID \rangle$, we embed it as a matched clique cycle instance, color that instance using algorithm $\cA$, and read from the resulting coloring a coloring rule for $\cC_m$.
Set $N = \lfloor N_B/\cliqsize \rfloor \ge 2m+2$. For $x \in [N]$ let $\psi(x) = \big((x-1)\cliqsize+1, \dots, (x-1)\cliqsize + \cliqsize\big)$; the blocks $\psi(x)$, $x \in [N]$, are pairwise disjoint $\cliqsize$-tuples of identifiers in $[\cliqsize N] \subseteq [N_B]$. 

For an identifier assignment $\ID$ of distinct identifiers from $[N]$ to the vertices of $\cC_m$, let $\widehat{\ID}$ be the assignment to the vertices of $B_{m,\cliqsize}$ that gives vertex $(i,a)$ the identifier $(\ID(i)-1)\cliqsize + a$. Distinct ring identifiers occupy disjoint blocks, so $\widehat{\ID}$ is an assignment of distinct identifiers from $[N_B]$.

Every edge of $B_{m,\cliqsize}$ changes the ring coordinate by at most one, so the radius-$T$ ball of $(i,a)$ is contained in $Q_{i-T} \cup \dots \cup Q_{i+T}$ (these cliques are distinct, as $T < m/2$; containment, rather than equality, is all that is used below).
The subgraph induced by $Q_{i-T} \cup \dots \cup Q_{i+T}$, rooted at $(i,a)$ and carrying its structural port numbers (including those of the boundary matching edges, which are always $\cliqsize$ and $\cliqsize+1$), is the same for every $i$ (it is determined by $m, \cliqsize, T$ alone), and under $\widehat{\ID}$ its identifier labeling is determined by the identifier window $W(i,\ID,T)$ via $\psi$. Since the output of a deterministic $T$-round algorithm, in particular $\cA$, at a vertex is a function of its identifier-labeled, rooted radius-$T$ ball (the global parameters $n, \Delta, m, \cliqsize$ being fixed), there are functions $A_1, \dots, A_\cliqsize$ such that, under every embedded assignment $\widehat{\ID}$, the color assigned by $\cA$ to $(i,a)$ equals $A_a\big(W(i,\ID,T)\big)$.

This observation enables us to define an induced coloring rule for $\cC_m$. For every $(2T+1)$-tuple $w$ of distinct identifiers from $[N]$ define $F(w) = \sum_{a=1}^{\cliqsize} \mathds{1}[A_a(w) = 1]$.
Fix any assignment $\ID$ on $\cC_m$. Since $Q_i$ is a clique of $B_{m,\cliqsize}$ and $\cA$ outputs a proper coloring on $\widehat{\ID}$, the colors $A_1\big(W(i,\ID,T)\big), \dots, A_\cliqsize\big(W(i,\ID,T)\big)$ of its $\cliqsize$ vertices are pairwise distinct, so $F\big(W(i,\ID,T)\big) \in \{0,1\}$: it is the indicator that some vertex of $Q_i$ receives color $1$. Summing over the cliques and using the exactness guarantee of $\cA$ on $\widehat{\ID}$,
\[ \sum_{i \in \mathbb{Z}_m} F\big(W(i,\ID,T)\big) \;=\; \big|\clr^{-1}(1)\big| \;=\; \sigma \qquad\mbox{for \emph{every} assignment } \ID .\]
This leads to contradiction. $F$ maps $(2T+1)$-tuples of distinct identifiers from $[N]$, $N \ge 2m+2$, to $\{0,1\}$, with $1 \le T \le m/40 - 1$, and its window sum equals $\sigma$ on every assignment of $\cC_m$. Theorem~\ref{thm:det-exact} (applied with $n = m$) forces $\sigma \in \{0, m\}$. But $\sigma = m\cdot\frac{\cliqsize}{\cliqsize+2}$ and $0 < \frac{\cliqsize}{\cliqsize+2} < 1$, so $0 < \sigma < m$, contradiction. 
The theorem follows.
\end{proof}

Note that for $B_{m,k}$, $D = \lfloor m/2\rfloor + 1$, so $m \ge 2D - 2$ and hence $m/40 - 1 \ge D/20 - 2$. This implies that exact equitable coloring of $B_{m,k}$ requires $\Omega(D)$ rounds. Since the problem is solvable in $D$ rounds, we get the following.

\begin{corollary} 
The deterministic $\LOCAL$ complexity of exact equitable $(\Delta+1)$-coloring of $B_{m,\cliqsize}$ is $\Theta(D)$.
\end{corollary}

The reduction extends to free colorings, via the following generalization of Theorem~\ref{thm:det-exact} to integer-valued rules.

\begin{corollary}
\label{cor:integer-rules}
(a) Theorem~\ref{thm:det-exact} holds as is for window rules $F$ with values in $\{0, 1, \dots, \vartheta\}$, $\vartheta \ge 1$, i.e., if the window sum equals a constant $K$ on every assignment, then $K \in \{0, n, 2n, \dots, \vartheta n\}$. 
\\
(b)
Theorem~\ref{thm:diam-lb} also holds for free colorings, i.e., no deterministic $(m/40-1)$-round algorithm outputs, on every assignment, a free $(\cliqsize+2)$-coloring of $B_{m,\cliqsize}$  
with all frequencies exactly $\sigma$.
\end{corollary}

\begin{proof}
Inspecting the proof of Theorem~\ref{thm:det-exact}, the range of $F$ enters only through the statement that loop weights are integers (in Lemma~\ref{lemma:loops}, where $[0,\mu]$ becomes $[0,\vartheta\mu]$); Lemmas~\ref{lemma:exchange}, \ref{lemma:boundary} and~\ref{lemma:linear} use only that $F$ is integer-valued. The final identity $\lambda n = K$ with $\lambda \in \mathbb{Z}$ and $0 \le K \le \vartheta n$ yields $K \in \{0, n, \dots, \vartheta n\}$. Claim (a) follows. 
For claim (b),
apply the reduction of Theorem~\ref{thm:diam-lb} without the properness step: $F(w) = \sum_a \mathds{1}[A_a(w) = 1] \in \{0, \dots, \cliqsize\}$ and $\sum_i F\big(W(i,\ID,T)\big) = \sigma$ for every assignment $\ID$ of $\cC_m$, so $\sigma$ would be a multiple of $m$; but $0 < \sigma < m$.
\end{proof}

The stability claim, Theorem \ref{thm:det-stab}, lifts as well.

\begin{corollary}
\label{cor:diam-stab}
Let $\cliqsize \ge 2$ and $m \ge 3$, with $\sigma = m\cliqsize/(\cliqsize+2)$ \emph{not} necessarily an integer, let $\imb$ be an integer with $0 \le \imb < \frac{2m}{\cliqsize+2}$, let $1 \le T \le \frac{m}{40(\imb+1)} - 1$, and let $N_B \ge 3n + 2\cliqsize$.
No deterministic $T$-round $\LOCAL$ algorithm on $B_{m,\cliqsize}$ outputs, on every assignment of distinct identifiers from $[N_B]$, a proper $(\cliqsize+2)$-coloring all of whose class sizes lie in $[\sigma - \imb,\, \sigma + \imb]$; the same holds for free colorings. 
Hence on $B_{m,\cliqsize}$, additive imbalance $\imb < 2\sigma/\cliqsize$ costs $\Omega\big(D/(\imb+1)\big)$ rounds deterministically, while $(1\pm\eta)$-equity at any fixed $\eta$ is diameter-free (Remark~\ref{rem:diam-sep}). 
Equivalently, since $\sigma = m\cliqsize/(\cliqsize+2)$, the lower bound
applies up to relative additive error
$\imb/\sigma < 2/\cliqsize$.
Thus the corollary captures the near-exact regime on this family. This should
be contrasted with any fixed coarse multiplicative slack $\eta>0$ in the
concentration regime, where Remark~\ref{rem:diam-sep} gives algorithms whose
round complexity is independent of $D$.
\end{corollary}

\begin{proof}
The proof consists of four stages.

\midinline (1) extracting the window rule:
Apply the reduction in the proof of Theorem~\ref{thm:diam-lb} as is (note that none of its steps uses the divisibility of $m$). With the structural port numbering and the block embedding (now with $N = \lfloor N_B/\cliqsize \rfloor \ge 3m+2$, using $N_B \ge \cliqsize(3m+2) = 3n + 2\cliqsize$), the rule $F(w) = \sum_a \mathds{1}[A_a(w) = 1]$ is $\{0,1\}$-valued by properness, and the size guarantee gives $\big|\sum_i F\big(W(i,\ID,T)\big) - \sigma\big| \le \imb$ for every assignment $\ID$ of $\cC_m$.

\midinline (2) applying stability:
Theorem~\ref{thm:det-stab}, applied with cycle length $m$ and target frequency $\sigma$, yields an integer $\lambda$ with $|\sigma - \lambda m| \le \imb$. In the free-coloring case, $F$ is $\{0, \dots, \cliqsize\}$-valued and Remark~\ref{rem:det-stab}(a) applies instead, with the same conclusion.

\midinline (3) the arithmetic separation:
Since $0 < \sigma < m$, and since $\cliqsize \ge 2$ gives $\sigma = m\cliqsize/(\cliqsize+2) \ge m/2 \ge m - \sigma$, every integer $\lambda$ satisfies
\[
|\sigma - \lambda m|
\;\ge\; \min(\sigma,\; m - \sigma)
\;=\; m - \sigma
\;=\; \frac{2m}{\cliqsize+2}
\;>\; \imb ,
\]
contradicting Step 2. This proves the unsolvability claims.

\midinline (4) from $m$ to $D$:
Finally, $\frac{2m}{\cliqsize+2} = 2\sigma/\cliqsize$, so the imbalance range of the corollary is as stated, and $m \ge 2D - 2$ turns the round threshold $\frac{m}{40(\imb+1)} - 1$ into $\Omega\big(D/(\imb+1)\big)$.
\end{proof}


\begin{remark}
\label{rem:diam-div}
Stability also disposes of the divisibility assumption in Theorem~\ref{thm:diam-lb}. For every $m \ge \cliqsize + 2$ (not necessarily a multiple of $\cliqsize+2$), the \emph{exact} equitability requirement, namely, all class sizes in $\{\lfloor\sigma\rfloor, \lceil\sigma\rceil\}$, confines the class-$1$ count to $[\sigma - 1, \sigma + 1]$, and $\frac{2m}{\cliqsize+2} \ge 2 > 1$, so Corollary~\ref{cor:diam-stab} with $\imb = 1$ (and a real, non-integer target, permitted by Remark~\ref{rem:det-stab}(a)) shows deterministic unsolvability in $m/80 - 1$ rounds. 
Thus exact equitable $(\Delta+1)$-coloring of $B_{m,\cliqsize}$ has deterministic $\LOCAL$ complexity $\Theta(D)$ for \emph{every} $m \ge \cliqsize+2$.
\end{remark}

Summarizing, we note the following.

\begin{remark}
\label{rem:diam-scope}
\mbox{~}
\\
(a) 
For every $\cliqsize \ge 2$ and every multiple $m$ of $\cliqsize+2$, the family member $B_{m,\cliqsize}$ has $n = m\cliqsize$ and $D = \Theta(m) = \Theta(n/\cliqsize)$; since $m \ge \cliqsize+2$ forces $m = \Omega(\sqrt n)$, the clique cycles alone realize the diameter scales from $\Theta(\sqrt n)$ to $\Theta(n)$. (Section~\ref{sec:twisted} decreases 
the lower restriction
down to $\Theta(\log n/\log\Delta)$ for every fixed degree.) 
\\
(b) 
The hardness is not an artifact of small color classes. In this family, the
target frequency is
$\sigma = \frac{m\cliqsize}{\cliqsize+2} = \Omega(D)$.
Thus, the lower bound does not arise because a class is too small to be sampled
locally; even frequencies as high as the diameter require global coordination when
prescribed exactly.
\\
(c) 
As in Remark~\ref{rem:det-robust},
the bound is insensitive to all vertices knowing $n$, $\Delta$, $m$, $\cliqsize$, the target frequencies, and the graph structure; only the identifiers are unknown. 
\\
(d) 
The assumption $(\cliqsize+2) \mid m$ is used only to make the exact target frequency $\sigma=m\cliqsize/(\cliqsize+2)$ a single integer. If $\sigma$ is not an integer, exact equitability means that each class has one of the two adjacent sizes $\lfloor\sigma\rfloor$ and $\lceil\sigma\rceil$. This is no longer a surely constant count, but it is a count confined to distance less than $1$ from the real target frequency $\sigma$. Remark~\ref{rem:diam-div} applies the stability corollary with $\imb=1$ and removes the divisibility assumption.
\end{remark}


\begin{remark}
\label{rem:diam-sep}
On the same family, near-equity is local: for $\Delta = \cliqsize+1 \le 2^{\sqrt{\log n}/C}$, Theorem~\ref{thm:symcolor} outputs w.h.p.\ a proper coloring with palette exactly $\Delta+1$ and all classes in $(1\pm\eta)\sigma$ within $O(\log(\Delta/\eta)) + O(\log^3\log n)$ rounds, independent of $D$. (For larger $\cliqsize$, Theorem~\ref{thm:split} gives the same conclusion with palette $(1+\eta)(\Delta+1)$ at every $\Delta \le n^{1-o(1)}$, in $\mathrm{poly}\log n$ rounds.) Thus, on matched clique cycles, the transition from two-sided slack $(1\pm\eta)$ to exactness causes the complexity to jump from $\mathrm{poly}(\log\log n)$ (or $\mathrm{poly}\log n$) to $\Theta(D)$ at every realizable diameter scale: the general-graph analogue of the $\Theta(\log^* n)$-vs-$\Theta(n)$ separation on rings (Theorem~\ref{thm:intro-cycle-ub}(c) vs.\ Theorem~\ref{thm:intro-ring-lb}(a)). In particular, the diameter-scale cost of the global accounting primitives of~\cite{NP25} is necessary for exact target frequencies and avoidable for two-sided approximate ones.
\end{remark}

\subsection{Lower bound on twisted fiber products
}
\label{sec:twisted}

The family $B_{m,\cliqsize}$ realizes only diameter scales $D = \Omega(\sqrt n)$, as its cliques must be smaller than its ring. Two observations enable us to remove this limitation. First, the clique structure is unnecessary. Once the integer-valued form of the stability theorem is invoked, any regular graph copied as a fiber over the base cycle works, and choosing a fiber of small diameter and large size (a complete tripartite graph, say) already yields diameters $D = o(\sqrt n)$, though at the price of maximum degree $\Theta(n/D)$ (Remark~\ref{rem:fiber-compare}). Secondly, and more usefully, the reduction never uses the fact that consecutive copies of the fiber are joined by the \emph{identity} matching; any fixed permutation, applied uniformly between every pair of consecutive layers, preserves the two properties that matter. This allows the fiber to be \emph{disconnected}: connectivity is restored by the twist, while bounded component size keeps the degree constant. This leads to the following definition.

\inline Twisted fiber product:
Let $H$ be a $d$-regular graph on vertex set $[p]$, let $\pi$ be a permutation of $[p]$, and let $m \ge 3$. The \emph{twisted fiber product} $\cC_m \ltimes_\pi H$ is the graph on vertex set $\mathbb{Z}_m \times [p]$ with \emph{fiber edges} $\{(i,a), (i,b)\}$ for every $i$ and every $ab \in E(H)$, and \emph{twist edges} $\{(i,a), (i+1, \pi(a))\}$ for every $i \in \mathbb{Z}_m$ and $a \in [p]$.

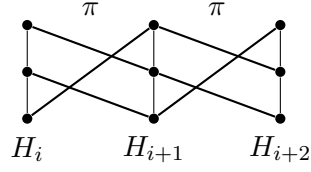
\begin{figure}[t]
\centering
\begin{tikzpicture}[scale=0.62, pt/.style={circle,fill,inner sep=1.35pt}]
\node[pt] (a0) at (0,1) {};
\node[pt] (b0) at (0,0) {};
\node[pt] (c0) at (0,-1) {};
\node[pt] (a1) at (2.7,1) {};
\node[pt] (b1) at (2.7,0) {};
\node[pt] (c1) at (2.7,-1) {};
\node[pt] (a2) at (5.4,1) {};
\node[pt] (b2) at (5.4,0) {};
\node[pt] (c2) at (5.4,-1) {};
\foreach \x in {0,1,2} {
  \draw (a\x) -- (b\x) -- (c\x);
}
\draw[thick] (a0) -- (b1);
\draw[thick] (b0) -- (c1);
\draw[thick] (c0) -- (a1);
\draw[thick] (a1) -- (b2);
\draw[thick] (b1) -- (c2);
\draw[thick] (c1) -- (a2);
\node[draw=none, fill=none] at (0,-1.65) {$H_i$};
\node[draw=none, fill=none] at (2.7,-1.65) {$H_{i+1}$};
\node[draw=none, fill=none] at (5.4,-1.65) {$H_{i+2}$};
\node[draw=none, fill=none] at (1.35,1.35) {$\pi$};
\node[draw=none, fill=none] at (4.05,1.35) {$\pi$};
\end{tikzpicture}
\caption{A twisted fiber product. Each vertical layer is one copy of the fiber
graph $H$ over a vertex of the base cycle. Consecutive layers are joined by the
same permutation $\pi$; the ring closes 
mod $m$.}
\label{fig:twisted-fiber}
\end{figure}

The graph $\cC_m \ltimes_\pi H$ is $(d+2)$-regular on $n = mp$ vertices, and every edge changes the ring coordinate by at most one. For $\pi = \mathrm{id}$ and $H$ a $\cliqsize$-clique it is the matched clique cycle $B_{m,\cliqsize}$, so the following theorem generalizes Theorem~\ref{thm:diam-lb} and Corollary~\ref{cor:diam-stab} simultaneously; its conclusion is driven by the quantity $\delta = \dist(p/q,\, \mathbb{Z})$, where $q = d + 3 = \Delta + 1$, which measures how far the ideal class size $\sigma = mp/q$ is from being realizable by an integer-per-layer rule.

\begin{theorem}
\label{thm:twisted-lift}
Let $G = \cC_m \ltimes_\pi H$ as above, let $q = d + 3 = \Delta(G) + 1$, $\sigma = mp/q$, and $\delta = \dist(p/q, \mathbb{Z})$. Let $\imb \ge 0$ be an integer with $\imb < m\delta$, let $1 \le T \le \frac{m}{40(\imb+1)} - 1$, and let the identifier universe be $[N_B]$ for any $N_B \ge p(3m+2)$. Then no deterministic $T$-round $\LOCAL$ algorithm outputs, on every assignment of distinct identifiers to $G$, a \emph{free} $q$-coloring all of whose color classes have size in $[\sigma - \imb,\, \sigma + \imb]$; in particular the statement holds for proper colorings. 
\end{theorem}

\begin{proof}
The proof of Theorem~\ref{thm:diam-lb} applies with the changes recorded in the following five steps.

\noindent
\emph{(i) Structural ports.} The structural port numbering is now: at $(i,a)$, ports $1, \dots, d$ lead to the fiber neighbors $(i,b)$ in increasing order of $b$, port $d+1$ leads to $(i+1, \pi(a))$, and port $d+2$ to $(i-1, \pi^{-1}(a))$; this numbering is determined by $(H, \pi)$ and is independent of the identifiers.

\noindent
\emph{(ii) No wraparound.} Every edge of $G$ changes the ring coordinate by at most one, so the radius-$T$ ball of $(i,a)$ is contained in the fibers over the $2T+1$ layers $i-T, \dots, i+T$, and these layers are pairwise distinct, since $T \le m/(40(\imb+1)) - 1 \le m/40 - 1 < m/2$, whence $2T+1 < m$. This is the only consequence of the hypothesis on $T$ that the reduction itself uses; the full strength of the hypothesis is consumed by Theorem~\ref{thm:det-stab} in step (iv).

\noindent
\emph{(iii) The layer-independent window rule.} Since the same permutation $\pi$ joins every pair of consecutive layers, the one-layer translation $(j, a') \mapsto (j+1, a')$ is an automorphism of the ported, unlabeled graph, so the rooted, ported, unlabeled radius-$T$ ball of $(i,a)$ is determined by $a$ and $T$ alone, independently of $i$. Under the block embedding $\psi$ (with blocks of size $p$, using $N = \lfloor N_B/p\rfloor \ge 3m+2$), the identifier labeling of this ball induced by an assignment $\ID$ of $\cC_m$ is determined by the window $W(i,\ID,T)$: a ball vertex lying over layer $i+t$, $-T \le t \le T$, carries an identifier from the $\psi$-block of the window entry at offset $t$, and by (ii) the layer offset $t$ of a ball vertex is unambiguous. (In particular, $W(i,\ID,T)$ reads $2T+1 < m$ pairwise distinct positions of $\cC_m$, hence is a tuple of distinct identifiers, as Theorem~\ref{thm:det-stab} requires.) Hence, exactly as before, there are functions $A_1, \dots, A_p$ with output $A_a(W(i,\ID,T))$ at $(i,a)$. Properness is not assumed, so the induced rule $F(w) = \sum_{a=1}^{p} \mathds{1}[A_a(w) = 1]$ takes values in $\{0, 1, \dots, p\}$ rather than $\{0,1\}$.

\noindent
\emph{(iv) Stability, in integer-valued real-target form.} The balance guarantee confines rather than fixes the window sum: $\big|\sum_{i \in \mathbb{Z}_m} F(W(i,\ID,T)) - \sigma\big| \le \imb$ for every assignment $\ID$ of $\cC_m$. The integer-valued, real-target form of Theorem~\ref{thm:det-stab} applies (integer-valuedness because $F$ may count several vertices of the same fiber; the real target because $\sigma$ need not be integral) and yields an integer $\lambda$ with $|\sigma - \lambda m| \le \imb$.

\noindent
\emph{(v) The $\delta$-contradiction.} On the other hand, we have
\begin{align*}
\min_{\lambda \in \mathbb{Z}} |\sigma - \lambda m|
  ~=~ m \cdot \dist\big(\sigma/m,\ \mathbb{Z}\big)
  ~=~ m\,\dist(p/q, \mathbb{Z})
  ~=~ m\delta
   ~>~ \imb,
\end{align*}
contradiction. The theorem follows.
\end{proof}

\begin{corollary}
\label{cor:twisted-exact}
Let $G=\cC_m\ltimes_\pi H$, let $q=d+3=\Delta(G)+1$, set
$\sigma=mp/q$ and $\delta=\dist(p/q,\mathbb{Z})$, and assume that the
identifier universe is $[N_B]$ with $N_B\ge p(3m+2)$. If $m\delta>1$, then
exact equitable $(\Delta+1)$-coloring of $G$ is deterministically unsolvable in
$m/80-1$ rounds, even for free colorings, and is solvable in
$\mathrm{diam}(G)$ rounds.
\end{corollary}

\begin{proof}
Exact equitability confines every class size to
$\{\lfloor\sigma\rfloor,\lceil\sigma\rceil\}$, hence to distance less than $1$
from the real target frequency $\sigma$. Applying Theorem~\ref{thm:twisted-lift} with
$\imb=1$ rules out every deterministic $T$-round algorithm for
$T\le m/80-1$. For the upper bound, every vertex collects the whole graph in
$\mathrm{diam}(G)$ rounds and applies the Hajnal--Szemer\'edi theorem~\cite{HS70}
locally to compute the same exact equitable coloring.
\end{proof}

We now choose the fiber and the twist so as to fix the degree and let the diameter range freely. The fiber is a disjoint union of $r$ cliques of size $\cliqsize$, which pins the degree at $\Delta = \cliqsize+1$; the twist is a base-$\cliqsize$ de Bruijn map on the cliques, which makes the clique coordinate mix in $\log_\cliqsize r$ steps and thereby keeps the diameter at $\Theta(m + \log_\cliqsize r)$.

\begin{definition}[de Bruijn clique cycles, $\dBCC$]
\label{def:debruijn}
Let $\cliqsize \ge 2$ and $r \ge 1$ with $\gcd(\cliqsize, r) = 1$, and let $m \ge 3$. Index the fiber by pairs $(s, a) \in \mathbb{Z}_r \times \{0, \dots, \cliqsize-1\}$, let $H_{r,\cliqsize}$ be the disjoint union of the $r$ cliques $\{s\} \times \{0, \dots, \cliqsize-1\}$, and let $\pi(s, a) = (\cliqsize s + a \bmod r,\ a)$; this is a permutation, since for fixed $a$ the map $s \mapsto \cliqsize s + a$ is a bijection of $\mathbb{Z}_r$. The \emph{de Bruijn clique cycle} is $G_{m,r,\cliqsize} = \cC_m \ltimes_\pi H_{r,\cliqsize}$; it is $(\cliqsize+1)$-regular on $n = mr\cliqsize$ vertices, and $G_{m,1,\cliqsize} = B_{m,\cliqsize}$.
\end{definition}

The diameter of the family $\dBCC$ is given by the following.

\begin{observation}
\label{lem:debruijn-diam}
$G_{m,r,\cliqsize}$ is connected, with $\lfloor m/2\rfloor \le \mathrm{diam}(G_{m,r,\cliqsize}) \le 2m + 4\lceil\log_\cliqsize r\rceil + 6$; hence $\mathrm{diam} = \Theta(m + \log_\cliqsize r)$, and the Moore bound gives $\mathrm{diam} = \Omega(\log n / \log(\cliqsize+1))$ unconditionally.
\end{observation}

\begin{proof}
A route from $(i, s, a)$ advancing one layer rightwards consists of at most one clique edge (choosing $a' \in \{0, \dots, \cliqsize-1\}$) followed by the twist edge to $(i+1,\, \cliqsize s + a' \bmod r,\, a')$; thus $\ell$ consecutive right-steps with chosen indices $a_1, \dots, a_\ell$ end at clique coordinate $\cliqsize^\ell s + \sum_{t=1}^{\ell} \cliqsize^{\ell-t} a_t \pmod r$ and fiber index $a_\ell$, at a cost of at most $2\ell$ edges.

Fix a target $(s', a')$. The final fiber index is forced by the last digit alone, so set $a_\ell = a'$; its contribution to the sum is the constant $\cliqsize^0 a_\ell = a'$, and, as the $\ell - 1$ free digits $a_1, \dots, a_{\ell-1}$ range over $\{0, \dots, \cliqsize-1\}^{\ell-1}$, the remaining sum $\cliqsize\sum_{t=1}^{\ell-1} \cliqsize^{(\ell-1)-t} a_t$ runs over exactly $\cliqsize$ times the base-$\cliqsize$ representations $\{0, 1, \dots, \cliqsize^{\ell-1} - 1\}$. The set of clique coordinates reachable by $\ell$ right-steps ending in fiber index $a'$ is therefore the translate
\[
\cliqsize^\ell s + a' \;+\; \cliqsize \cdot \{0, 1, \ldots, \cliqsize^{\ell-1} - 1\}
\pmod r .
\]
This set covers all residues modulo $r$ once $\cliqsize^{\ell-1} \ge r$: the interval $\{0, \dots, \cliqsize^{\ell-1}-1\}$ then contains every residue, multiplication by $\cliqsize$ permutes $\mathbb{Z}_r$ (as $\gcd(\cliqsize, r) = 1$), and so does the translation. Hence for every $\ell \ge \lceil \log_\cliqsize r\rceil + 1$, every target is reachable by a route of exactly $\ell$ right-steps.

It remains to enforce the layer: reaching layer $i'$ at rightward cyclic distance $d^+$ requires $\ell \equiv d^+ \pmod m$. If $d^+ \ge \lceil\log_\cliqsize r\rceil + 1$, take $\ell = d^+$; otherwise take $\ell = d^+ + m\,\big\lceil\big(\lceil\log_\cliqsize r\rceil + 1 - d^+\big)/m\big\rceil$, the smallest integer that is at least $\lceil\log_\cliqsize r\rceil + 1$ and equals $d^+$ mod $m$, wrapping around the ring as many times as needed. In both cases $\ell \le m + \lceil\log_\cliqsize r\rceil$, so the cost is at most $2m + 2\lceil\log_\cliqsize r\rceil$ edges, which gives connectivity and the upper bound. The lower bound holds since every edge changes the ring coordinate by at most one, and the Moore bound since the graph is $(\cliqsize+1)$-regular. The observation follows.
\end{proof}

Combining the twisted-fiber reduction with the de Bruijn family gives the
desired extension to every maximum degree and every admissible diameter scale.

\begin{corollary}
\label{cor:all-scales}
Let $\cliqsize \ge 2$ and let $r \equiv 1 \pmod{\mathrm{lcm}(\cliqsize,\, \cliqsize+2)}$ (such $r$ exist in every interval of length $\cliqsize(\cliqsize+2)$), and let $m \ge \max\{\cliqsize+2,\ \lceil\log_\cliqsize r\rceil\}$, $N_B \ge r\cliqsize(3m+2)$. Then on $G = G_{m,r,\cliqsize}$, with $\Delta = \cliqsize+1$ and $D = \mathrm{diam}(G) = \Theta(m)$, the following hold.
\\
(a) Exact equitable $(\Delta+1)$-coloring is deterministically unsolvable in
$m/80 - 1 = \Omega(D)$ rounds, even for free colorings, and solvable in $D$
rounds, so its deterministic $\LOCAL$ complexity is $\Theta(D)$.
\\
(b) Every integer additive imbalance $\imb < 2m/(\cliqsize+2)$ costs $\Omega(D/(\imb+1))$ rounds.
\\
(c)
Consequently, for every fixed $\Delta \ge 3$ the family $\{G_{m,r,\Delta-1}\}$ realizes, up to constant factors and divisibility adjustments, every pair $(n, D)$ with
\begin{equation}
\label{eq:all-scales-range}
C\,\frac{\log n}{\log\Delta} \;\le\; D \;\le\; \frac{n}{C\Delta}
\end{equation}
for a universal constant $C$, and on every member exact balance costs $\Theta(D)$ time.
\end{corollary}

\begin{proof}
The fiber $H_{r,\cliqsize}$ is $(\cliqsize-1)$-regular on $p = r\cliqsize$ vertices, so $q = \cliqsize+2 = \Delta+1$ and $\delta = \dist\big(r\cliqsize/(\cliqsize+2),\, \mathbb{Z}\big)$. Since $r \equiv 1 \pmod{\cliqsize+2}$, we have $r\cliqsize = r(\cliqsize+2) - 2r$ and $2r \equiv 2 \pmod{\cliqsize+2}$, so $\delta = \dist\big(2r/(\cliqsize+2), \mathbb{Z}\big) = 2/(\cliqsize+2)$, whence $m\delta = 2m/(\cliqsize+2) \ge 2 > 1$ (using $m \ge \cliqsize+2$). Also $r \equiv 1 \pmod \cliqsize$ gives $\gcd(\cliqsize,r) = 1$, so Definition~\ref{def:debruijn} applies, and $m \ge \lceil\log_\cliqsize r\rceil$ gives $D = \Theta(m)$ by Observation~\ref{lem:debruijn-diam}. Corollary~\ref{cor:twisted-exact} yields part (a), and Theorem~\ref{thm:twisted-lift} with general $\imb < m\delta$ yields part (b); $m = \Theta(D)$ converts the thresholds. For part (c),
given $n$, $\Delta = \cliqsize+1$ and a target $D$ satisfying Eq.~\eqref{eq:all-scales-range}, take $m = \Theta(D)$ and $r = n/(m\cliqsize)$ adjusted to the nearest admissible residue; the constraint $m \ge \lceil\log_\cliqsize r\rceil$ is the left inequality of Eq.~\eqref{eq:all-scales-range}, and $r \ge 1$ is the right one. The corollary follows.
\end{proof}

We next compare the two special graph classes mentioned earlier. 

\begin{remark}
\label{rem:fiber-compare}
Taking the complete tripartite graph $K_{s,s,s}$ as the fiber gives
$p=3s$, $q=2s+3$, and, with $\pi=\mathrm{id}$, diameter
$D=\lfloor m/2\rfloor+2$ and order $n=3ms$. For $s\ge 12$,
$\delta=\dist(3s/(2s+3),\mathbb{Z})\ge 1/3$. Thus, by setting
$s=\Theta(n/D)$, this choice realizes every diverging diameter scale,
including $D=o(\sqrt n)$. The price is maximum degree
$\Delta = 2s+2 = \Theta(n/D)$; at $D=\Theta(\log n)$ the degree is nearly
linear.
Thus, the de Bruijn clique cycle and the tripartite product are complementary
specializations of Theorem~\ref{thm:twisted-lift}: the former fixes the degree
and pays a $\log_\cliqsize r$ diameter term, absorbed by
$m \ge \log_\cliqsize r$, while the latter has constant fiber diameter but
degree $\Theta(n/D)$. Together with Remark~\ref{rem:diam-scope}(b), the
phenomenon is governed by neither degree nor class size: on
$G_{m,r,\cliqsize}$ with constant $\cliqsize$ the classes have size
$\sigma=\Theta(n)$, yet exactness still costs diameter time.
\end{remark}

We conclude the section by stating the remaining geometric question.

\begin{remark}
\label{rem:open}
At $m=\Theta(\log r)$ and constant $\cliqsize$, the graphs
$G_{m,r,\cliqsize}$ are bounded-degree graphs of logarithmic diameter, the
extreme point permitted by the Moore bound, but they are not spectral expanders
(the ring coordinate mixes in $\Theta(m^2)$ steps). Whether exact balance costs $\Omega(D)$ time on genuinely expanding graphs remains open; the present family
settles every diameter \emph{scale}, not every \emph{geometry}. As in
Remark~\ref{rem:diam-scope}(c), all bounds are insensitive to every vertex
knowing $n$, $\Delta$, $m$, $r$, $\cliqsize$, $\pi$ and the targets.
\end{remark}

\section{Balanced colorings of cycles: upper bounds}
\label{sec:ub}

\subsection{The randomized algorithm \BalSeg}
\label{sec:balseg}

\paragraph{Overview.}
Intuitively, the algorithm first cuts the cycle into long, deterministic pieces.
It computes the distance-$w$ ruling set supplied by Algorithm $\Ruling$ of Fact~\ref{fact:ruling}, and
uses its vertices as \emph{cuts}. The open intervals between consecutive cuts
are the \emph{segments}; each has length $\Theta(T/\log^* n)$. The algorithm then
colors every segment periodically, but chooses the starting phase of each segment
independently and uniformly in $\Zthree$. Finally, every cut vertex is colored
after seeing the two colors adjacent to it, using one private coin only in the
tie case.

This construction separates the deterministic part of the algorithm from the
random part. The ruling set guarantees, without any failure probability, that
the number of segments is $O((n/T)\log^* n)$. Each segment contributes only
$O(1)$ imbalance, and the random phase makes this contribution mean-zero. The
total segment imbalance is therefore a sum of bounded mean-zero terms; after
partitioning these terms into constantly many independent families, the
imbalance concentrates at scale $O(\sqrt{(n/T)\log^* n\log n})$. The cut
vertices contribute the same type of bounded mean-zero error, by
Lemma~\ref{lemma:junction}.

More formally, let $L = \max\{1, \log^* n\}$. The algorithm has a parameter $T \ge \const_1 L$ and sets $w = 2 + \lceil T/(3\const_1 L)\rceil$ for a sufficiently large constant $\const_1$. The properties of the cutting step are summarized next.

\begin{algorithm}[t]
\caption{$\BalSeg(\cC_n, T)$ \label{alg:balseg}}
\KwIn{Cycle $\cC_n$, parameter $T \ge \const_1 L$ where $L = \max\{1,\log^* n\}$;}
\KwOut{A proper $3$-coloring, $\imb$-additively balanced for $\imb = O\big(\sqrt{(n/T)\,L\log n}\big)$, w.h.p.}
$w \leftarrow 2 + \lceil T/(3\const_1 L)\rceil$\;
Compute a distance-$w$ ruling set $S$ with consecutive gaps in $[w+1, 2w+1]$ (by Algorithm $\Ruling$ of Fact~\ref{fact:ruling}); its elements are \emph{cuts}; maximal runs between cuts are \emph{segments}\;
\For{each segment $j$ (in parallel)}{
  Let $a_j, b_j$ be the incident cuts, $a_j$ the one with the smaller identifier; orient the segment from $a_j$ to $b_j$; let $\ell_j$ be its length\;
  The segment vertex adjacent to $a_j$ draws $s_j \in \Zthree$ uniformly and broadcasts $s_j$ to all vertices of segment $j$ by flooding along the segment\;
  The $t$-th segment vertex in the chosen direction outputs color $s_j + (t-1)$, for $t = 1, \dots, \ell_j$\;
}
\For{each cut $u$ (in parallel)}{
  Let $x, y$ be the colors of $u$'s two neighbors (both segment vertices)\;
  \lIf{$x \ne y$}{$u$ outputs the unique color in $\Zthree\setminus\{x,y\}$}
  \lElse{$u$ outputs a uniformly random color in $\Zthree\setminus\{x\}$}
}
\KwReturn{$\clr$}
\end{algorithm}

The first lemma records the deterministic guarantees of the cutting step.

\begin{lemma}
\label{lemma:balseg-valid}
Deterministically: every two cuts are at distance at least $w + 1 \ge 3$, so each segment has length $\ell_j \in [w,\, 2w+1]$ and each cut's two neighbors lie in segments; the number $M$ of cuts, which equals the number of segments, satisfies $2 \le M \le n/(w+1) = O\big((n/T)\,L\big)$, so the two neighbors of every cut lie in two \emph{distinct} segments; the cutting step costs $O(w\log^* n) \le T/3$ rounds; and the output coloring is proper.
\end{lemma}

\begin{proof}
Fact~\ref{fact:ruling} guarantees consecutive gaps in $[w+1, 2w+1]$, so consecutive cuts are at distance at least $w+1 \ge 3$ and the open runs between them have lengths in $[w, 2w]$; the count follows from $M(w+1) \le n$ together with $w + 1 = \Theta(T/L)$, and the round cost of Algorithm $\Ruling$ from Fact~\ref{fact:ruling} with the choice of $\const_1$. For $M \ge 2$, note that a single cut would make its own gap wrap around the entire cycle, forcing $n \le 2w+2$; but $w \le 3 + T/(3\const_1) \le 3 + n/3$, so $2w+2 \le 2n/3 + 8 < n$ for all sufficiently large $n$. Hence the two neighbors of a cut lie in the two distinct segments flanking it. Properness: inside segments, consecutive colors differ by $+1$; each segment--cut edge is handled by the cut's explicit avoidance; there are no cut--cut edges. The lemma follows.
\end{proof}

We remark that the construction is locally implementable without a globally consistent orientation.
a segment has length at most $2w+1 < T$, so each segment vertex sees, in its radius-$T$ ball, both cuts incident to its segment together with their identifiers and its own distance to each; the orientation ``from $a_j$ to $b_j$'', the index $t$, and the phase $s_j$ after it is flooded along the segment are then achievable by local computations.

For the frequency accounting, note that the segment structure is deterministic; the only randomness is the phase draws $(s_j)_j$ and the cuts' tie-breaking coins. For a segment $j$, denote by $m_j(c)$ the number of vertices of segment $j$ colored $c$, and set $e_j(c) = m_j(c) - \ell_j/3$. For a cut $u$, set $e'_u(c) = \mathds{1}[\clr(u) = c] - 1/3$. For a vector $x=(x(c))_{c\in\Zthree}$, denote $\|x\|_\infty = \max_{c\in\Zthree}|x(c)|$. We first bound the segment errors.

\begin{lemma}
\label{lemma:segment-error}
For every fixed $\ell_j$: if $\ell_j \equiv 0 \pmod 3$ then $e_j \equiv 0$; otherwise $\|e_j\|_\infty \le 2/3$ and $\E[e_j(c)] = 0$ for every $c$.
\end{lemma}

\begin{proof}
The periodic pattern of length $\ell$ starting at $s$ gives color $s + r$ frequency $\lceil (\ell - r)/3\rceil$ for $r \in \{0,1,2\}$. For $\ell \equiv 0$, all three counts equal $\ell/3$. For $\ell \equiv 1$, the error vector is $(2/3, -1/3, -1/3)$ placed at the rotation $s_j$; for $\ell \equiv 2$ it is $(1/3, 1/3, -2/3)$ at rotation $s_j$. In both cases the rotation is uniform, so the mean is zero. The lemma follows.
\end{proof}

The crux of the analysis is the following lemma.

\begin{lemma}
\label{lemma:junction}
For every cut $u$, the color $\clr(u)$ is exactly uniform on $\Zthree$; hence $\E[e'_u(c)] = 0$ and $\|e'_u\|_\infty \le 2/3$.
\end{lemma}

\begin{proof}
Let $u$ lie between segments $j$ and $j+1$ (distinct, by Lemma~\ref{lemma:balseg-valid}) with neighbor colors $x$ (from segment $j$) and $y$ (from segment $j+1$). The segment structure and the orientation choices are deterministic, so $x$ is a fixed shift of the phase $s_j$ and $y$ a fixed shift of the phase $s_{j+1}$; hence $x, y$ are uniform on $\Zthree$ and independent. Condition on $x$: with probability $1/3$, $y = x$ and $\clr(u)$ is uniform on $\Zthree\setminus\{x\}$ by the explicit coin; with probability $2/3$, $y$ is uniform on $\Zthree\setminus\{x\}$ and $\clr(u)$ is the third color, again uniform on $\Zthree\setminus\{x\}$. Hence $\clr(u) \mid x \sim \mathrm{Unif}(\Zthree\setminus\{x\})$ and $\Pr[\clr(u) = c] = \Pr[x \ne c] \cdot \tfrac12 = \tfrac13$. The lemma follows.
\end{proof}

We get the following result.

\begin{theorem}[Theorem~\ref{thm:intro-cycle-ub}(a), formal]
\label{thm:balseg}
Let $L = \max\{1, \log^* n\}$. For every $T \in [\const_1 L,\, n]$, Algorithm $\BalSeg(\cC_n, T)$ runs in $O(T)$ rounds and w.h.p.\ outputs a proper $3$-coloring with
$\max_c |\freq(c) - n/3| = O\big(\sqrt{(n/T)\,L\log n}\big)$. Consequently, for \emph{every} $\imb \ge 1$, imbalance $\imb$ is achievable w.h.p.\ in $O(\min\{\,n,\ n\,L\log n/\imb^2 + L\,\})$
rounds. (The additive $L$ term is the smallest legal horizon of the algorithm, and some additive term is unavoidable: proper $3$-coloring alone costs $\Omega(\log^* n)$ rounds, even randomized~\cite{Naor91}.)
\end{theorem}

\begin{proof}
It suffices to treat $T \le n/8$: for larger $T$, every vertex collects the entire cycle in $O(n) = O(T)$ rounds and outputs its color in a canonical proper $3$-coloring with imbalance at most $1$. The round complexity is immediate from Lemma~\ref{lemma:balseg-valid}. For every color $c$,
\[ \freq(c) - n/3 \;=\; \sum_{j} e_j(c) + \sum_{u\ \mathrm{cut}} e'_u(c), \]
since segments and cuts partition the cycle. All terms are mean-zero (Lemmas~\ref{lemma:segment-error}, \ref{lemma:junction}) and bounded by $1$. Their dependence structure is summarized by the following ledger of four families; independence is claimed only \emph{within} each family, never across families, and the union bound below needs no more.
\begin{itemize}
\item (F1), \emph{The segment family} $(e_j)_j$: mutually independent, since distinct segments hold disjoint phase coins.
\item (F2)--(F4), \emph{The three cut families.} The cut between segments $j, j+1$ reads $s_j$, $s_{j+1}$ and a private coin, so the cut terms $(e'_u)_u$ form a $1$-dependent family along the cyclic order of the cuts. Splitting them into the color classes of a proper $3$-coloring of this dependency cycle (Section~\ref{sec:prelim}; a parity split would fail when the number of cuts is odd) yields three families, each consisting of cuts pairwise at cyclic distance at least $2$, whose terms draw on pairwise disjoint sets of phases and coins --- hence each family is internally mutually independent.
\end{itemize}
Each family consists of mean-zero terms bounded by $1$ and has size at most $M = O((n/T)\,L)$. By Theorem~\ref{thm:hoeffding}(a), each family sum is $O(\sqrt{M\log n})$ with probability $1 - n^{-\gamma}$, and a union bound over the four families and the three colors yields, w.h.p.,
\[ \max_c |\freq(c) - n/3| \;=\; O\big(\sqrt{M\log n}\big) \;=\; O\big(\sqrt{(n/T)\,L\log n}\big). \]
For the restatement, let $\imb \ge 1$. If $(n/\imb^2)L\log n + L \ge n/(8C)$, run the global fallback ($T = n$): the output is balanced up to $1 \le \imb$ in $O(n)$ rounds, matching the minimum. Otherwise set $T = C\big((n/\imb^2)L\log n + L\big) \le n/8$; then $T \ge \const_1 L$ for $C \ge \const_1$, and the displayed imbalance is $O\big(\sqrt{(n/T)L\log n}\big) = O(\imb/\sqrt{C}) \le \imb$ for $C$ a large enough constant. The theorem follows.
\end{proof}


\begin{remark}
The anchors are the only deterministic ingredient of $\BalSeg$, and they are exactly what the $L$ factor pays for. Using random marks instead (each vertex marks itself with probability $1/w$, sparsified to pairwise distance $\ge 3$) yields the same analysis with $O((n/T)\log n)$ error terms, where the $\log n$ factor pays for maximum-gap control. This gives the weaker imbalance $O(\sqrt{n/T}\cdot\log n)$, a $\Theta(\log n)$ minimal horizon, and additional failure events. The deterministic ruling set removes all three losses simultaneously. This is the one place where the randomized and deterministic algorithms of this section share machinery: both stand on Fact~\ref{fact:ruling}, and Proposition~\ref{prop:ruling-lb} below shows this is no accident, namely, that the $\log^* n$ cost of the anchors themselves cannot be improved.
\end{remark}

\subsection{The deterministic algorithm \BalBlocks}
\label{sec:balblocks}

\paragraph{Overview.}
The deterministic algorithm replaces the random cuts by a ruling set and the random phases by exhaustive local choice: the anchors of a ruling set partition the cycle into blocks of length $\Theta(w)$, and each block is colored by a proper coloring whose class counts are within $2$ of perfectly balanced and whose endpoints avoid the colors of the already-colored neighboring blocks. Such a coloring always exists on paths; the building block is the following lemma.

\begin{lemma}
\label{lemma:balanced-path}
Let $P = u_1, \dots, u_\ell$ be a path, $\ell \ge 1$, and let $F_1, F_2 \subseteq \Zthree$, $|F_1|, |F_2| \le 1$, be forbidden sets for $u_1$ and $u_\ell$. There is a proper $3$-coloring $\clr$ of $P$ with $\clr(u_1) \notin F_1$, $\clr(u_\ell)\notin F_2$ and $|\freq(c) - \ell/3| \le 2$ for every $c$, computable from $(\ell, F_1, F_2)$ by a fixed rule.
\end{lemma}

\begin{proof}
If $\ell = 1$, color $u_1$ by the smallest color outside $F_1 \cup F_2$ (one exists since $|F_1 \cup F_2| \le 2$); every frequency deviation is at most $1$. For $\ell \ge 2$: let $x$ be the element of $F_1$ (or $x = 2$ if $F_1 = \emptyset$) and color periodically $u_t = x + t \bmod 3$; then $u_1 = x+1 \notin F_1$, the coloring is proper, and every frequency is within $1$ of $\ell/3$. If $\clr(u_\ell) \in F_2$, recolor $u_\ell$ by the unique color outside $\{\clr(u_{\ell-1})\}\cup F_2$ (a set of size exactly $2$, since $F_2 = \{\clr(u_\ell)\}$ and $\clr(u_\ell) \ne \clr(u_{\ell-1})$); this preserves properness and moves two frequencies by one, for final deviation at most $2$. The lemma follows.
\end{proof}

\begin{algorithm}[t]
\caption{$\BalBlocks(\cC_n, T)$ \label{alg:balblocks}}
\KwIn{Cycle $\cC_n$, parameter $T \ge \const_2\log^* n$;}
\KwOut{A proper $3$-coloring with imbalance $O\big((n/T)\log^* n\big)$.}
$w \leftarrow \lceil T/(\const_2\log^* n)\rceil$\;
Compute a ruling set $S$ with gaps in $[w+1, 2w+1]$ (by Algorithm $\Ruling$ of Fact~\ref{fact:ruling}); its elements are \emph{anchors}\;
Define blocks $B_j \leftarrow [a_j, a_{j+1})$ for consecutive anchors $a_j$; blocks partition $V$, $|B_j| \in [w+1, 2w+1]$\;
$3$-color the cycle of blocks (blocks as supernodes) by Linial color reduction followed by Cole--Vishkin-style reduction to three colors~\cite{Linial92,CV86}, in $O(w\log^* n)$ rounds; let $\psi(B_j)$ be the phase\;
\For{phase $p = 1,2,3$}{
  \For{each block $B_j$ with $\psi(B_j) = p$ (in parallel)}{
    Let $F_1$ ($F_2$) be the color of $B_j$'s left (right) outside neighbor if colored, else $\emptyset$\;
    Color $B_j$ by the rule of Lemma~\ref{lemma:balanced-path} on $(|B_j|, F_1, F_2)$\;
  }
}
\KwReturn{$\clr$}
\end{algorithm}

The properties of the algorithm are summarized in the following theorem.

\begin{theorem}[Theorem~\ref{thm:intro-cycle-ub}(b), formal]
\label{thm:balblocks}
For every $T \in [\const_2\log^* n, n]$, Algorithm $\BalBlocks(\cC_n, T)$ is deterministic, runs in $O(T)$ rounds, and outputs a proper $3$-coloring with $\max_c|\freq(c) - n/3| \le 2|S| = O((n/T)\log^* n)$. Equivalently, imbalance $\imb \ge 3$ is achievable in $O((n/\imb)\log^* n)$ rounds, for every $3 \le \imb < n/3$.
\end{theorem}

\begin{proof}
Let us first bound the time complexity.  Algorithm $\Ruling$ of Fact~\ref{fact:ruling} costs $O(w\log^* n) = O(T)$; the block cycle has at most $n/w$ supernodes and each of its rounds costs $O(w)$, so the supernode coloring costs $O(w\log^* n)$ (Linial's algorithm~\cite{Linial92} brings the supernode colors down to $O(1)$ in $O(\log^* n)$ supernode rounds, and a Cole--Vishkin-style reduction~\cite{CV86} finishes at three colors; neither step needs an orientation of the block cycle); each phase costs $O(w)$. 

Properness is guaranteed within blocks by Lemma~\ref{lemma:balanced-path}. Across each block boundary, the strictly-later-colored endpoint avoids its colored outside neighbor via $F_1/F_2$, and adjacent blocks are in different phases. 

For the balance bounds, note that blocks partition $V$ and each block's frequencies deviate from $|B_j|/3$ by at most $2$, so the total deviation is at most $2\,(\#\mathrm{blocks}) \le 2n/(w+1)$.

For the claim restatement, let $3 \le \imb < n/3$ and set $T = \lceil (6 n/\imb)\, \const_2\log^* n\rceil$. If $T \le n$, run $\BalBlocks(\cC_n, T)$: then $w + 1 \ge T/(\const_2\log^* n) \ge 6n/\imb$, so the imbalance is at most $2n/(w+1) \le \imb/3 \le \imb$. Otherwise $\imb \le 6\const_2\log^* n$, and we may instead let every vertex collect the entire labeled cycle in $n$ rounds and output its color in a fixed canonical proper $3$-coloring of $\cC_n$ with all class sizes within $1$ of $n/3$ (such a coloring exists for every $n \ge 3$, and $1 \le \imb$); in this case $n < T = O((n/\imb)\log^* n)$ as well. The theorem follows.
\end{proof}

\begin{proof}[Proof of Theorem~\ref{thm:intro-cycle-ub}(c)]
Run $\BalBlocks$ with $T = \Theta(\eta^{-1}\log^* n)$. Its imbalance $O((n/T)\log^* n)$ is at most $\eta n/3$. The matching lower bound is the classical $\Omega(\log^* n)$ bound for proper $3$-coloring of cycles~\cite{Linial92,Naor91}.
\end{proof}


We remark that both algorithms run without change on paths, with endpoints acting as permanent cuts and anchors. As $\imb$ decreases to $O(1)$, both bounds degrade to $\tilde O(n)$ rounds; Sections~\ref{sec:det-lb} and~\ref{sec:diameter-lb} showed that this degradation is necessary.

\subsection{A barrier: faster anchors cannot close the gap}
\label{sec:ruling-barrier}

The deterministic tradeoff now stands at $\Omega(n/\imb)$ versus $O((n/\imb)\log^* n)$, and both algorithms of this section spend their $\log^* n$ factor on computing the distance-$w$ ruling set of Fact~\ref{fact:ruling}. It is natural to ask whether a faster anchor construction could close the gap. We next show that it cannot.

\begin{proposition}
\label{prop:ruling-lb}
There is a constant $c > 0$ such that the following holds for every $n_0 \ge 16$ and every $w$ with $1 \le w \le n_0^{1/2}$.\footnote{Some largeness assumption is necessary: on $\cC_6$ with $w = 2$ and identifier universe exactly $[6]$, the $0$-round rule ``select if and only if the identifier equals $1$'' outputs on every assignment a single selected vertex, whose unique gap is $6 \in [3, 6]$.} Let $\cA$ be a deterministic $R$-round $\LOCAL$ algorithm that outputs, on $\cC_{n_0}$ with identifiers from a universe $[U_0]$ polynomial in $n_0$, on every assignment, a nonempty vertex set whose consecutive gaps all lie in $[w+1,\, 2w+2]$. Then:
\\
(a) $R \ge \max\{1,\ \lfloor w/2\rfloor\}$;
\\
(b) if, in addition, $U_0 \ge 2n_0$, then $R \ge c\, w \log^* n_0$.
\end{proposition}

Part (b) is proved in two stages. We first give the reduction in its cleanest setting, assuming that the block size $\lceil (w+1)/4\rceil$ divides $n_0$ (Definition~\ref{def:anchor-embed}--Lemma~\ref{lem:anchor-complete}); we then remove the divisibility assumption by re-running the reduction from MIS on \emph{paths} (Lemma~\ref{lem:anchor-paths}--Definition~\ref{def:anchor-path-embed}--Lemma~\ref{lem:anchor-path-complete}), whose endpoints supply, for free, the symmetry breaking that a cycle cannot. The universe hypothesis $U_0 \ge 2n_0$ of part (b), stated in the explicit style of the universe hypotheses of Theorems~\ref{thm:det-exact} and~\ref{thm:det-stab}, is what the second stage's identifier bookkeeping consumes; see Remark~\ref{rem:anchor-scope}(b).

Throughout the proof, $n_0 \le U_0 \le n_0^{O(1)}$. Since $w \le n_0^{1/2}$ and $n_0 \ge 16$, we have $2w + 2 \le 2\sqrt{n_0} + 2 < n_0$; consequently, on every assignment the output contains at least \emph{two} selected vertices, since a single selected vertex would have its unique gap equal to $n_0 > 2w+2$. As in the proof of Corollary~\ref{cor:det-exact}, we repeatedly restrict attention to instances whose port numbering is consistently oriented; the algorithm must in particular be correct on these, and on them the output bit at a vertex $v$ is a fixed function $F$ of the \emph{oriented} radius-$R$ identifier window of $v$.

\begin{proof}[Proof of Proposition~\ref{prop:ruling-lb}(a)]
First, $R \ge 1$. Suppose $R = 0$, and let $A_{\mathrm{sel}} \subseteq [U_0]$ be the set of identifier values on which the ($0$-round) rule selects. If $|A_{\mathrm{sel}}| \ge 2$, an assignment placing two selected values on adjacent vertices produces a consecutive gap of $1 < w+1$. If $A_{\mathrm{sel}} = \emptyset$, the output is empty. If $A_{\mathrm{sel}} = \{x\}$, then either some assignment avoids $x$ entirely, and its output is empty, or $U_0 = n_0$ and every assignment uses $x$, so the output is always a single vertex, of gap $n_0 > 2w+2$. Each case contradicts the specification.

Next, $R \ge \lfloor w/2\rfloor$; here we may assume $w \ge 2$. Suppose, towards a contradiction, that $R \le (w-2)/2$, i.e., $2R + 2 \le w$. Fix any assignment and, on it, two consecutive selected vertices; their distance is at least $w + 1 \ge 2R+3$, so their radius-$R$ windows occupy disjoint position sets and hence carry disjoint identifier sets. Denote these oriented windows by $W_1, W_2$; by construction, $F(W_1) = F(W_2) = 1$. Now build a \emph{new} assignment: place $W_1$ on positions $0, \dots, 2R$ and $W_2$ on positions $2R+2, \dots, 4R+2$ (legal by disjointness, and it fits, as $4R + 3 \le 2w - 1 < n_0$), and fill the remaining positions with distinct fresh identifiers. On this assignment the vertices at positions $R$ and $3R+2$ are both selected and at distance $2R+2 \le w$, so some consecutive gap is at most $2R+2 < w+1$, a contradiction. Hence $R > (w-2)/2$, i.e., $R \ge \lceil (w-1)/2\rceil = \lfloor w/2\rfloor$.
\end{proof}

\paragraph{The block cycle.}
For the first stage of part (b), set $b = \lceil (w+1)/4\rceil$, so that $(w+1)/4 \le b \le (w+4)/4$, and assume $b \mid n_0$; let $n = n_0/b$. Since $b \le (\sqrt{n_0}+4)/4 \le \sqrt{n_0}/2$ for $n_0 \ge 16$, we have $n \ge 2\sqrt{n_0}$; in particular $n_0 \le n^2 \le 2^n$, so $\log^* n \ge \log^* n_0 - 1$. Let $U' = \lfloor U_0/b\rfloor$; then $U' \ge \lfloor n_0/b\rfloor = n$ and $U' \le U_0 \le n^{O(1)}$, so $[U']$ is a legal polynomial identifier universe for $n$-vertex instances.

Given an instance $\nu$ of $\cC_n$ (an assignment of distinct identifiers from $[U']$), we define a virtual instance $\Phi(\nu)$ of $\cC_{n_0}$, to be simulated by the vertices of $\cC_n$. The construction must overcome the fact that $\cC_n$ carries neither coordinates nor a globally consistent orientation, while all simulating vertices must reconstruct fragments of \emph{one and the same} virtual instance; every choice below is therefore made canonical in the identifiers. (Identifier-dependent choices are safe here, in contrast with the reduction of Theorem~\ref{thm:diam-lb}, whose footnote rules them out: there, the view of a simulated vertex had to be determined by a short identifier window, whereas here the simulator simply learns every identifier its computation depends on.)

\begin{definition}
\label{def:anchor-embed}
Fix, for the purpose of this definition only, one of the two cyclic orientations of $\cC_n$, writing its vertices $v_0, \dots, v_{n-1}$ in cyclic order, and identify $V(\cC_{n_0})$ with the positions $0, \dots, n_0 - 1$ in cyclic order. \emph{Block} $i$ is the position set $\{ib, \dots, ib+b-1\}$, owned by $v_i$. In $\Phi(\nu)$:
(i) block $i$ carries the $b$ virtual identifiers $b(\nu(v_i)-1)+1, \dots, b(\nu(v_i)-1)+b$; distinct blocks receive disjoint ranges, and all values lie in $[\,b\,U'\,] \subseteq [U_0]$;
(ii) within block $i$, the virtual identifiers are written in increasing order \emph{towards the one of the two neighbors of $v_i$ whose identifier is larger} (for $b = 1$ the rule is vacuous);
(iii) at every virtual vertex, port $1$ leads to the one of its two neighbors whose virtual identifier is smaller, and port $2$ to the other one.
\end{definition}

\begin{lemma}
\label{lem:anchor-sim}
Let $q = \lceil R/b\rceil$.
\\
(a) $\Phi(\nu)$ is a well-defined identified, ported cycle of length $n_0$, independent of the orientation fixed in Definition~\ref{def:anchor-embed}, and a legal input of $\cA$.
\\
(b) The output of $\cA$ at a virtual vertex $u$ is a function of the virtual identifiers at virtual distance at most $R+1$ from $u$, read as an unoriented rooted labeled path.
\\
(c) After $q+2$ rounds on $\cC_n$, every vertex can compute $\cA$'s output at every virtual vertex of its own block, all with respect to the single common instance $\Phi(\nu)$.
\end{lemma}

\begin{proof}
(a) Reversing the fixed orientation reverses the cyclic order of the blocks and swaps the names of the two neighbors of each $v_i$. The layout rule (ii) places the largest entry of block $i$ at the end facing the neighbor with the larger identifier --- a condition on the physical neighbor, not on its name --- so each block's content occupies the same physical positions either way, and the resulting cyclic identifier sequence is the exact reversal of the original one; two mutually reversed cyclic sequences present the same unoriented identified cycle. The ports (iii) refer only to identifier comparisons, hence are presentation-independent. Finally, $\cA$ must be correct for every legal port numbering, so it is correct on $\Phi(\nu)$.

(b) After $R$ rounds, the state of $u$ is determined by its rooted radius-$R$ view: the identifiers, the incidence structure, and the port labels of all vertices at distance at most $R$ (a vertex's port labels are part of its initial state, which reaches $u$ within $R$ rounds). Under (iii), the port labels at a vertex $x$ are determined by the identifiers of $x$'s two neighbors, which, for $x$ at distance exactly $R$ from $u$, include identifiers at distance $R+1$; conversely, the identifiers within distance $R+1$ determine all identifiers and all port labels within distance $R$, hence the entire view. The view is a rooted unoriented object, so every party reconstructing it from these identifiers obtains the same output value.

(c) In $q + 2$ rounds, $v_i$ learns $\nu$ at all vertices within distance $q+2$, which determines the contents (from $\nu(v_j)$) and the layouts (from $\nu(v_{j-1}), \nu(v_{j+1})$) of all blocks $j$ with $|j - i| \le q+1$. A walk of length $R + 1$ starting in block $i$ crosses at most $q+1$ block boundaries, since crossing $q+2$ of them requires more than $(q+1)\,b \ge R+1$ edges; so every position within virtual distance $R+1$ of a virtual vertex of block $i$ lies in these known blocks. By (a) and (b), $v_i$ can therefore evaluate $\cA$'s output at each virtual vertex of its block, and the value does not depend on which of its two local directions $v_i$ calls clockwise: in either reading it reconstructs the same unoriented rooted view within the same instance $\Phi(\nu)$.
\end{proof}

Let $S_0 \subseteq \{0, \dots, n_0 - 1\}$ denote $\cA$'s output on $\Phi(\nu)$; by correctness, $S_0$ is nonempty with consecutive gaps in $[w+1, 2w+2]$, and $|S_0| \ge 2$ as observed above. Define the \emph{projection}
\[
P \;=\; \{\, i \in \mathbb{Z}_n \;:\; \mbox{block } i \mbox{ contains an element of } S_0 \,\} \;\subseteq\; V(\cC_n);
\]
by Lemma~\ref{lem:anchor-sim}(c), after $q+2$ rounds every vertex of $\cC_n$ knows its own membership in $P$.

\begin{lemma}
\label{lem:anchor-project}
(a) Every block contains at most one element of $S_0$. Hence $|P| = |S_0| \ge 2$, and consecutive elements of $S_0$ lie in consecutive elements of $P$, whose $P$-gap equals the block distance between them.
\\
(b) Every consecutive gap $d$ of $P$ on $\cC_n$ satisfies $2 \le d \le 8$, and $d \ge 3$ whenever $w \ge 5$. In particular, $P$ is a nonempty independent set of $\cC_n$ with constant-length gaps.
\end{lemma}

\begin{proof}
Consider two elements of $S_0$, lying in blocks $i$ and $j$, and an arc between them of length $g$ (edges) whose direction passes through $d \ge 0$ block boundaries. Reading positions along the arc, the first element lies in $\{ib, \dots, ib+b-1\}$ and the second in $\{(i+d)b, \dots, (i+d)b+b-1\}$, whence the \emph{sandwich bounds}
\[
(d-1)\,b + 1 \;\le\; g \;\le\; (d+1)\,b - 1 .
\]

(a) If two elements of $S_0$ shared a block, the case $d = 0$ of the sandwich bounds would give an arc of length at most $b - 1 \le w/4$ between them, inside which some consecutive gap of $S_0$ is at most $b - 1 < w + 1$, a contradiction. For the second claim, let $u, u'$ be consecutive elements of $S_0$, in blocks $i \ne j$: every position of the intermediate blocks $i+1, \dots, j-1$ (along the arc from $u$ to $u'$) lies strictly between $u$ and $u'$, since block $i+1$ begins after $u$ and block $j-1$ ends before $u'$; so the intermediate blocks are unselected, $i$ and $j$ are consecutive in $P$, and their $P$-gap is $d$.

(b) Now let $g \in [w+1,\, 2w+2]$ be the gap between two consecutive elements of $S_0$ and $d$ the corresponding $P$-gap. By the sandwich bounds and $(w+1)/4 \le b \le (w+4)/4$,
\[
d \;\ge\; \frac{g+1}{b} - 1 \;\ge\; \frac{4(w+2)}{w+4} - 1 \;=\; \frac{3w+4}{w+4} \;>\; 1\]
and
\[d \;\le\; \frac{g-1}{b} + 1 \;\le\; \frac{4(2w+1)}{w+1} + 1 \;=\; 9 - \frac{4}{w+1} \;<\; 9 ,
\]
so $2 \le d \le 8$. Finally, $(3w+4)/(w+4) > 2$ if and only if $w > 4$, so $d \ge 3$ for every $w \ge 5$. (For $w \le 4$ the bound $d \ge 2$ is attained; independence of $P$ is all that is used below.)
\end{proof}

\begin{lemma}
\label{lem:anchor-complete}
Nine further rounds on $\cC_n$ suffice for every vertex to compute its membership in a set $I \supseteq P$ that is a maximal independent set of $\cC_n$. The completion rule is deterministic and canonical in the identifiers; no globally consistent orientation and no additional symmetry breaking are used.
\end{lemma}

\begin{proof}
For nine rounds, every vertex forwards all (identifier, $P$-bit) pairs it has received; afterwards, every vertex knows both values at every vertex within distance $9$. Since every $P$-gap has length at most $8$ (Lemma~\ref{lem:anchor-project}(b)), every vertex now sees, in full, the maximal $P$-free arc containing it (if any), together with the two bounding $P$-vertices and their identifiers.

Consider a maximal $P$-free arc, with bounding vertices $z_1, z_2 \in P$ and interior vertices $x_1, \dots, x_m$ ($1 \le m \le 7$), indexed by distance from the bounding vertex $z \in \{z_1, z_2\}$ with the \emph{smaller} identifier. This origin is well defined (identifiers are distinct), the indexing uses only the unoriented arc structure, and if $|P| = 2$ the two arcs between the same pair of $P$-vertices are processed separately, each with its own origin. Select
$T = \{ x_t \mid t \mbox{ even},\ 2 \le t \le m - 1 \}$,
or equivalently, process the interior greedily from $z$, skipping every vertex adjacent to $P$ or to a previously selected vertex --- and let $I = P \cup \bigcup T$, the union over all maximal $P$-free arcs. Every vertex of an arc sees the same data and computes the same $T$, so the output is globally consistent, and each vertex decides its own membership with no further communication.

$I$ is independent: $P$ is independent by Lemma~\ref{lem:anchor-project}(b); within an arc, the selected vertices are pairwise at distance at least $2$; every selected $x_t$ has $2 \le t \le m-1$, hence is adjacent to no $P$-vertex; and selected vertices of \emph{different} arcs are separated by a $P$-vertex lying at distance at least $2$ from each, hence are at distance at least $4$ from one another. $I$ is dominating: $x_1$ is adjacent to $z$; an unselected $x_t$ with $t$ odd, $t \ge 3$, is adjacent to $x_{t-1} \in T$ (as $t-1$ is even and $2 \le t - 1 \le m-1$); and an unselected $x_t$ with $t$ even can only have $t = m$, and is then adjacent to the far bounding vertex. Hence $I$ is a maximal independent set of $\cC_n$.
\end{proof}

\begin{proof}[Proof of Proposition~\ref{prop:ruling-lb}(b) when $\lceil (w+1)/4\rceil$ divides $n_0$]
The pieces assemble into a deterministic $\LOCAL$ algorithm $\mathcal{B}$ on $\cC_n$: simulate $\cA$ on $\Phi(\nu)$ and project (Lemmas~\ref{lem:anchor-sim} and~\ref{lem:anchor-project}), then complete (Lemma~\ref{lem:anchor-complete}). For every instance $\nu$ with distinct identifiers from $[U']$, the virtual instance $\Phi(\nu)$ is a legal input of $\cA$, so $\cA$'s gap guarantee holds surely, and $\mathcal{B}$ outputs a maximal independent set of $\cC_n$, within $T_{\mathcal{B}} = (q + 2) + 9 \le R/b + 12$ rounds. Linial's lower bound~\cite{Linial92} provides a constant $c_L > 0$ and a threshold $n_L$ such that every deterministic MIS algorithm on $\cC_n$, correct for every assignment of distinct identifiers from a universe of size at least $n$, requires at least $c_L \log^* n$ rounds, for all $n \ge n_L$ (an algorithm correct for the universe $[U']$ is in particular correct on the instances with identifiers from $[n] \subseteq [U']$); the two-regime assembly in the general-case proof below then yields $R \ge c\, w\log^* n_0$, with the better intermediate constant $R/b + 12 \ge c_L\log^* n$ in place of the general case's $2R/b + 40 \ge \min\{\cdot\}$. As the general case subsumes this one, we do not repeat the assembly here.
\end{proof}

\paragraph{Removing the divisibility assumption.}
Uniform blocks force $b \mid n_0$: the block sizes can depend neither on position (the vertices of $\cC_n$ carry no coordinates) nor on the identifiers (the virtual cycle must have length exactly $n_0$ on \emph{every} instance, while any identifier-dependent size pattern varies across instances). We therefore re-run the reduction from MIS on the \emph{path} $P_n$, whose two endpoints are locally recognizable (degree $1$) and may thus, with no symmetry breaking, carry blocks of a different size that absorb the remainder $n_0 \bmod b$. The price is a \emph{seam}: the virtual cycle closes between the two endpoint blocks by an edge with no physical counterpart, so the vertices near the endpoints cannot simulate $\cA$; they will instead be completed canonically. Figure~\ref{fig:path-embed} shows the layout. We first record the lower bound for MIS on paths, which we derive from Linial's cycle bound rather than cite.

\begin{lemma}
\label{lem:anchor-paths}
Let $c_L, n_L$ be as above, and let $\mathcal{B}$ be a deterministic $t$-round $\LOCAL$ algorithm computing a maximal independent set of the $n$-vertex path $P_n$, correct for every assignment of distinct identifiers from a universe $[U']$ with $U' \ge n$ of size polynomial in $n$, and for every port numbering. Then $t \ge \min\{\,c_L\log^* n,\ (n-8)/4\,\}$ for all $n \ge n_L$.
\end{lemma}

\begin{proof}
If $4t + 8 > n$, then $t > (n-8)/4$ and there is nothing to prove; so assume $n \ge 4t+8$. Run $\mathcal{B}$ verbatim on an arbitrary instance of $\cC_n$ with identifiers from $[U']$ (every vertex of the cycle has degree $2$ and simply executes the uniform code of $\mathcal{B}$). We claim the output is a maximal independent set of $\cC_n$. Suppose not; then some vertex $x$ witnesses a violation --- two adjacent selected vertices at distance at most $1$ from $x$, or $x$ unselected with both neighbors unselected --- and all vertices involved lie within distance $1$ of $x$. Let $\Sigma$ be the cyclic segment of radius $t+1$ around $x$ ($2t+3 \le n$ distinct vertices). Build a \emph{path} instance of $P_n$: place the identifiers of $\Sigma$, with the same port labels, on $2t+3$ consecutive vertices centered in $P_n$ --- at distance at least $\lfloor (n - 2t-3)/2\rfloor \ge t+1$ from both endpoints, using $n \ge 4t+8$ --- and fill the remaining vertices with distinct fresh identifiers from $[U']$ (possible, as $U' \ge n$) and arbitrary ports. Every vertex at distance at most $1$ from the copy of $x$ has the same rooted radius-$t$ view as its original on the cycle: its ball lies inside the copy of $\Sigma$, all of whose vertices are interior (degree $2$), with identical identifiers and port labels. Hence $\mathcal{B}$ outputs the same bits there, reproducing the violation on $P_n$ and contradicting its correctness. So $\mathcal{B}$, run on $\cC_n$, computes an MIS on every instance with identifiers from $[U'] \supseteq [n]$, and Linial's bound gives $t \ge c_L \log^* n$ for $n \ge n_L$.
\end{proof}

\begin{definition}
\label{def:anchor-path-embed}
Let $b$ be the smallest \emph{odd} integer with $b \ge (w+1)/8$; then $(w+1)/8 \le b \le (w+17)/8$, and for every $w \ge 1$ both $2b < w+2$ and $2b < w+3$ hold ($b = 1$ for $w \le 7$). Let $n \in \{\lfloor n_0/b\rfloor, \lfloor n_0/b\rfloor - 1\}$ be chosen so that $\rho = n_0 - nb$ is even (possible, and uniquely so, since $b$ is odd and changing $n$ by one changes $\rho$ by $b$); then $0 \le \rho \le 2b - 2$, and set $s = b + \rho/2 \le 2b-1$. Let $U' = \lfloor U_0/b\rfloor - 2$; by $U_0 \ge 2n_0 \ge 2nb$ we have $U' \ge 2n - 3$, which is at least $n$ (as $n \ge 3$ in all uses below), and $U' \le n^{O(1)}$ as before.

Given an instance $\nu$ of $P_n$ (distinct identifiers from $[U']$), write the path as $v_1, \dots, v_n$ (for the definition only; the labeling of the two traversal directions is immaterial by the reversal argument of Lemma~\ref{lem:anchor-sim}(a), all rules below being direction-free). The endpoint vertices $v_1, v_n$ own blocks of $s$ consecutive virtual positions each; the interior vertices own $b$ each; the blocks are laid consecutively along the path, for a total of $(n-2)b + 2s = nb + \rho = n_0$ positions, and the virtual cycle $\cC_{n_0}$ is closed by the \emph{seam}, the edge joining the outer end of $v_n$'s block to the outer end of $v_1$'s block. In $\Phi(\nu)$:
(i) interior $v_i$ carries the identifiers $b(\nu(v_i)-1)+1, \dots, b\,\nu(v_i)$; each endpoint carries the same range for the $b$ positions nearest its unique neighbor, together with $\rho/2$ \emph{reserved} identifiers for its remaining positions: the endpoint with the smaller $\nu$-identifier takes $bU'+1, \dots, bU'+\rho/2$, and the other takes $bU'+b, \dots, bU'+b+\rho/2-1$; these two ranges are disjoint from each other (as $\rho/2 \le b-1$) and from all interior ranges, and all values lie in $[\,b(U'+2)\,] \subseteq [U_0]$;
(ii) within an interior block, the identifiers are written in increasing order towards the one of the two neighbors of $v_i$ whose identifier is larger; within an endpoint block, in increasing order towards the unique neighbor;
(iii) at every virtual vertex, port $1$ leads to the one of its two neighbors whose virtual identifier is smaller, and port $2$ to the other one.
The reserved identifiers of the endpoint blocks involve a global comparison (which endpoint holds the smaller $\nu$-identifier) and are \emph{not} locally computable; this is legitimate because they are never decoded by any simulator: the simulation below stops short of the endpoints, and the endpoint blocks need only \emph{exist} for $\Phi(\nu)$ to be a well-defined instance of $\cA$.
\end{definition}

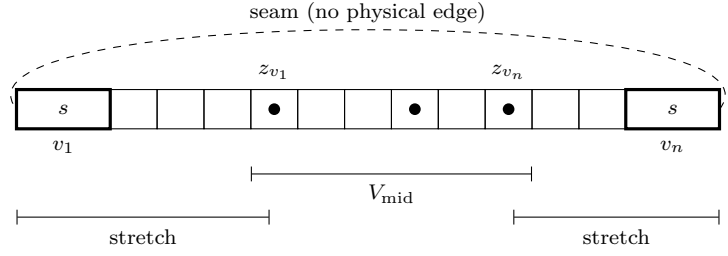
\begin{figure}[t]
\centering
\begin{tikzpicture}[x=0.62cm, y=0.52cm]
\draw[very thick] (0,0) rectangle (2,1);
\node at (1,0.5) {\scriptsize $s$};
\node at (1,-0.45) {\scriptsize $v_1$};
\foreach \i in {2,3,...,12} { \draw (\i,0) rectangle (\i+1,1); }
\draw[very thick] (13,0) rectangle (15,1);
\node at (14,0.5) {\scriptsize $s$};
\node at (14,-0.45) {\scriptsize $v_n$};
\draw[dashed] (15,0.5) .. controls (16.8,3.2) and (-1.8,3.2) .. (0,0.5);
\node at (7.5,2.95) {\scriptsize seam (no physical edge)};
\draw[|-|] (5,-1.15) -- (11,-1.15);
\node at (8,-1.65) {\scriptsize $\Vmid$};
\fill (5.5,0.5) circle (2.2pt);
\fill (8.5,0.5) circle (2.2pt);
\fill (10.5,0.5) circle (2.2pt);
\node at (5.5,1.45) {\scriptsize $z_{v_1}$};
\node at (10.5,1.45) {\scriptsize $z_{v_n}$};
\draw[|-|] (0,-2.25) -- (5.4,-2.25);
\node at (2.7,-2.75) {\scriptsize stretch};
\draw[|-|] (10.6,-2.25) -- (15,-2.25);
\node at (12.8,-2.75) {\scriptsize stretch};
\end{tikzpicture}
\caption{The path block embedding of Definition~\ref{def:anchor-path-embed}. The two enlarged endpoint blocks (of size $s = b + \rho/2$ each, drawn thick) absorb the remainder $n_0 \bmod b$, and the virtual cycle $\cC_{n_0}$ is closed by the \emph{seam}, which has no physical counterpart in $P_n$. Only the vertices of the middle region $\Vmid$, at distance at least $q+2$ from both endpoints, simulate $\cA$; the dots mark the projected selected blocks $P$, with $z_{v_1}, z_{v_n}$ the elements of $P$ nearest the two endpoints, and the two \emph{endpoint stretches} are completed by the one-sided rule of Lemma~\ref{lem:anchor-path-complete}.}
\label{fig:path-embed}
\end{figure}

\begin{lemma}
\label{lem:anchor-path-sim}
Let $q = \lceil R/b\rceil$ and let $\Vmid = \{v_i : q+3 \le i \le n-q-2\}$ (the \emph{middle region}, consisting of the vertices at distance at least $q+2$ from both endpoints).
\\
(a) $\Phi(\nu)$ is a well-defined identified, ported cycle of length $n_0$ with distinct identifiers from $[U_0]$, independent of the direction labeling, and a legal input of $\cA$; and the output of $\cA$ at a virtual vertex is a function of the virtual identifiers within distance $R+1$ of it, exactly as in Lemma~\ref{lem:anchor-sim}(a,b).
\\
(b) After $q+2$ rounds on $P_n$, every vertex of $\Vmid$ has computed $\cA$'s output at every virtual vertex of its own block, all with respect to the single common instance $\Phi(\nu)$; in particular, writing $S_0$ for $\cA$'s output on $\Phi(\nu)$ and $P = \{v_i \in \Vmid : \mbox{block } i \mbox{ meets } S_0\}$, every vertex of $\Vmid$ knows its own membership in $P$.
\\
(c) Every block of $\Phi(\nu)$ contains at most one element of $S_0$. Moreover:
(i) any two elements of $P$ are at path-distance at least $2$, so $P$ is independent;
(ii) if two elements of $P$ have no element of $P$ strictly between them on the path, their path-distance is at most $16$; and
(iii) every $17$ consecutive blocks owned by vertices of $\Vmid$ include one that meets $S_0$.
\end{lemma}

\begin{proof}
(a) The direction-freeness is verified as in Lemma~\ref{lem:anchor-sim}(a): reversing the traversal labeling swaps the endpoint names, and every rule of Definition~\ref{def:anchor-path-embed} --- the layout rules, the ports, and the assignment of the reserved ranges by the comparison of the two endpoints' $\nu$-identifiers --- refers to physical vertices and identifier comparisons only. The identifiers are distinct and lie in $[U_0]$, the lengths sum to $n_0$, and $\cA$ is correct for every port numbering, in particular the identifier-determined one. The locality statement is Lemma~\ref{lem:anchor-sim}(b) verbatim (its proof did not use the block structure).

(b) In $q+2$ rounds, $v_i \in \Vmid$ learns $\nu$ at all vertices within distance $q+2$; these all exist, by the definition of $\Vmid$. This determines the contents and layouts of all \emph{interior} blocks $j$ with $|j-i| \le q+1$: contents from $\nu(v_j)$, layouts from $\nu(v_{j\pm 1})$ (which may include an endpoint's $\nu$-identifier --- visible --- but never an endpoint's block content). Since $v_i \in \Vmid$, all blocks $j$ with $|j - i| \le q+1$ are interior. A walk of length $R+1$ from a virtual vertex of block $i$ crosses at most $q+1$ block boundaries (as in Lemma~\ref{lem:anchor-sim}(c)), so it stays within these known interior blocks; in particular it does not reach an endpoint block and does not cross the seam, which lies between the two endpoint blocks. By (a), $v_i$ evaluates $\cA$'s output at each virtual vertex of its block, consistently in either reading.

(c) One element per block: two elements of $S_0$ in one block would be at arc distance at most $s - 1 \le 2b-2 < w+1$ (using $2b < w+3$), inside which some consecutive gap of $S_0$ is smaller than $w+1$, a contradiction; for interior blocks the bound is $b-1$, smaller still.
\\
(i) If two selected blocks were adjacent on the path, the sandwich bounds of Lemma~\ref{lem:anchor-project} with $d = 1$ --- legitimate, since the arc between the two anchors runs through the two adjacent interior blocks and does not approach the seam --- would put their two anchors at arc distance at most $2b - 1 < w+1$, using $2b < w+2$; inside that arc, some consecutive gap of $S_0$ is smaller than $w+1$, a contradiction.
\\
(ii) Let $v_i, v_j \in P$, $i < j$, with no element of $P$ strictly between them on the path. All blocks strictly between them are owned by vertices of $\Vmid$ (whose indices form a contiguous interval), and none of them meets $S_0$; hence the arc of the virtual cycle from the anchor in block $i$ to the anchor in block $j$ that runs along the interior of the path contains no further anchors, so its length $g$ is a consecutive gap of $S_0$ and satisfies $g \le 2w+2$. This arc crosses interior blocks only, all of size $b$, so the sandwich bounds give $g \ge (d-1)b + 1$ with $d = j - i$, whence $d \le (g-1)/b + 1 \le (2w+1)\cdot 8/(w+1) + 1 = 17 - 8/(w+1) < 17$, i.e., $d \le 16$.
\\
(iii) $17$ consecutive blocks of size $b$ span $17b \ge 17(w+1)/8 > 2w+2$ consecutive virtual positions --- more than the maximal gap of $S_0$ --- hence contain an element of $S_0$.
\end{proof}

\begin{lemma}
\label{lem:anchor-path-complete}
Within $q + 36$ further rounds, every vertex of $P_n$ computes its membership in a set $I$ that is a maximal independent set of $P_n$. The rules are deterministic and canonical in the identifiers; no orientation and no additional symmetry breaking are used. Consequently, the whole procedure is a deterministic MIS algorithm $\mathcal{B}$ on $P_n$, correct for every $[U']$-instance, with round complexity $T_{\mathcal{B}} \le 2R/b + 40$.
\end{lemma}

\begin{proof}
For $q+36$ rounds, every vertex forwards all pairs (identifier, $P$-bit if computed). Since information travels one hop per round from the start of the procedure, afterwards every vertex knows the identifier of every vertex within distance $(q+2) + (q+36) = 2q+38$, and the $P$-bit of every $\Vmid$-vertex within distance $q+36$. Every vertex can also tell, for each vertex it sees, whether it belongs to $\Vmid$ (membership in $\Vmid$ is determined by the distance to the endpoints, and any vertex within distance $q+36$ of an endpoint sees that endpoint).

\midinline The short case:
If $n \le 2q+37$, every vertex sees the entire path, with all identifiers, and outputs its membership in the canonical maximal independent set of $P_n$ consisting of the vertices at even distance from the endpoint with the smaller identifier; all vertices compute the same set, and no orientation is used. Assume from now on $n \ge 2q+38$, so that $|\Vmid| = n - 2q - 4 \ge 34$.

\midinline The extreme projected anchors, and the partition:
By Lemma~\ref{lem:anchor-path-sim}(c)(iii), the first $17$ and the last $17$ blocks owned by vertices of $\Vmid$ --- disjoint families, as $|\Vmid| \ge 34$ --- each contain a block meeting $S_0$. Hence, for each endpoint $e \in \{v_1, v_n\}$, the element $z_e$ of $P$ nearest to $e$ exists and lies within distance $(q+2) + 16$ of $e$, and $z_{v_1} \ne z_{v_n}$; in particular, $|P| \ge 2$. Every vertex of $P_n$ not in $P$ therefore lies either \emph{strictly between two consecutive elements of $P$} (a \emph{two-sided arc}, of length at most $16$ by Lemma~\ref{lem:anchor-path-sim}(c)(ii)) or \emph{between $z_e$ and $e$, inclusive of $e$} (an \emph{endpoint stretch}, of at most $q + 18$ vertices). Each vertex determines its category, and all data relevant to its arc or stretch, from its knowledge: a stretch vertex sees $e$, sees $z_e$, and sees the $P$-bits of all $\Vmid$-vertices between itself and $z_e$, confirming minimality; a two-sided-arc vertex sees both bounding $P$-vertices within distance $16$.

\midinline The completion rules:
The output is $I = P \cup \bigcup T$, where the selections $T$ are made per arc and per stretch:
\emph{(two-sided arcs)} exactly the rule of Lemma~\ref{lem:anchor-complete}: index the interior $x_1, \dots, x_m$ ($1 \le m \le 15$, by Lemma~\ref{lem:anchor-path-sim}(c)(i,ii)) from the bounding $P$-vertex with the smaller identifier and select $T = \{x_t : t \mbox{ even},\ 2 \le t \le m-1\}$;
\emph{(endpoint stretches)} index the stretch $x_1, \dots, x_L$ from $z_e$ (so $x_L = e$; this origin needs no identifier comparison, one end of the stretch being the degree-one vertex $e$ itself) and select $T = \{x_t : t \mbox{ even},\ 2 \le t \le L\}$ --- one-sided greedy from $z_e$, with no constraint at the far end, since the path terminates there.

$I$ is independent: within arcs and stretches, selected vertices are pairwise at distance $\ge 2$ and start at $x_2$, hence are adjacent to no $P$-vertex; selections of different arcs/stretches are separated by a $P$-vertex at distance $\ge 2$ from each. $I$ is dominating: in a two-sided arc, exactly as in Lemma~\ref{lem:anchor-complete}; in a stretch, $x_1$ is adjacent to $z_e$, an unselected $x_t$ with $t$ odd and $t \ge 3$ is adjacent to $x_{t-1} \in T$, and an unselected even index cannot occur among $t \le L$; in particular $e = x_L$ is selected if $L$ is even and adjacent to $x_{L-1} \in T$ if $L$ is odd ($L \ge 2$; if $L = 1$, $e = x_1$ is adjacent to $z_e$). Every vertex of $P_n$ is in $P$, in a two-sided arc, or in a stretch, so $I$ is a maximal independent set of $P_n$.

All rules --- including the short-case rule --- are functions of identifiers, $P$-bits and distances visible within the knowledge radii recorded at the start of the proof, computed identically by all vertices concerned; no orientation enters (arcs are indexed from an identifier comparison, stretches and the short-case set from the locally recognizable endpoints). The total round complexity is $(q+2) + (q+36) = 2q + 38 \le 2R/b + 40$.
\end{proof}

\begin{proof}[Proof of Proposition~\ref{prop:ruling-lb}(b), general case]
Let $\mathcal{B}$ be the algorithm of Lemma~\ref{lem:anchor-path-complete}; it computes an MIS of $P_n$ on every instance with identifiers from $[U']$, $n \le U' \le n^{O(1)}$, in $T_{\mathcal{B}} \le 2R/b + 40$ rounds. Set
\[
\Lambda = \max\Big\{ \frac{80}{c_L} + 2,\ \log^*\!\big(n_L^2\big) + 1,\ 7 \Big\}
\qquad\mbox{and}\qquad
c = \min\Big\{ \frac{c_L}{32},\ \frac{1}{4\Lambda} \Big\} .
\]
If $\log^* n_0 < \Lambda$, part (a) gives $R \ge \max\{1, \lfloor w/2\rfloor\} \ge w/4 \ge w \log^* n_0/(4\Lambda) \ge c\,w\log^* n_0$. If $\log^* n_0 \ge \Lambda$, then $n_0 \ge \max\{289,\ n_L^2\}$, so $n \ge n_0/b - 2 \ge 8n_0/(\sqrt{n_0}+17) - 2 \ge 2\sqrt{n_0} \ge n_L$, and $n_0 \le n^2 \le 2^n$ gives $\log^* n \ge \log^* n_0 - 1$. Lemma~\ref{lem:anchor-paths} yields $T_{\mathcal{B}} \ge \min\{c_L\log^* n,\ (n-8)/4\}$. In the first branch,
\[
\frac{2R}{b} + 40 \;\ge\; c_L\log^* n \;\ge\; c_L\big(\log^* n_0 - 1\big)
\qquad\mbox{and}\qquad
c_L\big(\log^* n_0 - 1\big) - 40 \;\ge\; \frac{c_L}{2}\log^* n_0
\]
(the latter because $\log^* n_0 \ge 80/c_L + 2$), whence
$R \ge \frac{b}{2}\cdot\frac{c_L}{2}\log^* n_0 \ge \frac{w+1}{16}\cdot\frac{c_L}{2}\log^* n_0 \ge c\, w\log^* n_0$.
In the second branch, $2R/b + 40 \ge (n-8)/4$ gives $R \ge b(n-8)/8 - 20b \ge (n_0 - \rho)/8 - 21b \ge n_0/16 \ge w\log^* n_0 \ge c\,w\log^* n_0$, using $\rho \le 2b - 2$, $b \le \sqrt{n_0}$, and, in the last two steps, that $\log^* n_0 \ge \Lambda$ makes $n_0$ large enough that $n_0/16 \ge \sqrt{n_0}\,\log^* n_0 \ge w \log^* n_0$.
\end{proof}

\begin{remark}
\label{rem:anchor-scope}
(a) Both ends of the projected-gap ranges are attained ($[2,8]$ in the first stage, $[2,16]$ in the second); only the independence of $P$ and the constant upper bounds are used, and none of the constants was optimized.
\\
(b) Under the divisibility assumption $\lceil (w+1)/4\rceil \mid n_0$, the first stage proves part (b) already for every universe with $U_0 \ge n_0$: the uniform blocks tile $[b\,U']$ with $U' = \lfloor U_0/b\rfloor \ge n$ exactly. The requirement $U_0 \ge 2n_0$ of the general case pays for the reserved identifier ranges of the endpoint blocks and is an artifact of this bookkeeping; we did not attempt to optimize it. Note that some hypothesis on the universe is unavoidable in this context: as Theorem~\ref{thm:universe} illustrates in the rigidity setting, lower-bound statements on identified cycles can genuinely change character at universe size exactly $n$.
\end{remark}

Thus Fact~\ref{fact:ruling} is optimal --- over polynomial identifier universes of size at least $2n$, the regime of Proposition~\ref{prop:ruling-lb}(b) --- and closing the remaining $\log^* n$ gap for balanced coloring itself would require either an algorithm that dispenses with anchors altogether, or a lower bound marrying the rigidity mechanism of Section~\ref{sec:det-lb} with Linial-type indistinguishability. Section~\ref{sec:open} isolates the combinatorial core of this question.

\section{The concentration regime}
\label{sec:concentration}

We now prove Theorem~\ref{thm:intro-general}(a) on general graphs. Vertices output values in $\cP = \{1, \dots, \Delta+1\}$; every vertex always outputs some color, and properness holds w.h.p.

Both algorithms of this section finish with a \emph{completion phase}; we isolate the exact black box used, in the form we need it.

\begin{fact}
\label{fact:completion}
Let $H$ be an $n'$-vertex graph, $n' \le n$, in which every vertex $v$ holds a list $L_v$ of allowed colors from a universe of size $\mathrm{poly}(n)$ with $|L_v| \ge \deg_H(v) + 1$, and assume distinct identifiers from a $\mathrm{poly}(n)$ universe. Then a proper coloring of $H$ from the lists is computable:
\\
(a)~deterministically, in $O(\Delta(H)^2 + \log^* n)$ rounds;
\\
(b)~deterministically, in $O\big(\log^2\Delta(H)\cdot\log n\big)$ rounds~\cite{GK21}; and
\\
(c)~with randomization, for every fixed failure exponent $\gamma$, in $O(\log^3\log n)$ rounds with probability at least $1 - n^{-\gamma}$~\cite{HKNT22}, with no further assumption on $H$, on the lists, or on their relation to any global palette, and with no hidden $\Delta$- or $D$-dependent terms.
\end{fact}

\begin{proof}
We first prove part (a).
Compute an $O(\Delta(H)^2)$-coloring $\psi$ of $H$ in $\log^* n + O(1)$ rounds by Linial's algorithm~\cite{Linial92}, then iterate over the $\psi$-classes: in step $i$, every $v$ with $\psi(v) = i$ takes the smallest color of $L_v$ not held by an already-colored neighbor. Such a color exists because at most $\deg_H(v)$ colors are excluded and $|L_v| \ge \deg_H(v)+1$, and no two adjacent vertices choose in the same step, so the coloring is proper.

Part (b) is the deterministic $(\mathrm{deg}{+}1)$-list-coloring algorithm of Ghaffari and Kuhn~\cite{GK21} invoked directly on $H$; its hypotheses are exactly those stated in the fact. 

Part (c) is Theorem~2 of Halld\'orsson et al.~\cite{HKNT22} (numbering as in the full version of that paper, arXiv:2112.00604), 
which reads: ``There is a randomized distributed algorithm to solve the D1LC problem on $n$-node graphs in $O(\log^3\log n)$ rounds, w.h.p.'' Their \emph{degree$+$1 list coloring} (D1LC) problem is the same as the problem addressed in the current fact, and in the same setting.
The $O(\log^3\log n)$ bound is obtained by completing the shattered components with the deterministic algorithm of~\cite{GK21}, so part (c) improves in lockstep with deterministic $(\mathrm{deg}{+}1)$-list coloring, as noted below.
Three bookkeeping points align their statement with ours. First, the color space: in $\LOCAL$, \cite{HKNT22} needs no bound on it, and the polynomial universe in the hypothesis above is what parts (a) and (b) require; part (c) simply inherits it. Secondly, the size parameter: we invoke the algorithm on $H$ with the ambient parameter $n$ rather than $n' = |V(H)|$. This is legitimate: $H$ together with $n - n'$ virtual isolated nodes, each holding a nonempty list and hence colored in zero rounds, is an $n$-node D1LC instance, and the model of~\cite{HKNT22} explicitly permits the nodes to know $n$ (our vertices know $n$ outright, and the virtual nodes need not physically exist, as they never communicate). Thirdly, the failure probability: with the ambient parameter, the guarantee reads $O(\log^3\log n)$ rounds with failure probability $1/\mathrm{poly}(n)$ in $n$, not in $n'$, and the failure exponent $\gamma$ enters the constants in the standard way, since the probabilistic guarantees of~\cite{HKNT22} are stated for an arbitrary constant exponent while their completion phase is deterministic, hence error-free.
\end{proof}

We use part (c) only as a black box for randomized completion phases; all bounds in this section degrade gracefully, in that replacing $O(\log^3\log n)$ by any time bound $T_{\mathrm{list}}(n) = \mathrm{poly}(\log\log n)$ changes only the additive completion term in the round complexities.

The frequency analysis rests on two properties of the first phase of our algorithms, namely, palette symmetry and small influence. This leads to the following definition.

\begin{definition}
\label{def:symmetric}
A randomized \emph{partial-coloring} algorithm $\cA$ outputs $\clr: V \to \cP \cup \{\bot\}$, where $\bot$ stands for ``uncolored'' (an algorithm that always colors simply never outputs $\bot$). It is \emph{palette-symmetric} if for every permutation $\pi$ of $\cP$ (extended by $\pi(\bot) = \bot$) the law of $\pi\circ\clr$ equals the law of $\clr$. The \emph{influence radius} of $\cA$ is the smallest $\tau$ such that, for every vertex $v$ and any two assignments of the private random strings that differ only in the string $\Rand_v$ held by $v$, all outputs outside $\Gamma_\tau(v)$ coincide. (A $T$-round algorithm always has influence radius at most $T$.)
\end{definition}

We establish the following transfer lemma.

\begin{lemma}
\label{lemma:transfer}
Let $\cA$ be a palette-symmetric partial-coloring algorithm with influence radius $\tau$, let $b = \max_v|\Gamma_\tau(v)|$, and let $\freq(\bot) = |\clr^{-1}(\bot)|$ be the number of uncolored vertices. Then:
\\
(a) 
$\E[\freq(c)] = (n - \E[\freq(\bot)])/(\Delta+1)$ for every $c \in \cP$,
and 
\\
(b) for every $t > 0$, with probability at least $1 - 2(\Delta+2)\exp\big(-2t^2/(n b^2)\big)$, simultaneously 
\\
$|\freq(c) - \E[\freq(c)]| \le t$ for all $c \in \cP$ and $|\freq(\bot) - \E[\freq(\bot)]| \le t$.
In particular, for always-coloring algorithms, meaning algorithms satisfying $\freq(\bot)=0$ in every execution, $\E[\freq(c)] = \sigma$ exactly.
\end{lemma}

\begin{proof}
By palette symmetry the variables $(\freq(c))_{c \in \cP}$ are exchangeable, and they sum to $n - \freq(\bot)$; taking expectations gives claim (a). Each of the $\Delta+2$ random variables $\freq(1), \dots, \freq(\Delta+1), \freq(\bot)$ is a function of the independent random strings $(\Rand_v)_v$, and resampling one string changes at most $b$ outputs, hence it changes each of these random variables by at most $b$; Theorem~\ref{thm:hoeffding}(b) and a union bound finish the proof. 
\end{proof}

\subsection{The algorithm \SymColor\ and the residual-decay lemma}

Algorithm $\SymColor$ operates in two phases. In the \emph{proposal phase}, iterated $\tau_0 = O(\log(\Delta/\eta))$ times, every uncolored vertex tosses a fair coin and, on success, proposes a uniformly random color from its current list, namely, the palette minus the colors already fixed at its neighbors; the proposal is kept, permanently, unless some neighbor proposes the same color in the same round. In the \emph{completion phase}, the graph induced by the uncolored vertices is list-colored by the black box of Fact~\ref{fact:completion}. The proposal phase treats all colors symmetrically, so the transfer lemma applies to it; the completion phase is arbitrary, and is harmless simply because the residue it colors is small.

\begin{algorithm}[t]
\caption{$\SymColor(G, \eta)$ \label{alg:symcolor}}
\KwIn{Graph $G$ with degree bound $\Delta$; $\eta$.}
\KwOut{Proper coloring, palette $\Delta+1$, all class sizes in $(1\pm\eta)\sigma$, w.h.p.}

$\tau_0 \leftarrow \lceil 6\ln(16(\Delta+1)/\eta)\rceil$\;
\emph{Proposal phase.} 

\For{round $r = 1, \dots, \tau_0$}
{
  Each uncolored $v$ becomes active with probability $1/2$\; 
  Each active vertex proposes a uniformly random color $c_v$ from $L_v = \cP \setminus \{ \mbox{colors of }v\mbox{'s colored neighbors} \}$\;
  $v$ keeps $c_v$ permanently (and becomes colored) if no neighbor of $v$ proposes $c_v$ in this round\;
}
\emph{Completion phase.} 

On the graph $H$ induced by uncolored vertices, complete the coloring by $(\deg+1)$-list coloring with the lists $L_v$: deterministically in $O(\Delta^2 + \log^* n)$ rounds or randomized in $O(\log^3\log n)$ rounds w.h.p.\ (Fact~\ref{fact:completion})\;
\KwReturn{$\clr$}
\end{algorithm}

The proposal phase is palette-symmetric: by induction on rounds, applying a palette permutation $\pi$ maps every execution to an equally likely execution with all lists, proposals and kept colors renamed by $\pi$. It is important to note that per-vertex constant success probability is in general \emph{false} for proposal dynamics: a vertex whose uncolored neighbors all have lists of size $2$ succeeds with probability exponentially small in its uncolored degree. The correct statement is the following global residual decay.

\begin{lemma}
\label{lemma:decay}
Condition on the history before a proposal round and let $U$ be the current set of uncolored vertices. Then the expected number of vertices of $U$ colored in this round is at least $|U|/6$. Consequently, after $\tau_0$ rounds, $\E[\freq(\bot)] \le n\,e^{-\tau_0/6} \le \eta\sigma/16$, where $\freq(\bot)$ is the number of vertices left uncolored by the proposal phase.
\end{lemma}

\begin{proof}
Each uncolored $u$ has $|L_u| \ge (\Delta+1) - (\deg(u) - d_U(u)) \ge d_U(u) + 1$, where $d_U(u)$ is its number of uncolored neighbors. Fix $v \in U$ and condition on $v$ proposing $c_v$. Define
$S(v)=\sum_{u\in\Gamma(v)\cap U} 1/|L_u|$.
For each uncolored neighbor $u$, $\Pr[u \mbox{ proposes } c_v] \le \tfrac12 \cdot \mathds{1}[c_v \in L_u]/|L_u| \le 1/(2|L_u|)$, and these events are independent across $u$. Hence, using $1 - x \ge e^{-2x}$ for $x \in [0, \tfrac12]$,
\begin{align*}
\Pr[v \mbox{ is colored}]
  ~\ge~ \frac12 \prod_{u \in \Gamma(v)\cap U}
       \Big(1 - \frac{1}{2|L_u|}\Big)
  ~\ge~ \frac12\, e^{-S(v)}.
\end{align*}
(Colors held by colored neighbors are excluded by $L_v$, so the only conflicts are simultaneous proposals.) Summing $S$ over $U$ and exchanging the order of summation,
\begin{align*}
\sum_{v \in U} S(v)
  ~=~ \sum_{u \in U} \frac{d_U(u)}{|L_u|}
  ~\le~ \sum_{u \in U} \frac{d_U(u)}{d_U(u) + 1}
  ~\le~ |U| .
\end{align*}
By Jensen's inequality, $\sum_{v\in U} e^{-S(v)} \ge |U|\, e^{-\frac{1}{|U|}\sum_v S(v)} \ge |U|\,e^{-1}$, so the expected number of newly colored vertices is at least $\tfrac12 e^{-1}|U| \ge |U|/6$. Iterating, $\E[\freq(\bot)] \le n(1 - 1/6)^{\tau_0} \le n e^{-\tau_0/6} \le n \cdot \frac{\eta}{16(\Delta+1)} = \frac{\eta\sigma}{16}$ by the choice of $\tau_0$.
\end{proof}

\subsection{The main theorem}

We are now ready to establish the main result of this section.

\begin{theorem}[Theorem~\ref{thm:intro-general}(a), formal]
\label{thm:symcolor}
Set $C = 16$. For all sufficiently large $n$, all $\Delta$ with $\Delta + 1 \le 2^{\sqrt{\log n}/C}$ and all $\eta \in [2^{-\sqrt{\log n}/C},\, 1]$, Algorithm $\SymColor(G, \eta)$ runs in $O(\log(\Delta/\eta)) + O(\log^3\log n)$ rounds and w.h.p.\ outputs a proper coloring with palette $\Delta+1$ and $\freq(c) \in [(1-\eta)\sigma, (1+\eta)\sigma]$ for every $c$. For $\Delta = O(1)$, using the deterministic completion, the running time is $O(\log(1/\eta) + \log^* n)$.
\end{theorem}

\begin{proof}
Let $\freq(c) = f_1(c) + f_2(c)$, where $f_1$ counts vertices colored in the proposal phase and $f_2$ in the completion; $0 \le f_2(c) \le \freq(\bot)$ for every $c$. The proposal phase has influence radius at most $\tau_0$ (since effects propagate one hop per round), so $f_1(c)$ and $\freq(\bot)$ are functions of the independent random strings $\Rand_v$ with per-string influence at most $b_0 = \max_v |\Gamma_{\tau_0}(v)| \le (\Delta+1)^{\tau_0}$. Define the \emph{concentration threshold}
$t = b_0\sqrt{(\gamma+1)\,n\ln n}$.
The parameter verification below establishes two estimates,
\[
\mbox{(i)}\quad b_0 \;\le\; n^{1/4}
\qquad\qquad\mbox{and}\qquad\qquad
\mbox{(ii)}\quad t \;\le\; \frac{\eta\sigma}{16}\,,
\]
of which (i) feeds the concentration step and (ii) the final assembly.

\noindent
We first prove estimate (i). As
$$\tau_0 \le 6\ln(16(\Delta+1)/\eta) + 1 \le 25\big(\ln(\Delta+1) + \ln(1/\eta) + 1\big)$$
and 
$$\ln(\Delta+1), \ln(1/\eta) \le (\ln 2)\sqrt{\log n}/C,$$ 
we have
\begin{eqnarray*}
\ln b_0 &\le& \tau_0 \ln(\Delta+1) \;\le\; 25\big(2(\ln 2)\tfrac{\sqrt{\log n}}{C} + 1\big)\cdot (\ln 2)\tfrac{\sqrt{\log n}}{C} \\
&\le& \frac{50(\ln 2)^2}{C^2}\log n + o(\log n) \;\le\; \frac{\ln n}{4}
\end{eqnarray*}
for $C = 16$ and large $n$; hence $b_0 \le n^{1/4}$, proving (i).

\noindent
Next, we establish \emph{concentration.}
Let $\cE_{\mbox{\footnotesize conc}}$ be the event that $|f_1(c') - \E[f_1(c')]| \le t$ for all colors $c'$ and $|\freq(\bot) - \E[\freq(\bot)]| \le t$.
The proposal phase is a palette-symmetric partial-coloring algorithm with influence radius at most $\tau_0$, so Lemma~\ref{lemma:transfer}, applied with the concentration threshold $t$, ensures that event $\cE_{\mbox{\footnotesize conc}}$ holds with probability at least $1 - 2(\Delta+2)\,n^{-2(\gamma+1)}$. This is at least $1 - n^{-\gamma}$ for every $n \ge 4$:
$\Delta + 2 \le n + 1 \le 2n$, so, since $4n \le n^2 \le n^{\gamma+2}$ for $n \ge 4$ and $\gamma \ge 0$, we get
\[
2(\Delta+2)\,n^{-2(\gamma+1)}
\;\le\; 4n\cdot n^{-2(\gamma+1)}
\;\le\; n^{\gamma+2}\cdot n^{-2(\gamma+1)}
\;=\; n^{-\gamma}.
\]

\noindent
We next prove estimate (ii); together with the absorption of the lower-order term in the proof of (i) above, this is the step at which the theorem's assumption that $n$ be sufficiently large is used, and we spell the requirement out. By (i), $t \le n^{3/4}\sqrt{(\gamma+1)\ln n}$; on the other hand, the hypotheses $\Delta + 1 \le 2^{\sqrt{\log n}/16}$ and $\eta \ge 2^{-\sqrt{\log n}/16}$ give $\eta\sigma = \eta\, n/(\Delta+1) \ge n\, 2^{-\sqrt{\log n}/8}$. Hence $t \le \eta\sigma/16$ follows from
\[
16\sqrt{(\gamma+1)\ln n}\;\cdot\; 2^{\sqrt{\log n}/8} \;\le\; n^{1/4} .
\]
Since $\sqrt{\log n} \le \log n$ gives $2^{\sqrt{\log n}/8} \le n^{1/8}$, it suffices that $16\sqrt{(\gamma+1)\ln n} \le n^{1/8}$, i.e., that $256\,(\gamma+1)\ln n \le n^{1/4}$; this holds for all $n$ above an explicit threshold $n_0(\gamma)$ depending only on $\gamma$, and it is the threshold behind ``for all sufficiently large $n$'' in the statement of the theorem. This proves (ii).

\noindent
Combining Lemmas~\ref{lemma:transfer} and~\ref{lemma:decay} gives $\E[f_1(c)] = (n - \E[\freq(\bot)])/(\Delta+1) \in [\sigma(1 - \eta/16), \sigma]$. On $\cE_{\mbox{\footnotesize conc}}$, and using $\freq(\bot) \le \E[\freq(\bot)] + t \le \eta\sigma/16 + \eta\sigma/16$ (by Lemma~\ref{lemma:decay} and estimate (ii)), we get
\begin{align*}
\freq(c) &\ge f_1(c) \;\ge\; \sigma(1 - \eta/16) - \eta\sigma/16 \;\ge\; (1-\eta)\sigma,
\\
\freq(c) &\le f_1(c) + \freq(\bot) \;\le\; \sigma + \tfrac{3\eta\sigma}{16} \;\le\; (1+\eta)\sigma.
\end{align*}
It remains to verify that the final coloring is proper.
Kept proposals are conflict-free by construction, and the completion is a proper list completion of $H$ with lists of size $\ge \deg_H + 1$, valid w.h.p.\ (randomized case) or always (deterministic case). The round complexity is $\tau_0 = O(\log(\Delta/\eta))$ plus the completion bounds of Fact~\ref{fact:completion}. The theorem follows.
\end{proof}


\begin{remark}
\label{rem:conc-range}
The restriction $\Delta \le 2^{\sqrt{\log n}/C}$ comes solely from the crude influence bound $b \le (\Delta+1)^{\tau_0}$; the information-theoretic limit of the statement is $\sigma = \Theta(\eta^{-2}\log n)$, below which even an idealized multinomial allocation fails to be $(1\pm\eta)$-balanced w.h.p. We conjecture the theorem extends to that threshold, matching the condition $\sigma = \Omega(\log n)$ that algorithm $\LogCompact$ of~\cite{NP25} had to assume; see Section~\ref{sec:open}. Section~\ref{sec:split} shows that if the palette is relaxed by a $(1+\eta)$ factor, the restriction can instead be bypassed entirely, for all $\Delta \le n^{1-o(1)}$; Section~\ref{sec:growth} bypasses it with palette exactly $\Delta+1$ on graphs of mild ball growth, up to $\Delta = n^{1/4-o(1)}$.
\end{remark}

In its parameter regime, Theorem~\ref{thm:symcolor} compares favorably with each row of Table~\ref{tab:compare}: palette exactly $\Delta+1$; two-sided control $(1\pm\eta)\sigma$, strictly inside the ranges $[\sigma/2, 2k\sigma]$ and $[\sigma/2, O(\sigma\log\Delta)]$, and incomparable with $[\sigma/2,\sigma]$ (a stronger upper cap, but with palette up to $2(\Delta+1)$ and a wider lower spread); and no $D$ or $\Delta$ term in the running time, whereas every $\CONGEST$ algorithm of~\cite{NP25} pays $\Theta(D+\Delta)$ per invocation of its global primitives. The mechanism is that concentration replaces color accounting, and the smallness of the asymmetric completion replaces quota assignment. The lower bounds of Sections~\ref{sec:det-lb} and~\ref{sec:diameter-lb} show the largeness condition is not an artifact: with small imbalance budgets, balance is genuinely global.

\subsection{A growth-sensitive form of the exact-palette theorem}
\label{sec:growth}

The degree restriction of Theorem~\ref{thm:symcolor} is a limitation of its proof, not of the algorithm: the maximum degree enters only through the worst-case estimate $b \le (\Delta+1)^{\tau_0}$ on the influence balls, which is attained by trees but pessimistic for locally dense graphs (in a clique, the ball at any radius is the whole clique). Making the ball size itself the parameter yields a strictly stronger statement at no cost. We get the following.

\begin{theorem}
\label{thm:symcolor-growth}
Fix $\gamma \ge 1$ and $\eta \in (0,1]$, let $\tau_0 = \lceil 6\ln(16(\Delta+1)/\eta)\rceil$, and let $b = \max_v |\Gamma_{\tau_0}(v)|$. If
$b\,\sqrt{(\gamma+1)\,n\ln n} \le \eta\sigma/16$,
then $\SymColor(G,\eta)$ outputs, with probability at least $1 - n^{-\gamma}$ (for sufficiently large $n$), a proper coloring with palette exactly $\Delta+1$ and all class sizes in $[(1-\eta)\sigma, (1+\eta)\sigma]$, within $\tau_0$ rounds plus the completion bound of Fact~\ref{fact:completion}.
\end{theorem}

\begin{proof}
This is the proof of Theorem~\ref{thm:symcolor} with the paragraph \emph{Bounding $b$} replaced by the hypothesis: the concentration step applies Lemma~\ref{lemma:transfer} with $t = b\sqrt{(\gamma+1)n\ln n} \le \eta\sigma/16$, and the decay, assembly and properness steps are unchanged. (The hypothesis of Theorem~\ref{thm:symcolor} implies the present one: there, $b \le (\Delta+1)^{\tau_0} \le n^{1/4}$ while $\eta\sigma \ge n^{1-o(1)}$.)
\end{proof}

For graphs of polynomially bounded ball growth, the 
claim
takes the following concrete form.

\begin{corollary}
\label{cor:growth}
Fix $\gamma \ge 1$, $\eta \in (0,1]$ and constants $\const_g, J \ge 1$, and suppose every ball of $G$ satisfies $|\Gamma_{\tau_0}(v)| \le \const_g\,\Delta\,\tau_0^{J}$. If
\begin{equation}
\label{eq:ball-growth}
16\,\const_g\,\Delta(\Delta+1)\, \tau_0^{J}\sqrt{(\gamma+1)\ln n} \;\le\; \eta\sqrt n,
\end{equation}
then the conclusion of Theorem~\ref{thm:symcolor-growth} holds; for constant $\eta$ this admits every $\Delta \le n^{1/4 - o(1)}$. In particular, the matched clique cycle $B_{m,\cliqsize}$ satisfies $|\Gamma_r(v)| \le (2r+1)\cliqsize \le 3r\Delta$ for $r \ge 1$, so for constant $\eta$ and every $\cliqsize \le n^{1/4-o(1)}$, the family $B_{m,\cliqsize}$ admits a proper coloring with palette \emph{exactly} $\Delta+1$ and all classes in $(1\pm\eta)\sigma$ in $O(\log\Delta) + O(\log^3\log n)$ rounds, whereas exact balance on the same family costs $\Theta(D)$ rounds (Theorem~\ref{thm:diam-lb}). The exact-vs-approximate separation of Remark~\ref{rem:diam-sep} thus holds with the optimal palette throughout this range.
\end{corollary}

\begin{proof}
The growth hypothesis gives $b \le \const_g\Delta\tau_0^{J}$, so the 
condition in Eq. \eqref{eq:ball-growth} yields $b\sqrt{(\gamma+1)n\ln n} \le \eta n/(16(\Delta+1)) = \eta\sigma/16$, and Theorem~\ref{thm:symcolor-growth} applies; for constant $\eta$ the condition reads $\Delta^2 \le \eta\sqrt n / O(\log^{J+1} n)$, i.e.\ $\Delta \le n^{1/4-o(1)}$. For $B_{m,\cliqsize}$: every edge changes the ring coordinate by at most one (Observation~\ref{obs:clique-cycle}), so $\Gamma_r(v)$ is contained in $2r+1$ cliques, whence $|\Gamma_r(v)| \le (2r+1)\cliqsize \le 3r\Delta$ for $r \ge 1$. The corollary follows.
\end{proof}

\subsection{All degrees, with palette slack}
\label{sec:split}

Theorem~\ref{thm:symcolor} keeps the palette at exactly $\Delta+1$, and its degree restriction comes from the influence balls $(\Delta+1)^{\tau_0}$ of the proposal dynamics. If the palette is relaxed by a $(1+\eta)$ factor (the palette regime of the tradeoff suite of~\cite{NP25}), the concentration approach extends to all $\Delta \le n^{1-o(1)}$, still independently of $D$:

\begin{corollary}
\label{cor:split-all}
Fix $\eta \in (0,1]$. There is $C = C(\eta)$ such that for every $n$-vertex graph with $\sigma \ge 2^{C(\log\log n)^2}$ (i.e., every $\Delta \le n/2^{C(\log\log n)^2} = n^{1-o(1)}$), a proper coloring with palette size $\chi \le (1+\eta)(\Delta+1)$ and all class sizes in $(1\pm\eta)\,n/\chi$ can be computed w.h.p.\ in $O\big(\log n\,(\log\log n)^2\big)$ rounds of $\LOCAL$. More generally, with a universal constant $C$, the same holds in $O\big((\log\log n + \log(2/\eta))^2\log n\big)$ rounds for every $\eta \in [2^{-\sqrt{\log n}/64}, 1]$ and every graph with $\sigma \ge 2^{C(\log\log n + \log(2/\eta))^2}$; there is no polynomial dependence on $1/\eta$.
\end{corollary}

The mechanism is a one-level random split. Vertices hash themselves into $L$ parts, the palette is cut into $L$ disjoint blocks, and each part becomes a \emph{private} coloring instance of polylogarithmic maximum degree, inside which the entire machinery of this section (transfer lemma, residual decay, bounded differences, completion) runs with polylogarithmic influence balls. Properness across parts is automatic (since each part uses a disjoint block of colors), and, crucially, every budget becomes \emph{per part}: a color can only ever be used inside its own part, so the residue that may pile onto it is the part's residue, an $\eta/\mathrm{poly}\log n$ \emph{fraction} of the part, reachable in $O(\log\log n + \log(1/\eta))$ proposal rounds. This decoupling of the decay horizon from $\Delta$ is what breaks the $(\Delta+1)^{\tau_0}$ barrier; the details follow.

We now prove Corollary~\ref{cor:split-all} via Algorithm~\ref{alg:splitsym} ($\SplitSym$). Given $\eta \in (0,1]$ and a failure exponent $\gamma \ge 1$, set
\[
\epsilon_0 = \eta/8, \qquad
\ksplit = \lceil 12(\gamma+3)\,\epsilon_0^{-2}\ln n\rceil, \qquad
L = \lceil \Delta/\ksplit\rceil,
\]
\[
p = \lceil(1+\epsilon_0)\ksplit\rceil + 1, \qquad
\chi = Lp, \qquad
\tau_1 = \lceil 6\ln(64\,p/\eta)\rceil .
\]
For later use we record the parameter scales in the following ledger; here $\gamma$ is a constant, and $b = \big((1+\epsilon_0)\ksplit+1\big)^{\tau_1}$ denotes the influence-ball bound of Lemma~\ref{lemma:per-part} below:
\begin{align}
\label{eq:notation-ledger}
\ksplit &= \Theta(\eta^{-2}\log n),
&
p &= \Theta(\eta^{-2}\log n),
\\
\chi &\le (1+\eta/4)(\Delta+1),
&
\ln b &= O\big((\log\log n + \log(2/\eta))^2\big).
\nonumber
\end{align}
The first two scales are immediate from the definitions; the palette bound is proved in Lemma~\ref{lemma:split-good} below (note $1 + 2\epsilon_0 = 1 + \eta/4$); and the bound on $\ln b = \tau_1\ln\big((1+\epsilon_0)\ksplit+1\big)$ follows by multiplying $\tau_1 \le 6\ln(64p/\eta) + 1 = O(\log\log n + \log(2/\eta))$ and $\ln\big((1+\epsilon_0)\ksplit+1\big) = O(\log\log n + \log(2/\eta))$.

The palette is $\cP = [\chi]$, cut into the consecutive blocks $\cP_\ell = \{(\ell-1)p+1, \dots, \ell p\}$, $\ell \in [L]$. More formally, the algorithm is the following.

\bigskip
\begin{algorithm}[H]
\caption{$\SplitSym(G, \eta)$ \label{alg:splitsym}}
\KwIn{Graph $G$ with degree bound $\Delta$; parameters $\eta$, $\gamma$, $\epsilon_0$, $\ksplit$, $L$, $p$, $\chi$, and $\tau_1$ as defined above.}
\KwOut{A proper coloring with palette $\cP=[\chi]$ and every class size in $(1\pm\eta)n/\chi$, w.h.p.}
Every vertex $v$ draws $\ell(v) \in [L]$ uniformly at random and announces it to its neighbors; let 
$V_\ell = \{ v \mid \ell(v)=\ell\}$ 
and $H_\ell = G[V_\ell]$\;
\For{\emph{round} $r = 1, \dots, \tau_1$ \emph{(all parts in parallel)}}{
Each uncolored $v$, with probability $1/2$, proposes a uniform color $c_v$ from $L_v = \cP_{\ell(v)} \setminus \{\mbox{colors of } v\mbox{'s colored neighbors}\}$\;
$v$ keeps $c_v$ permanently if no neighbor of $v$ proposes $c_v$ in this round (only same-part neighbors can)\;
}
\emph{Completion, per part:} color the uncolored set $U_\ell$ by $(\mathrm{deg}+1)$-list coloring of $H_\ell[U_\ell]$ with the lists $L_v$, deterministically in $O(\log^2 p\cdot\log n)$ rounds (Fact~\ref{fact:completion}(b), using $\Delta(H_\ell[U_\ell]) < p$ on $\cE_1$)\;
\KwReturn{$\clr$}
\end{algorithm}

\bigskip
The first lemma shows that the random split is well behaved.

\begin{lemma}
\label{lemma:split-good}
Suppose $\Delta \ge 8\ksplit/\epsilon_0$. With probability at least $1 - n^{-(\gamma+1)}$ the following event $\cE_1$ holds: (i) every $v$ satisfies $\deg_{H_{\ell(v)}}(v) \le (1+\epsilon_0)\ksplit \le p - 1$; and (ii) every part satisfies $n_\ell = |V_\ell| \in (1\pm\epsilon_0)\,n/L$. Moreover, deterministically, $\chi \le (1+2\epsilon_0)(\Delta+1) \le (1+\eta/4)(\Delta+1)$.
\end{lemma}

\begin{proof}
Conditioned on $\ell(v)$, the quantity $\deg_{H_{\ell(v)}}(v)$ is $\mathrm{Bin}(\deg(v), 1/L)$ with mean $\deg(v)/L \le \Delta/L \le \ksplit$. Since $\Delta \le L\ksplit$, it is stochastically dominated by $\mathrm{Bin}(L\ksplit, 1/L)$, of mean $\ksplit$, and the multiplicative Chernoff bound gives $\Pr[\deg_{H_{\ell(v)}}(v) \ge (1+\epsilon_0)\ksplit] \le \exp(-\epsilon_0^2\ksplit/3) \le n^{-4(\gamma+3)}$; a union bound over the $n$ vertices yields $(i)$. 

To see that $(ii)$ holds, note that $n_\ell \sim \mathrm{Bin}(n, 1/L)$ with mean $n/L \ge n\ksplit/(2\Delta) \ge \ksplit\sigma/2 \ge \ksplit/2$, so $\Pr[|n_\ell - n/L| \ge \epsilon_0 n/L] \le 2\exp(-\epsilon_0^2\ksplit/6) \le 2n^{-2(\gamma+3)}$; finally apply the union bound  over the $L \le n$ parts. The total failure probability is at most $n^{-(\gamma+1)}$.

For bounding the palette size, note that 
$$\chi = Lp \le (\Delta/\ksplit + 1)\big((1+\epsilon_0)\ksplit + 2\big) = \Delta(1+\epsilon_0) + 2\Delta/\ksplit + (1+\epsilon_0)\ksplit + 2.$$ 
Now $2\Delta/\ksplit \le \epsilon_0\Delta/4$ since $\ksplit \ge 8/\epsilon_0$, and $(1+\epsilon_0)\ksplit + 2 \le 2\ksplit \le \epsilon_0\Delta/4$ since $\Delta \ge 8\ksplit/\epsilon_0$ and $(1-\epsilon_0)\ksplit \ge 2$. Hence $\chi \le (1+\epsilon_0)\Delta + \epsilon_0\Delta/2 \le (1+2\epsilon_0)(\Delta+1)$. The lemma follows.
\end{proof}

Within a part, the machinery of Section~\ref{sec:concentration} applies without change, as follows.

\begin{lemma}
\label{lemma:per-part}
Condition on any labeling $\ell(\cdot)$ in $\cE_1$ and fix a part $\ell$. Set $b = \big((1+\epsilon_0)\ksplit+1\big)^{\tau_1}$ and $t = b\sqrt{(\gamma+3)\,n_\ell\ln n}$. With probability at least $1 - 2(p+2)\,n^{-2(\gamma+3)}$ over the proposal randomness of $V_\ell$, the part ends properly colored from $\cP_\ell$ and every color $c' \in \cP_\ell$ satisfies
\[ \frac{n_\ell}{p}\Big(1 - \frac{\eta}{64}\Big) - t \;\le\; \freq(c') \;\le\; \frac{n_\ell}{p}\Big(1 + \frac{\eta}{64}\Big) + 2t . \]
\end{lemma}

\begin{proof}
Given the labeling, the proposal phase restricted to part $\ell$ is a partial-coloring algorithm on the instance $(H_\ell, \cP_\ell)$, driven by the independent strings of $V_\ell$ alone; proposals of other parts lie in disjoint blocks and can never collide with proposals in $\cP_\ell$, so the conflict rule coincides with the same-part conflict rule.

We first establish the decay of the uncolored vertex set $U$. On $\cE_1$ every $u \in V_\ell$ has $\deg_{H_\ell}(u) + 1 \le p$, so every uncolored $u$ always satisfies $|L_u| \ge p - (\deg_{H_\ell}(u) - d_U(u)) \ge d_U(u) + 1$, where $d_U(u)$ is $u$'s uncolored degree in $H_\ell$. This is the only property of the palette used in the proof of Lemma~\ref{lemma:decay}, which therefore applies as is to $(H_\ell, \cP_\ell)$ and yields $\E[\upsilon_\ell] \le n_\ell\, e^{-\tau_1/6} \le n_\ell \cdot \eta/(64p)$, where $\upsilon_\ell = |U_\ell|$ and the last step is the choice of $\tau_1$.

We next establish transfer and concentration. The phase is palette-symmetric over $\cP_\ell$ (the uniform proposal rule is invariant under every permutation of $\cP_\ell$) and has influence radius at most $\tau_1$ in $H_\ell$, whose balls have size at most $(\Delta(H_\ell)+1)^{\tau_1} \le b$ on $\cE_1$. The proof of Lemma~\ref{lemma:transfer} uses only palette symmetry, the influence bound and the independence of the strings (never the relation $|\cP| = \Delta+1$), so it applies to $(H_\ell, \cP_\ell)$: $\E[\freq_1(c')] = (n_\ell - \E[\upsilon_\ell])/p$ for every $c' \in \cP_\ell$, and with probability at least $1 - 2(p+2)\exp(-2t^2/(n_\ell b^2)) = 1 - 2(p+2)n^{-2(\gamma+3)}$, simultaneously $|\freq_1(c') - \E[\freq_1(c')]| \le t$ for all $c' \in \cP_\ell$ and $|\upsilon_\ell - \E[\upsilon_\ell]| \le t$, where $\freq_1$ counts the proposal phase only.

Finally we consider the completion phase. The lists of uncolored vertices satisfy $|L_v| \ge p - (\deg_{H_\ell}(v) - \deg_{U_\ell}(v)) \ge \deg_{U_\ell}(v) + 1$, so the $(\mathrm{deg}+1)$-list instance on $H_\ell[U_\ell]$ is legal, with $\Delta(H_\ell[U_\ell]) \le (1+\epsilon_0)\ksplit < p$ on $\cE_1$, lists from the polynomial universe $[\chi]$, and polynomial identifiers; the deterministic completion of Fact~\ref{fact:completion}(b) therefore colors it fully and properly. The parts are vertex-disjoint and the completions run on disjoint induced subgraphs in parallel without interacting; determinism means no failure probabilities need to be union-bounded over the parts. The completion uses only colors of $\cP_\ell$ and adds between $0$ and $\upsilon_\ell$ vertices to each class. Hence $\freq_1(c') \le \freq(c') \le \freq_1(c') + \upsilon_\ell$, and on the event above,
\begin{align*}
\freq(c') &\ge \frac{n_\ell - \E[\upsilon_\ell]}{p} - t
  \ge \frac{n_\ell}{p}\Big(1-\frac{\eta}{64}\Big) - t,\\
\freq(c') &\le \frac{n_\ell}{p} + t + \big(\E[\upsilon_\ell] + t\big)
  \le \frac{n_\ell}{p}\Big(1+\frac{\eta}{64}\Big) + 2t . \qedhere
\end{align*}
\end{proof}

We are now ready to establish the main result of this subsection.

\begin{theorem}
\label{thm:split}
For every constant $\gamma \ge 1$ there is a constant $C = C(\gamma)$ such that the following holds for every $\eta \in (0,1]$ and every $n$-vertex graph $G$ with
\begin{align*}
\Delta ~\ge~ C\,\eta^{-3}\log n
\qquad and \qquad
\sigma=\frac{n}{\Delta+1}
  ~\ge~ 2^{\,C(\log\log n + \log(2/\eta))^2}.
\end{align*}
Algorithm~\ref{alg:splitsym} ($\SplitSym$) runs in $O\big((\log\log n + \log(2/\eta))^2\,\log n\big)$ rounds and, with probability at least $1 - n^{-\gamma}$, outputs a proper coloring of $G$ with palette size $\chi \le (1+\eta/4)(\Delta+1)$ with all frequencies in $\big[(1-\eta)\,n/\chi,\; (1+\eta)\,n/\chi\big]$. 
The round complexity is independent of $D$; for fixed $\eta$ it reads $O(\log n\,(\log\log n)^2)$, with no polynomial dependence on $1/\eta$.
\end{theorem}

\begin{proof}
The palette bound and the event $\cE_1$ are given in Lemma~\ref{lemma:split-good} (the hypothesis $\Delta \ge 8\ksplit/\epsilon_0$ holds by $\Delta \ge C\eta^{-3}\log n$ for $C$ large). Properness follows since colors of adjacent vertices in different parts lie in disjoint blocks; within a part, Lemma~\ref{lemma:per-part} applies. Every vertex is colored (the completion is total). Fix a color $c' \in \cP_\ell$; only vertices of $V_\ell$ ever receive $c'$, so $\freq(c')$ is the part-$\ell$ class size. On $\cE_1$, $n_\ell/p \in (1\pm\epsilon_0)\,n/(Lp) = (1\pm\eta/8)\,n/\chi$, so on the events of Lemma~\ref{lemma:per-part} for all parts,
\begin{align*}
\freq(c')
  &\ge \Big(1-\frac{\eta}{8}\Big)
        \Big(1-\frac{\eta}{64}\Big)\frac n\chi - t,
\\
\freq(c')
  &\le \Big(1+\frac{\eta}{8}\Big)
        \Big(1+\frac{\eta}{64}\Big)\frac n\chi + 2t,
\end{align*}
and both bounds land in $(1\pm\eta)\,n/\chi$ provided $t \le (\eta/4)\,n/\chi$.

We next bound $t$. On $\cE_1$, $n_\ell \le 2n/L = 2p\,(n/\chi)$, so $t \le b\sqrt{2(\gamma+3)\,p\,(n/\chi)\ln n}$, and $t \le (\eta/4)\,n/\chi$ follows from
\begin{equation}
\label{eq:required-bound-ln-b}
\ln b \le \tfrac12\ln(n/\chi) - \tfrac12\ln\big(2(\gamma+3)\,p\ln n\big) - \ln(4/\eta).
\end{equation}
By the parameter ledger in Eq. \eqref{eq:notation-ledger}, $\ln b = O\big((\log\log n + \log(2/\eta))^2\big)$, while the two remaining terms on the left-hand side below are $O(\log\log n + \log(2/\eta))$ (as $p = \Theta(\eta^{-2}\log n)$); so there is $C' = C'(\gamma)$ with
\begin{equation}
\label{eq:known}
\ln b + \tfrac12\ln\big(2(\gamma+3)\,p\ln n\big) + \ln(4/\eta)
  ~\le~ C'\big(\log\log n + \log(2/\eta)\big)^2 .
\end{equation}
Since $n/\chi \ge n/\big((1+2\epsilon_0)(\Delta+1)\big) \ge \sigma/2$, the hypothesis $\sigma \ge 2^{C(\log\log n + \log(2/\eta))^2}$ with $C = 4C' + 4$ gives $\tfrac12\ln(n/\chi) \ge \tfrac12\ln\sigma - 1 \ge C'\big(\log\log n + \log(2/\eta)\big)^2$, which combined with Eq. \eqref{eq:known} 
yields the requirement of Eq. \eqref{eq:required-bound-ln-b}; hence $t \le (\eta/4)\,n/\chi$.

Finally we analyze the failure probability and time complexity. A union bound combines the bounds $n^{-(\gamma+1)}$ for $\cE_1$ and $L \cdot 2(p+2)n^{-2(\gamma+3)} \le n^{-(\gamma+2)}$ for the parts, yielding failure probability at most $n^{-\gamma}$ in total; the completion is deterministic and contributes no failure. The total number of rounds is $1 + \tau_1 + O(\log^2 p\cdot\log n)$, and since $p = O\big((\gamma+3)\eta^{-2}\log n\big)$ gives $\log p = O(\log\log n + \log(2/\eta))$, this is $O\big((\log\log n + \log(2/\eta))^2\log n\big)$, of which the completion term dominates $\tau_1 = O(\log\log n + \log(2/\eta))$. The theorem follows.
\end{proof}

\begin{proof}[Proof of Corollary~\ref{cor:split-all}]
If $\Delta \ge C\eta^{-3}\log n$, apply Theorem~\ref{thm:split}. Otherwise
\begin{align*}
\Delta + 1
  ~\le~ C\eta^{-3}\log n + 1
  ~\le~ 2^{\,3\sqrt{\log n}/64 + O(\log\log n)}
  ~\le~ 2^{\sqrt{\log n}/16}~,
\end{align*}
for large $n$, and $\eta \ge 2^{-\sqrt{\log n}/64} \ge 2^{-\sqrt{\log n}/16}$, so Theorem~\ref{thm:symcolor} applies and gives the stronger guarantee $\chi = \Delta+1$ with classes $(1\pm\eta)\sigma$.
\end{proof}

We close by comparing the resulting guarantee with the prior palette--frequency tradeoffs.
%
In its range, Corollary~\ref{cor:split-all} improves on the palette--frequency tradeoffs of~\cite{NP25} in $\LOCAL$: for every $k \le O(1/\eta)$, palette $(1+1/k)(\Delta+1)$ with frequencies $[\sigma/2, 2k\sigma]$ (computed there in $\CONGEST$ with $\Theta(D+\Delta)$-scale terms) is replaced by palette $(1+\eta)(\Delta+1)$ with the two-sided range $(1\pm\eta)n/\chi$, at every $\Delta \le n^{1-o(1)}$, with no diameter dependence. We made no attempt to optimize the rounds: the dominant term is the deterministic per-part completion of Fact~\ref{fact:completion}(b), $O(\log^2 p\cdot\log n) = O\big((\log\log n + \log(2/\eta))^2\log n\big)$ rounds, which any faster deterministic $(\mathrm{deg}{+}1)$-list coloring would improve as a black box; the proposal phase itself takes only $O(\log\log n + \log(1/\eta))$ rounds. We use a deterministic completion because the parts are small ($n_\ell \approx \sigma\,\mathrm{poly}\log n$), so per-part randomized guarantees of the form $1 - 1/\mathrm{poly}(n_\ell)$ would not union-bound over $L$ parts.

The exact-palette obstruction of this approach is summarized in the following remark.

\begin{remark}
\label{rem:split-barrier}
The $(1+\eta)$ palette relaxation in Theorem~\ref{thm:split} is not an accident of the analysis. Any scheme that partitions both the vertices and the palette into private sub-instances conserves the palette-to-degree ratio: the sub-palettes sum to $\chi$ while the expected induced degrees sum to $\Delta$ per vertex, so a part receives palette $\approx$ (its mean degree) $+ (\chi - \Delta)/L$. For $\chi = \Delta+1$ the global unit of slack splits into $1/L < 1$ per part, while the parts' degree fluctuations are $\pm\Theta(\sqrt{\ksplit\log n})$: the sub-instances are palette-\emph{deficient}, and with no global slack the deficiency cannot be absorbed. Extending the exact-palette guarantee of Theorem~\ref{thm:symcolor} beyond $\Delta = 2^{\Theta(\sqrt{\log n})}$ therefore appears to require controlling the influence of the \emph{unsplit} dynamics, which remains open (Section~\ref{sec:open}).
\end{remark}

\section{Limitations and open problems}
\label{sec:open}

\paragraph{Limitations.}
Four limitations should be borne in mind when interpreting the results. First, all lower bounds in this paper are deterministic; we prove no randomized lower bound, and on cycles randomization provably beats the deterministic tradeoff (Theorem~\ref{thm:intro-cycle-ub}(a) vs.\ Theorem~\ref{thm:intro-ring-lb}(b)). Second, the exact-palette guarantee on general graphs is restricted to $\Delta \le 2^{\sqrt{\log n}/C}$ (Theorem~\ref{thm:symcolor}), or $\Delta \le n^{1/4-o(1)}$ under mild ball growth (Section~\ref{sec:growth}); for larger degrees we require palette slack (Corollary~\ref{cor:split-all}). Third, the twisted fiber products realize every diameter \emph{scale}, but not every \emph{geometry}: their ring coordinate mixes slowly, so whether exact equity costs $\Omega(D)$ on genuinely expanding graphs is open (Remark~\ref{rem:open}). Finally, the companion manuscript~\cite{NP25} is unpublished; its role here is confined to the restated rows of Table~\ref{tab:compare}.

\dnsparagraph{Randomized optimality on cycles.}
We conjecture that Theorem~\ref{thm:intro-cycle-ub}(a) is optimal up to logarithmic factors: there is a constant $\const_0 > 0$ such that every $T$-round randomized $\LOCAL$ algorithm on $\cC_n$, $T \le n/\const_0$, that outputs with probability at least $2/3$ a free $3$-coloring with all classes of size $n/3 \pm \imb$ must have $\imb \ge \sqrt{n/T}/\const_0$. Together with Theorem~\ref{thm:balseg}, this would pin the randomized tradeoff at $\tilde\Theta(n/\imb^2)$ and extend the continuum of Corollary~\ref{cor:continuum} to randomized complexity. Theorem~\ref{thm:det-stab} provides the deterministic template, namely, that an approximately confined count forces a near-integer density; what is missing is its probabilistic analogue. 

\dnsparagraph{The $\log^* n$ gap on the deterministic tradeoff.}
The deterministic cycle tradeoff stands at $\Omega(n/\imb)$ versus $O((n/\imb)\log^* n)$, and Proposition~\ref{prop:ruling-lb} shows that faster anchors cannot close the gap (over polynomial identifier universes of size at least $2n$).
One may reduce the difficulty to solving a one-dimensional \emph{cut-set} problem.
Cut sets require $\Omega(w)$ rounds, but no super-$\Omega(w)$ lower bound is known. Hence the open question is whether a balanced-coloring algorithm can dispense with ruling-set anchors, or whether cut sets themselves incur an additional $\log^* n$ factor.

\dnsparagraph{Exact palette at larger degrees.}
One open question is whether the exact-palette guarantee of Theorem~\ref{thm:symcolor} extends beyond $\Delta = 2^{\Theta(\sqrt{\log n})}$ on general graphs. Theorem~\ref{thm:symcolor-growth} handles graphs of polynomially bounded ball growth up to $\Delta = n^{1/4-o(1)}$, while Remark~\ref{rem:split-barrier} shows why vertex-and-palette splitting cannot reach the exact palette. The remaining case is therefore expander-like ball growth. Relatedly, we conjecture that the optimal largeness threshold for the concentration regime is $\sigma = \Theta(\eta^{-2}\log n)$ (Remark~\ref{rem:conc-range}).

\subsubsection*{AI-use disclosure}
AI assistants 
were used for language editing, document organization and \LaTeX{} assistance throughout the paper, and as a sounding board in exploratory discussions of proof strategies. In addition, first drafts of three technical parts (the stability argument in Section~\ref{sec:det-stab}, the clique-cycle reduction in Section~\ref{sec:mcc}, and the palette-splitting construction in Section~\ref{sec:split}) were produced with AI assistance. Every definition, statement and proof in the paper was subsequently checked line by line, and where necessary corrected or reproved, by the authors, who take sole responsibility for all mathematical content, citations and originality claims.

\clearpage
\appendix
\centerline{\Large\bf Appendix}
\section{The universe threshold: proof of Theorem~\ref{thm:universe}}
\label{app:universe}

This appendix proves Theorem~\ref{thm:universe} in full. Part (a) is elementary and is proved first, together with a complete characterization of the $0$-round rules admitting a sure count (Propositions~\ref{prop:app-modthree} and~\ref{prop:app-zeroround}); the characterization shows, in particular, that one extra identifier already defeats every $0$-round rule. Part (b) is Theorem~\ref{thm:det-exact} re-proved with a tighter identifier accounting (Theorem~\ref{thm:app-sharp}): all loop lengths are kept below $\approx n/2$ (the identifier expenditure of Lemmas~\ref{lemma:boundary} and~\ref{lemma:loops} is proportional to the loop length, and this is the only place where the hypothesis $N \ge 2n+2$ was used), and the final splitting of the $n$-window loop of an actual assignment is into \emph{three} arcs of length $\approx n/3$ rather than two arcs of length $\approx n/2$, so that every loop that must be evaluated fits inside the restricted range; the exchange lemma itself needs only $N \ge n+1$, where the count is exactly tight. The price is a larger constant, $T \le n/50 - 1$ instead of $T \le n/40 - 1$ (neither constant was optimized).

We keep the notation of Section~\ref{sec:det-lb}: assignments $\ID$, windows $W(v_i,\ID,T)$, the count $f(\ID) = \sum_v F(W(v,\ID,T))$, sure counts, $k$-tuples, heads and tails, steps, chains, loops and their weights $\omega$, arrangements and their inner sums $P$, contexts $\cI$ and local window sums $\Psi(\cI; x)$, and the paths $\ShiftIn(\cT, \cT')$. Recall that every window of a step consists of distinct identifiers, while distinct steps of a chain may reuse identifiers (a loop necessarily does), and that for every $2T$-tuple $\cT$ and every $\mu \ge 2T+1$ with $\mu + 2T \le N$, a loop of length $\mu$ from $\cT$ exists: arrange the entries of $\cT$ followed by $\mu - 2T$ fresh identifiers in a cycle and read off its windows.

\subsection{The universe $N = n$, and $0$-round rules}
\label{app:universe-easy}

\begin{proposition}[Theorem~\ref{thm:universe}(a)]
\label{prop:app-modthree}
Let $3 \mid n$ and $N = n$. The $0$-round rule that colors each vertex $v$ with $\ID(v) \bmod 3 \in \{0,1,2\}$ outputs, on \emph{every} assignment, a free $3$-coloring of $\cC_n$ with all frequencies exactly $n/3$. Equivalently, for each $c \in \{0,1,2\}$ the $0$-round window rule $F_c(x) = \mathds{1}[x \equiv c \pmod 3]$ has sure count $n/3 \notin \{0, n\}$. More generally, at $N = n$, for every $A \subseteq [n]$ the rule $F = \mathds{1}_A$ has sure count $|A|$, so every $K \in \{0, \dots, n\}$ is a sure count of a $0$-round rule.
\end{proposition}

\begin{proof}
Since $N = n$, every assignment $\ID$ is a bijection from the vertex set onto $[n]$. Hence for any $A \subseteq [n]$, the count of the rule $\mathds{1}_A$ equals $|\ID(V) \cap A| = |[n] \cap A| = |A|$, independently of $\ID$. For $A_c = \{ x \in [n] : x \equiv c \pmod 3\}$ we have $|A_c| = n/3$ for each $c \in \{0,1,2\}$, because $3 \mid n$ makes the residues of $1, 2, \dots, n$ modulo $3$ exactly balanced. The three classes of the coloring $v \mapsto \ID(v) \bmod 3$ are the preimages of $A_0, A_1, A_2$, of sizes $n/3$ each, on every assignment.
\end{proof}

\begin{proposition}
\label{prop:app-zeroround}
Let $n \ge 1$ and $N \ge n$, let $F : [N] \to \{0,1\}$ be a $0$-round window rule, and put $A = F^{-1}(1)$. Then $F$ has a sure count if and only if one of the following holds:
(i) $N = n$; then the sure count is $K = |A|$, and every $K \in \{0, \dots, n\}$ arises this way;
(ii) $A = \emptyset$; then $K = 0$ (any $N$); or
(iii) $A = [N]$; then $K = n$ (any $N$).
Consequently, for every $N \ge n+1$ the only $0$-round sure counts are $0$ and $n$: already $N = n+1$ breaks the $\bmod\,3$ rule of Proposition~\ref{prop:app-modthree}, and no other $0$-round rule can replace it.
\end{proposition}

\begin{proof}
For any assignment $\ID$ with image $S = \ID(V)$, an $n$-element subset of $[N]$, we have $f(\ID) = |S \cap A|$; conversely, every $n$-element subset $S \subseteq [N]$ is the image of some assignment (place its elements on the cycle in any order). Hence $F$ has sure count $K$ if and only if \emph{every} $n$-element subset of $[N]$ meets $A$ in exactly $K$ elements.

Sufficiency of (i)--(iii) is immediate: in case (i) the only $n$-subset of $[n]$ is $[n]$ itself, and $|[n] \cap A| = |A|$; in cases (ii) and (iii), $|S \cap A|$ is $0$ resp.\ $n$ for every $S$. For necessity, suppose $N \ge n+1$ and $\emptyset \ne A \ne [N]$; we show no sure count exists. Pick $a \in A$ and $b \in [N] \setminus A$. Since $|[N] \setminus \{a, b\}| = N - 2 \ge n - 1$, we may pick a set $C \subseteq [N] \setminus \{a,b\}$ with $|C| = n-1$. Then $S = C \cup \{a\}$ and $S' = C \cup \{b\}$ are both $n$-element subsets of $[N]$, and
\[ |S \cap A| \;=\; |C \cap A| + 1 \;>\; |C \cap A| \;=\; |S' \cap A|, \]
so the count is not constant. Hence a sure count with $N \ge n+1$ forces $A \in \{\emptyset, [N]\}$ and $K \in \{0, n\}$, and a sure count $K \notin \{0,n\}$ forces $N = n$.

For the last sentence: with $N = n + 1$ and $3 \mid n$, the residue classes of $[n+1]$ modulo $3$ have sizes $n/3 + 1$, $n/3$, $n/3$ (in some order), so each rule $F_c$ of Proposition~\ref{prop:app-modthree} has $\emptyset \ne A_c \ne [n+1]$ and therefore no sure count; concretely, an assignment omitting an identifier of the enlarged residue class produces class sizes $(n/3, n/3, n/3)$, while an assignment omitting an identifier of another residue class produces one class of size $n/3 + 1$ and one of size $n/3 - 1$. By the previous paragraph, no $\{0,1\}$-valued $0$-round rule --- hence, applied per color, no $0$-round free-coloring rule --- can be surely exact at $N = n+1$.
\end{proof}

\subsection{Rigidity at every $N \ge n+1$: the re-accounted machinery}
\label{app:universe-lemmas}

\begin{theorem}[Theorem~\ref{thm:universe}(b)]
\label{thm:app-sharp}
Let $T \ge 1$ and $n \ge 50(T+1)$ (equivalently, $1 \le T \le n/50 - 1$), and let $N \ge n+1$. Let $F$ be any function mapping $(2T+1)$-tuples of distinct identifiers from $[N]$ to $\{0,1\}$. If $f(\ID) = K$ for every assignment $\ID$, then $K \in \{0, n\}$.
\end{theorem}

Together with Proposition~\ref{prop:app-zeroround} (the case $T = 0$) and Proposition~\ref{prop:app-modthree} (the case $N = n$), this determines the threshold completely in the local regime $0 \le T \le n/50 - 1$: a sure count $K \notin \{0,n\}$ is achievable if and only if $N = n$.

Fix $K \notin \{0, n\}$ and assume, towards a contradiction, that $f(\ID) = K$ for every assignment $\ID$. Throughout the appendix, set
\[ M = 4T + 2, \qquad
   \lmax = \Big\lfloor \frac{n + 4T}{2} \Big\rfloor, \qquad
   \mumax = \Big\lfloor \frac{n - 4T}{2} \Big\rfloor - 3 ; \]
$\lmax$ bounds the arrangement lengths (levels) and $\mumax$ the loop lengths evaluated below. We record the numerical facts used; all follow from $T \ge 1$ and $n \ge 50(T+1)$, i.e.\ $n \ge 50T + 50$.
\begin{itemize}
\item[(N1)] $M + \mumax \le \lmax$: indeed, $M + \mumax = \lfloor (n-4T)/2 \rfloor + 4T - 1 = \lfloor (n+4T)/2 \rfloor - 1 < \lmax$.
\item[(N2)] $2\mumax + 4T \le n - 6 \le N - 7$; also $\lmax \le n$ and $2 \lmax - 4T \le n \le N - 1$.
\item[(N3)] $\mumax \ge 4M$: since $\mumax \ge \frac{n - 4T - 1}{2} - 3 \ge \frac{46T + 49}{2} - 3 \ge 16T + 8 = 4M$.
\item[(N4)] $\big\lceil n/3 \big\rceil + 1 + 2T \le \mumax - M$: indeed, $\mumax - M \ge \frac{n-4T-1}{2} - 3 - (4T+2) = \frac{n}{2} - 6T - \frac{11}{2}$ and $\lceil n/3 \rceil + 1 + 2T \le \frac{n}{3} + 2T + 2$, so it suffices that $\frac{n}{6} \ge 8T + \frac{15}{2}$, i.e.\ $n \ge 48T + 45$, which holds as $n \ge 50T + 50$.
\item[(N5)] $\lfloor n/3 \rfloor - 1 \ge 2T$; $6T \ge 2T + 1$; and $M + 4T \le \mumax - M$, because $M + 4T = 8T + 2 \le 2M$ and $2M \le \mumax - M$ by (N3).
\end{itemize}

The four lemmas of Section~\ref{sec:det-exact-LB} are now re-proved with these restricted parameters. Only the identifier accounting changes, and we spell it out in full.

\begin{lemma}
\label{lem:app-exchange}
$\Psi(\cI; x) = \Psi(\cI; y)$ for every context $\cI$ and all identifiers $x, y$ avoiding $\cI$.
\end{lemma}

\begin{proof}
Build an assignment $\ID$ of $\cC_n$ conforming to $\cI[x]$ by filling the remaining $n - (4T+1)$ positions with distinct identifiers avoiding the $4T + 2$ identifiers of $\cI \cup \{x, y\}$; this is possible because $N - (4T+2) \ge (n+1) - (4T+2) = n - (4T+1)$, with equality permitted (the count is exactly tight at $N = n+1$). Let $\ID'$ be $\ID$ with $x$ replaced by $y$; then $\ID'$ is an assignment ($y$ appears nowhere in $\ID$) conforming to $\cI[y]$. By the standing assumption, $f(\ID) = f(\ID') = K$. Every window not containing the center position $v_0$ reads the same ordered tuple in $\ID$ and $\ID'$ and contributes equally to both counts; the remaining windows are those centered at $v_i$, $-T \le i \le T$, summing to $\Psi(\cI; x)$ in $f(\ID)$ and to $\Psi(\cI; y)$ in $f(\ID')$. Subtracting, $\Psi(\cI;x) = \Psi(\cI;y)$.
\end{proof}

\begin{lemma}
\label{lem:app-boundary}
For every $\ell$ with $4T + 2 \le \ell \le \lmax$ there is a function $H_\ell$ such that $P(a) = H_\ell(\mathrm{head}(a), \mathrm{tail}(a))$ for every arrangement $a$ of length $\ell$.
\end{lemma}

\begin{proof}
First, changing a single \emph{interior} entry $a_i$ ($2T + 1 \le i \le \ell - 2T$) to any fresh value $z \notin a$ leaves $P(a)$ unchanged: the affected terms of $P$ are exactly the windows centered at positions in $[i - T, i + T] \subseteq [T+1, \ell - T]$, and their sum is $\Psi(\cI; a_i) = \Psi(\cI; z)$ by Lemma~\ref{lem:app-exchange}, where $\cI$ is the context formed by the $2T$ entries of $a$ on each side of position $i$.

Now let $a, a'$ be arrangements of length $\ell$ with equal heads and tails; we transform $a$ into $a'$ by single interior replacements. Maintain a current arrangement $b$, initially $a$, always agreeing with $a'$ on head and tail. While $b \ne a'$, pick an interior position $i$ with $b_i \ne a'_i$. If $a'_i$ does not occur in $b$, replace $b_i \leftarrow a'_i$. Otherwise $a'_i$ occurs in $b$ at some position $j \ne i$, necessarily interior (the head and tail of $b$ equal those of $a'$, which avoid the interior value $a'_i$ of $a'$); first replace $b_j \leftarrow z$ for some $z \notin b \cup a'$, and then $b_i \leftarrow a'_i$. Such $z$ exists because $b$ and $a'$ share their $4T$ head and tail entries, so $|b \cup a'| \le 2\ell - 4T \le 2 \lmax - 4T \le N - 1$ by (N2). Replacements keep all entries distinct and preserve $P$ by the previous paragraph, and each iteration strictly increases the number of agreeing positions (position $i$ becomes correct; position $j$ disagreed before and still disagrees). The process terminates with $b = a'$, whence $P(a) = P(a')$: the inner sum depends only on the head and tail.
\end{proof}

\begin{lemma}
\label{lem:app-loops}
Fix a $2T$-tuple $A$.
\\ (a) For every step $\cT \to \cT'$ via window $\hat{\cT}$ with entries disjoint from $A$, and every $\ell$ with $4T + 2 \le \ell \le \lmax - 1$:
$H_{\ell+1}(A, \cT') = H_\ell(A, \cT) + F(\hat{\cT})$.
\\ (b) For every loop $\cL$ of length $\mu \le \mumax$ from basepoint $\cT_0$ with all window entries disjoint from $A$:
$\omega(\cL) = H_{4T+2+\mu}(A, \cT_0) - H_{4T+2}(A, \cT_0)$.
\\ (c) Any two loops of the same length $\mu \in [2T+1, \mumax]$ from the same basepoint $\cT_0$ have the same weight, denoted $\Lambda_{\cT_0}(\mu)$; it is an integer in $[0, \mu]$; and $\Lambda_{\cT_0}(\mu_1 + \mu_2) = \Lambda_{\cT_0}(\mu_1) + \Lambda_{\cT_0}(\mu_2)$ whenever $\mu_1, \mu_2 \ge 2T+1$ and $\mu_1 + \mu_2 \le \mumax$.
\end{lemma}

\begin{proof}
(a) Take any arrangement $a'$ of length $\ell + 1$ whose first $2T$ entries are $A$ and whose last $2T + 1$ entries are $\hat{\cT}$, with fresh distinct interior entries avoiding $A \cup \hat{\cT}$. Such an $a'$ exists: it uses $\ell + 1$ distinct identifiers in total, and $\ell + 1 \le \lmax \le n \le N - 1$ by (N2). Its head and tail are disjoint as sets ($A$ is disjoint from $\hat{\cT} \supseteq \cT \cup \cT'$ by hypothesis, and $\ell + 1 \ge 4T + 2$ keeps head and tail non-overlapping as positions). By definition of the inner sum, $P(a')$ equals the inner sum of the length-$\ell$ prefix of $a'$ plus the single window value $F(\hat{\cT})$, of the window centered at position $\ell + 1 - T$. Applying Lemma~\ref{lem:app-boundary} to $a'$ (level $\ell+1 \le \lmax$) and to its prefix (level $\ell$, head $A$, tail $\cT$) turns this identity into the displayed one.

(b) Apply (a) along the loop starting from level $\ell^* = 4T + 2$: $H_{\ell^* + k + 1}(A, \cT_{k+1}) = H_{\ell^* + k}(A, \cT_k) + F(\hat{\cT}_k)$ for $k = 0, \dots, \mu - 1$; all levels stay at most $\ell^* + \mu \le M + \mumax \le \lmax$ by (N1), so each application is legal. Telescoping gives the formula.

(c) Let $\cL, \cL'$ be loops of length $\mu \le \mumax$ from $\cT_0$. The identifiers appearing in $\cL$ are those of $\cT_0$ plus at most one new identifier per step, hence at most $2T + \mu$ in total; likewise for $\cL'$, and the two share $\cT_0$, so together they use at most $2T + 2\mu$ distinct identifiers. Since $(2T + 2\mu) + 2T \le 2\mumax + 4T \le N - 7$ by (N2), there is a $2T$-tuple $A$ disjoint from the identifiers of \emph{both} loops. Part (b) with this common $A$ gives $\omega(\cL) = H_{4T+2+\mu}(A, \cT_0) - H_{4T+2}(A, \cT_0) = \omega(\cL')$. The weight is a sum of $\mu$ values of $F \in \{0,1\}$, hence an integer in $[0, \mu]$. For additivity, take loops $\cL_1, \cL_2$ from $\cT_0$ of lengths $\mu_1, \mu_2$ (they exist: $\mu_i \ge 2T+1$ and $\mu_i + 2T \le \mumax + 2T \le N$); their concatenation $\cL_1 \cL_2$ is a loop from $\cT_0$ of length $\mu_1 + \mu_2 \le \mumax$ and weight $\omega(\cL_1) + \omega(\cL_2)$, so well-definedness at length $\mu_1 + \mu_2$ gives the additivity claim.
\end{proof}

\begin{lemma}
\label{lem:app-linear}
There is an integer $\lambda$ such that $\Lambda_{\cT}(\mu) = \lambda\, \mu$ for every $2T$-tuple $\cT$ and every $\mu \in [2T+1,\, \mumax - M]$.
\end{lemma}

\begin{proof}
Fix a basepoint $\cT$. For every $\mu \in [2T+1, \mumax - M - 1]$, additivity (Lemma~\ref{lem:app-loops}(c); all lengths involved are $\ge 2T+1$ and all sums are $\le \mumax$) gives
\[ \Lambda_{\cT}(\mu + 1) + \Lambda_{\cT}(M) \;=\; \Lambda_{\cT}(\mu + 1 + M) \;=\; \Lambda_{\cT}(\mu) + \Lambda_{\cT}(M + 1), \]
so the difference $\Lambda_{\cT}(\mu+1) - \Lambda_{\cT}(\mu) = \Lambda_{\cT}(M+1) - \Lambda_{\cT}(M) =: \lambda(\cT)$ is constant for $\mu \in [2T+1, \mumax - M - 1]$, and it is an integer, being a difference of integers. Hence $\Lambda_{\cT}$ is affine on $[2T+1, \mumax - M]$: $\Lambda_{\cT}(\mu) = \Lambda_{\cT}(M) + (\mu - M)\lambda(\cT)$ (note that $M = 4T + 2$ lies in this range by (N3)).

The intercept vanishes: additivity gives $\Lambda_{\cT}(3M) = 3 \Lambda_{\cT}(M)$ (legal, as $2M \le \mumax$ and $3M \le \mumax$ by (N3)), while affinity at $3M$ (legal, as $3M \le \mumax - M$ by (N3)) gives $\Lambda_{\cT}(3M) = \Lambda_{\cT}(M) + 2M \lambda(\cT)$. Comparing, $\Lambda_{\cT}(M) = M \lambda(\cT)$, and therefore $\Lambda_{\cT}(\mu) = \lambda(\cT)\,\mu$ on $[2T+1, \mumax - M]$.

The slope is basepoint-free. First let $\cT, \cT'$ be $2T$-tuples with disjoint entries, let $\alpha = \ShiftIn(\cT, \cT')$ and $\beta = \ShiftIn(\cT', \cT)$, and let $L$ be a loop of length $M$ from $\cT'$ whose identifiers avoid $\cT$: arrange $\cT'$ followed by $M - 2T$ fresh identifiers avoiding $\cT \cup \cT'$ in a cycle (this needs $N \ge 4T + (M - 2T) = 6T + 2$, which is clear). Then $\alpha L \beta$ and $\alpha \beta$ are loops from $\cT$ of lengths $M + 4T$ and $4T$, and
\[ \Lambda_{\cT'}(M) \;=\; \omega(L) \;=\; \omega(\alpha L \beta) - \omega(\alpha\beta) \;=\; \Lambda_{\cT}(M + 4T) - \Lambda_{\cT}(4T) \;=\; \lambda(\cT)\, M , \]
using $2T + 1 \le 4T$ and $M + 4T \le \mumax - M$ (by (N5)), so that both evaluations lie in the affine range of $\cT$. Since also $\Lambda_{\cT'}(M) = \lambda(\cT') M$ and $M \ge 1$, we get $\lambda(\cT') = \lambda(\cT)$. Finally, arbitrary $2T$-tuples $\cT, \cT''$ are linked through a third tuple disjoint from both, which exists since $|\cT \cup \cT''| + 2T \le 6T \le N$. Hence $\lambda(\cT) =: \lambda$ is one integer for all basepoints.
\end{proof}

\subsection{Proof of Theorem~\ref{thm:app-sharp}: the three-way split}
\label{app:universe-proof}

\begin{proof}[Proof of Theorem~\ref{thm:app-sharp}]
Assume the standing hypothesis, towards a contradiction. Take any assignment $\ID$ of $\cC_n$ (one exists as $N \ge n$) and let $\cT^*_0, \cT^*_1, \dots, \cT^*_{n-1}$ be the cyclic sequence of its boundary $2T$-tuples, $\cT^*_k = (\ID(k-T), \dots, \ID(k+T-1))$, so that the $n$ windows of $\ID$ form a loop of length $n$ from $\cT^*_0$ of weight $f(\ID) = K$. Set $j_1 = \lfloor n/3 \rfloor$ and $j_2 = \lfloor 2n/3 \rfloor$, and split the window loop at $0, j_1, j_2$ into the three paths
\[ p_1 : \cT^*_0 \to \cT^*_{j_1}, \qquad p_2 : \cT^*_{j_1} \to \cT^*_{j_2}, \qquad p_3 : \cT^*_{j_2} \to \cT^*_0 \]
of lengths $j_1$, $j_2 - j_1$, $n - j_2$, each lying in $[\lfloor n/3 \rfloor - 1,\, \lceil n/3 \rceil + 1]$. The tuples $\cT^*_0, \cT^*_{j_1}, \cT^*_{j_2}$ occupy the position sets $[-T, T-1]$, $[j_1 - T, j_1 + T - 1]$, $[j_2 - T, j_2 + T - 1]$ (mod $n$), which are pairwise disjoint because all three circular gaps $j_1$, $j_2 - j_1$, $n - j_2$ are at least $\lfloor n/3 \rfloor - 1 \ge 2T$ by (N5); hence the three tuples are pairwise disjoint as identifier sets. Let
\[ q_1 = \ShiftIn(\cT^*_{j_1}, \cT^*_0), \qquad q_2 = \ShiftIn(\cT^*_{j_2}, \cT^*_{j_1}), \qquad q_3 = \ShiftIn(\cT^*_0, \cT^*_{j_2}) \]
(legal by the pairwise disjointness). Consider the four loops
\[ \cL_1 = p_1 q_1, \quad \cL_2 = p_2 q_2, \quad \cL_3 = p_3 q_3, \quad \cL_0 = q_3 q_2 q_1 , \]
from $\cT^*_0$, $\cT^*_{j_1}$, $\cT^*_{j_2}$, $\cT^*_0$, of lengths $j_1 + 2T$, $(j_2 - j_1) + 2T$, $(n - j_2) + 2T$, $6T$, respectively. Summing weights and using $\omega(p_1) + \omega(p_2) + \omega(p_3) = f(\ID) = K$:
\[ \omega(\cL_1) + \omega(\cL_2) + \omega(\cL_3) \;=\; K + \omega(q_1) + \omega(q_2) + \omega(q_3) \;=\; K + \omega(\cL_0). \]
All four loop lengths lie in $[2T+1,\, \mumax - M]$. Lower bounds: $6T \ge 2T + 1$ by (N5), and each of $\cL_1, \cL_2, \cL_3$ has length $(\mbox{piece length}) + 2T \ge 1 + 2T$. Upper bounds: $6T \le M + 4T \le \mumax - M$ by (N5), and each of $\cL_1, \cL_2, \cL_3$ has length at most $\lceil n/3 \rceil + 1 + 2T \le \mumax - M$ by (N4). Hence Lemma~\ref{lem:app-linear} evaluates every term:
\[ \lambda\,(j_1 + 2T) + \lambda\,\big((j_2 - j_1) + 2T\big) + \lambda\,\big((n - j_2) + 2T\big) \;=\; K + \lambda \cdot 6T , \]
that is, $\lambda\,(n + 6T) = K + 6T\lambda$, i.e.\ $\lambda n = K$. Since $\lambda$ is an integer and $0 \le K \le n$, this forces $K \in \{0, n\}$, contradicting the standing assumption $K \notin \{0, n\}$. The theorem follows.
\end{proof}

Theorem~\ref{thm:universe} now follows: part (a) is Proposition~\ref{prop:app-modthree}, and part (b) is Theorem~\ref{thm:app-sharp}. (By Proposition~\ref{prop:app-zeroround}, part (b) in fact extends to $T = 0$.)

\bigskip
We conclude with some remarks on sharpness.
\\
(a) \emph{The locality restriction is necessary.} For $T \ge \lfloor n/2 \rfloor$ (so that $2T + 1 \ge n$) the radius-$T$ window of every vertex contains the entire assignment, and the rank rule ``output $(\mbox{rank of } \ID(v) \mbox{ in the sorted order of } \ID(V)) \bmod 3$'' is surely exactly balanced for \emph{every} universe size $N$. So the threshold statement is genuinely about the local regime $T = O(n)$; only the constant ($50$ here, $40$ in Theorem~\ref{thm:det-exact} for $N \ge 2n+2$) is unoptimized.
\\
(b) \emph{Integer-valued rules.} As in Corollary~\ref{cor:integer-rules}, Theorem~\ref{thm:app-sharp} holds as is for rules $F$ with values in $\{0, 1, \dots, \vartheta\}$: the range of $F$ enters only through ``loop weights are integers'' (in Lemma~\ref{lem:app-loops}(c), where $[0, \mu]$ becomes $[0, \vartheta\mu]$), and the final identity $\lambda n = K$ with $\lambda \in \mathbb{Z}$ and $0 \le K \le \vartheta n$ yields $K \in \{0, n, 2n, \dots, \vartheta n\}$. Likewise, at $N = n$ every $K \in \{0, 1, \dots, \vartheta n\}$ is a sure count of a $0$-round $\{0, \dots, \vartheta\}$-valued rule: pick any $h : [n] \to \{0, \dots, \vartheta\}$ with $\sum_{x \in [n]} h(x) = K$; since every assignment is a bijection onto $[n]$, its count is $\sum_{x} h(x) = K$ surely.
\\
(c) \emph{What breaks at $N = n$.} In the language of the proof, at $N = n$ the exchange lemma is vacuous: with the $4T + 2$ identifiers of $\cI \cup \{x, y\}$ excluded, only $n - 4T - 2$ identifiers remain for the other $n - 4T - 1$ positions, so no assignment containing $\cI[x]$ while avoiding $y$ exists, and no two assignments differ in exactly one identifier (any two bijections onto $[n]$ differ in at least two positions). The very first lemma of the rigidity machinery fails, and Proposition~\ref{prop:app-modthree} shows that this failure is real, not an artifact of the proof.
\\
(d) \emph{Relation to the constants of Theorem~\ref{thm:det-exact}.} Theorem~\ref{thm:det-exact} is proved for $N \ge 2n + 2$ and $T \le n/40 - 1$; Theorem~\ref{thm:app-sharp} trades the constant ($T \le n/50 - 1$) for the optimal universe bound $N \ge n + 1$. In the narrow band $n/50 - 1 < T \le n/40 - 1$ only the case $N \ge 2n+2$ is covered (by Theorem~\ref{thm:det-exact}); this affects no threshold statement, since for each fixed $T \le n/50 - 1$ the dichotomy ``a sure count $K \notin \{0,n\}$ is possible if and only if $N = n$'' is complete.

\clearpage

\end{document}